\documentclass[a4paper,UKenglish,cleveref, autoref, thm-restate]{lipics-v2021}

\usepackage[utf8]{inputenc}

\usepackage{amsmath,amssymb,amsthm}
\usepackage{enumitem}

\usepackage{tikz}
\usetikzlibrary{decorations.pathreplacing}
\usetikzlibrary{positioning,shapes,fit,calc}
\usepackage{graphicx}
\usepackage{subcaption} 
\usepackage[table]{colortbl}
\usepackage{nicefrac}
\usepackage{todonotes}
\usepackage{mathtools} 
\usepackage{multirow}
\usepackage{xcolor}
\usepackage{mathrsfs}

\usepackage[autostyle]{csquotes}
\MakeOuterQuote{"}

\usepackage[linesnumbered, 
algoruled, 
noend] 
{algorithm2e}

\usepackage{xcolor}  
\usepackage{hyperref}

\newenvironment{pfof}[1]{\begin{proof}[\emph{\textbf{Proof of #1: }}]}{\end{proof}}
\newcommand{\uproman}[1]{\uppercase\expandafter{\romannumeral#1}}

\newenvironment{sketch}{%
	\renewcommand{\proofname}{Proof Sketch}\proof}{\endproof}

\def\algCrown{\texttt{FindBCD}}
\def\algDivide{\texttt{Divide}}
\def\algCut{\texttt{Cut}}
\def\algCutClean{\texttt{CutCleanup}}
\def\algDorC{\texttt{DivideOrCutVertex}}
\def\algmaxmin{\texttt{MaxMinApx}}
\def\algminmax{\texttt{MinMaxApx}}
\def\algkernelsep{\texttt{WSepKernel}}
\def\preStep5{\texttt{PreStep5}}
\def\algStep5{\texttt{AlgStep5}}

\def\algkernelpack{\texttt{WPackKernel}}

\def\maxminbcp{\textsc{Max-Min BCP}}
\def\minmaxbcp{\textsc{Min-Max BCP}}
\def\wsep{\textsc{kWSep}}
\def\wpack{\textsc{kWPack}}

\def\comp{\mathbb{CC}}
\def\bigO{\mathcal{O}}
\def\N{\mathbb{N}}
\def\Nz{\mathbb{N}_0}

\def\calS{\mathcal{S}}

\def\calR{\mathfrak{R}}
\def\calC{\mathfrak{C}}
\def\calN{\mathcal{N}}
\def\calU{\mathcal{U}}

\def\calV{\mathcal{V}}
\def\nef{N^{\bar{g}}}
\def\nu{\calN^{\bar{g}}}

\def\ta{\widetilde{A}}
\def\tb{\widetilde{B}}
\def\ccd{$\lambda$-BCD}
\def\wccd{$W$-BCD}
\def\cod{$\lambda$-CoD}
\def\pbcod{PB-$\lambda$-CoD}
\def\fbcod{FB-$\lambda$-CoD}

\def\wgg{w^{g+g'}}

\def\wggone{w^{g_1+g_1'}}

\def\ncalC{\calN_{\calC}}

\def\chrf{C,H,\calR,f}
\def\chrg{\calC,H,\calR,g}
\def\chrgone{\calC_1,H_1,\calR_1,g_1}
\def\chrfstar{C^*,H^*,\calR^*,f^*}
\def\Cold{\calC_0}
\def\gold{g_0}
\def\gpold{g'_0}
\def\goldinv{\gold^{-1}}
\def\chrgg{\calC,H,\calR,\allowbreak g,g'}

\def\chrggone{\calC_1,H_1, \allowbreak \calR_1,\allowbreak g_1,g_1'}

\def\chrggthr{\calC_3,H_3,\calR_3,g_3,g_3'}
\def\chrgthr{\calC_3,H_3,\calR_3,g_3}
\def\oldchrgg{\Cold,H,\calR,\gold,\gpold}

\def\compcond{component condition}
\def\gnbrcond{$g$-neighbor condition}
\def\Rcond{$\calR$-condition}
\def\gdnbrcond{$g'$-neighbor condition}
\def\gwcond{$g$-weight condition}
\def\ggwcond{$(g+g')$-weight condition}
\def\splitcond{$0.5(\lambda-1)$ condition}
\def\hatc{\widehat{C}}
\def\hatE{\widehat{E}}
\def\hath{\widehat{H}}
\def\hatQ{\widehat{\calQ}}
\def\hatq{\hat{Q}}
\def\hatq{\hat{Q}}

\def\divR{\hat{R}}
\def\calQ{\mathcal{Q}}
\def\conQ{\mathscr{Q}}

\def\cCpriv{\calC^{\mathtt{priv}}}

\def\qpriv{\calQ^{\mathtt{priv}}}

\def\finv{f^{-1}}
\def\ginv{g^{-1}}

\def\ginvone{g_1^{-1}}
\def\ginvtwo{g_2^{-1}}
\def\g2inv{g_2^{-1}}

\def\tC{\widetilde{C}}

\def\hold{H_{0}}
\def\Rold{\calR_{0}}
\def\hcrown{H_1}
\def\Qcrown{\calQ_1}

\def\hrem{H_2}

\def\T{\mathcal{T}}
\def\wmax{w_{\max}}

\def\cvp{\textsc{CVP}}

\def\wmin{W_{\min}}
\def\hr{\widehat{R}}

\def\hk{\widehat{k}}

\newcommand{\dir}[1]{\overrightarrow{#1}}

\newcommand\pto{\mathrel{\ooalign{\hfil$\mapstochar\mkern5mu$\hfil\cr$\to$\cr}}}

\def\fexpanse{{fractional balanced expansion}}
\def\expanse{{balanced expansion}}
\def\Nw{H}

\title{Balanced Crown Decomposition for Connectivity Constraints} 

\titlerunning{Balanced Crown Decomposition for Connectivity Constraints} 

\author{Katrin Casel}{University Potsdam, Potsdam \and Hasso-Plattner-Institut, Potsdam}{Katrin.Casel@hpi.de}{https://orcid.org/0000-0001-6146-8684}{}
\author{Tobias Friedrich}{University Potsdam, Potsdam \and Hasso-Plattner-Institut, Potsdam}{Tobias.Friedrich@hpi.de}{{https://orcid.org/0000-0003-0076-6308}}{}
\author{Davis Issac}{University Potsdam, Potsdam \and Hasso-Plattner-Institut, Potsdam}{Davis.Issac@hpi.de}{}{}
\author{Aikaterini Niklanovits}{University Potsdam, Potsdam \and Hasso-Plattner-Institut, Potsdam}{Aikaterini.Niklanovits@hpi.de}{}{}
\author{Ziena Zeif}{University Potsdam, Potsdam \and Hasso-Plattner-Institut, Potsdam}{Ziena.Zeif@hpi.de}{}{}

\authorrunning{K. Casel, T. Friedrich, D. Issac A. Niklanovits and Z. Zeif} 

\Copyright{Katrin Casel, Tobias Friedrich, Davis Issac, Aikaterini Niklanovits and Ziena Zeif} 



\category{} 

\relatedversion{} 

\nolinenumbers 

\hideLIPIcs  

\EventEditors{John Q. Open and Joan R. Access}
\EventNoEds{2}
\EventLongTitle{42nd Conference on Very Important Topics (CVIT 2016)}
\EventShortTitle{CVIT 2016}
\EventAcronym{CVIT}
\EventYear{2016}
\EventDate{December 24--27, 2016}
\EventLocation{Little Whinging, United Kingdom}
\EventLogo{}
\SeriesVolume{42}
\ArticleNo{23}

\begin{document}

\maketitle

	\begin{abstract}
		We introduce the \emph{balanced crown decomposition} that captures the structure imposed on graphs by their connected induced subgraphs of a given size.
		Such subgraphs are a popular modeling tool in various application areas, where the non-local nature of the connectivity condition usually results in very challenging algorithmic tasks. The balanced crown decomposition is a combination of a crown decomposition and a balanced partition which makes it applicable to graph editing as well as graph packing and partitioning problems. We illustrate this by deriving improved approximation algorithms and kernelization for a variety of such problems.
		
		In particular, through this structure, we obtain the first constant-factor approximation for the \textsc{Balanced Connected Partition (BCP)} problem, where the task is to partition a vertex-weighted graph into $k$ connected components of approximately equal weight. We derive a 3-approximation for the two most commonly used objectives of maximizing the weight of the lightest component or minimizing the weight of the heaviest component.
	\end{abstract}

	\section{Introduction}
		Connected subgraphs are one of the most natural structures to encode aspects of a practical task, modeled as a graph problem. 
On the one side, such subgraphs represent structures we seek to discover, such as territories for postal delivery and similar districting problems (see e.g.\ the survey \cite{Kalcsics2019}). From another perspective, the structures of interest could be operations that scatter a graph into small connected components; a structure e.g. used to model vulnerability in network security (see e.g.~the survey~\cite{vulnerability}). Partitioning a graph into connected components of a given size is also used as a model for task allocation to robots~\cite{zhou2019balanced}. From an algorithmic perspective, connectivity is a non-local requirement which makes it particularly challenging. We introduce a graph structure that can be used to design efficient algorithms for a broad class of problems involving connected subgraphs of a given size.

One useful structure to derive information about connected subgraphs is a \emph{connected partition}, defined as follows. Given a graph $G=(V,E)$, a connected partition of $G$ is a partition $V_1,\dots,V_k$ of~$V$ such that the graph induced by the vertex set $V_i$ is connected for each $1\leq i\leq k$, where $k\in \mathbb N$ is the size of the partition. Often, we are not interested in just any connected partition but in those that have the additional property of being \emph{balanced}. Informally speaking, a connected partition is considered balanced  if the sets $V_i$ have approximately equal cardinality. There are several measures to assess the quality of a \emph{balanced connected partition} (BCP for short), with the two most commonly used objectives being to maximize $\min_{1\leq i\leq k}|V_i|$ or to minimize  $\max_{1\leq i\leq k}|V_i|$, as first introduced for trees in \cite{perl1981max} and~\cite{kundu1977linear}, respectively. 
Despite extensive studies on these problems for the past 40 years, see e.g.~\cite{chataigner2007approximation,perl1981max,suzuki1990linear,wada1993efficient,wang2013max,wu20107,wu2012fully} the best known approximation ratio for the Min-Max objective depends on the number of sets $k$~\cite{chen2019approximation}. For the Max-Min objective not even such a result is known; only for the cases with $k$ restricted to~2 or~3, there exist approximations with ratio  $\frac{4}{3}$~\cite{chlebikova1996approximating} and $\frac{5}{3}$~\cite{chen2020approximation}, respectively. Deriving an approximation with a ratio independent of~$k$ seems to require a new strategy.

Helpful structures for both of the objectives Max-Min and Min-Max, as a sort of compromise, are BCP's such that $\lambda_l\leq |V_i|\leq \lambda_r$ for some fixed bounds $\lambda_l,\lambda_r$. We call this compromise structure $[\lambda_l, \lambda_r]$-{\cvp} (connected \emph{vertex} partition),
and it is one ingredient of \emph{balanced crown decomposition},
the main structural object that we present.
In the case that no $[\lambda_l, \lambda_r]$-{\cvp} exists for a graph, we can learn something about its structure. 
In particular, our \emph{balanced crown decomposition theorem} (\cref{theorem:lccd}) shows that for any $k,\lambda>0$, the non-existence of a $[\lambda, 3\lambda-3]$-{\cvp} implies the existence of a vertex set $H\subseteq V$ of cardinality at most $k$ that disconnects at least one component of size less than $\lambda$ from $G$.  
Such sets $H$ are in a sense the dual of balanced connected partitions.

Small subsets of vertices that disconnect a graph are usually called \emph{vertex separators}, and they are one of the most powerful tools for designing efficient graph algorithms. In a sense, they are the base requirement of successful divide-and-conquer strategies for graph problems. This generality and their wide applicability has made the study of separators a rich and active research field, see e.g.~the book by  Rosenberg and Heath~\cite{separator_survey}, or the line of research initialized by the seminal paper of Lipton and Tarjan~\cite{LiptonTarjan} on separators in planar graphs. Numerous different types of separator structures emerged over the past couple of decades. In the context of connectivity problems, the separator structures of particular interest are crown decompositions; a classical tool to derive kernelizations in the field of parameterized complexity. We refer to chapter~4 of the book on kernelization~\cite{fomin2019kernelization} for more details on crown decompositions and their applications.

\emph{Crown decomposition} was introduced as a generalization of Hall's Theorem in~\cite{ChorFJ04}. More precisely, a \emph{crown decomposition} of a graph  $G=(V,E)$ is a partition of $V$ into three sets $H$ (\emph{head}), $C$ (\emph{crown}) and $R$ (\emph{royal body}), such that $H$ separates $C$ from $R$, $C$ is an independent set in $G$, and there exists a matching of size $|H|$ among the edges $E\cap(H\times C)$. Notice that the set $H$ is the separator set, and the property of $C$ being an independent set can be seen as $H$ splitting connected components of size~1 from the graph. The condition of the matching from $H$ into $C$ models a trade-off between the size of the separator and the amount of small sets that are separated. Different versions of crown decompositions have been introduced in the literature, adjusting the structure to specific application scenarios.

The structure of particular interest to us is the \emph{$q$-Weighted Crown Decomposition} introduced by Xiao~\cite{xiao2017linear}.
Here, the crown $C$ is no longer an independent set, but has the restriction that each connected component in it has size at most~$q$ (generalizing the notion of  independent set for $q=1$);  and there exist an assignment of connected components of the crown to the head such that each vertex in the head is assigned at least $q$ vertices.
This assignment generalizes the notion of matching in the original crown decomposition.
Such a weighted crown decomposition can be derived using
the \emph{Expansion Lemma} as stated in~\cite[Chapter 5.3]{FominLMPS16}, and its generalization to the \emph{Weighted Expansion Lemma} as given in slightly different forms in~\cite{kumar20162} and~\cite{xiao2017linear}. 
The \emph{expansions} derived by these lemmas can be thought of as bipartite analogues of the crown decomposition. Formally, given a bipartite graph $G=(A,B,E)$, a $q$-(Weighted) Expansion is given by sets $H\subseteq A$ and $C\subseteq B$ such that the neighborhood of $C$ is contained in $H$ and an assignment $f\colon C\to H$ such that the number (resp.~weight, in the vertex-weighted case) of vertices assigned to each vertex in $H$ is at least~$q$ (resp. $q-W+1$ where $W$ is the largest weight).
Both Kumar and Lokshtanov~\cite{kumar20162} and Xiao~\cite{xiao2017linear} use their respective {Weighted Expansion Lemma} to derive kernels for \textsc{Component Order Connectivity}, a version of the editing problem that we also consider in a more general form under the name \textsc{$W$-Weight Separator}. 

To create our new structure \emph{balanced crown decomposition}, we combine balanced connected partitions and crown decompositions to derive a tool that has the advantages of both of the individual structures. 
Essentially, it is a weighted crown decomposition with the additional property that the body has a balanced connected partition.
Also, note that we allow a more generalized version of weighted crown decomposition than Xiao~\cite{xiao2017linear}, by considering weighted vertices. 
Formally, we consider \emph{vertex-weighted} graphs $G=(V,E,w)$ with integer-weights $w\colon V\rightarrow \mathbb N$. For simplicity we use $w(V')=\sum_{v\in V'} w(v)$ for the \emph{weight of a subset} $V'\subseteq V$.

We show that  balanced crown decompositions have applications for various kinds of problems involving connectivity constraints. Specifically we discuss for the three types editing, packing and partitioning the following problems on input $G=(V,E,w)$ and $k,W\in \mathbb N$:
\begin{description}
	\item \textsc{Max-Min (Min-Max) BCP}: Decide if there exists a connected partition $V_1,\dots, V_k$ of $V$ such that $w(V_i)\geq W$ (resp. $w(V_i)\leq W$) for each $i\in [k]$; usually stated as optimization problem to maximize/minimize~$W$. 
 \item \textsc{$W$-weight Separator}: Decide if there exists a set $S\subseteq V$ with $|S|\leq k$ whose removal from $G$ yields a graph with each connected component having weight less than~$W$.
 \item \textsc{$W$-weight Packing}: Decide if  there exist $k$ pairwise disjoint sets $V_1,\dots, V_k\subseteq V$ with  $w(V_i)\geq W$, such that  the graph induced by   $V_i$ is connected, for each $i\in [k]$.
\end{description}

We remark that the problems \textsc{$W$-weight Separator} and \textsc{$W$-weight Packing} have been studied mostly on the unweighted versions, also known as \textsc{Component Order Connectivity} and \textsc{$T_r$-Packing}, respectively.

For all results of this paper, we consider the RAM model of computation with word size $O(\log(|V|+\max_{v\in V}w(v)))$.
All our algorithms are polynomial w.r.t.~the encoding of input. 

\subsection{Our Contribution}
Our major contributions can be summarized as follows: 
\begin{enumerate}
	\item	\textbf{Balanced Crown Decomposition} (BCD): The main contribution of our paper is a new crown decomposition tailored for problems with connectivity constraints. Our novel addition over previous crown decompositions is that we also give a partition of the \emph{body} into connected parts of roughly similar size.
		More precisely, we divide the graph into $C,H,$ and $R$ such that $C,H,R$ is a weighted-crown decomposition and also a $[\lambda,3\lambda]$-CVP is given for $R$.
		\cref{def:lccd} gives the formal definition of BCD, and \cref{theorem:lccd} gives our main result about computing BCD.
		We believe that apart from the applications used in this paper, BCD will find applications for other problems with connectivity constraints.
	\item \textbf{Balanced Expansion}: We also give a novel variation of the expansion lemma,
		which is an important constituent of our algorithm for BCD. 
		Given a bipartite graph $G=(A,B,E)$, we give an expansion with $H\subseteq A, C\subseteq B$, with the addition that the expansion $f$ while being a weighted $q$-expansion from $C$ to $H$, now also maps $B\setminus C$ to neighbors in $A\setminus H $ such that only a bounded  weight is assigned to each vertex in $A\setminus H$. See Definition~\ref{definition::expanse} for a more formal definition and \cref{lemma::cycleCanceling} for our result on computing balanced expansion.
		Apart from its usage here to compute BCD, the balanced expansion could be of independent interest, given the significance of the Expansion Lemma in parameterized complexity.
		
	\item \textbf{Approximation algorithms for BCP}: Using  BCD, we give $3$-approximation algorithms for both \textsc{Max-Min} and \textsc{Min-Max BCP}. 	These are the first constant approximations for both problems in polynomial time for a general $k$. Recall that despite numerous efforts in the past 40 years,
	only a $k/2$-approximation for \textsc{Min-Max BCP}~\cite{chen2019approximation}, and constant-factor approximations for the particular cases of $k=2,3$ for \textsc{Max-Min BCP}~\cite{chlebikova1996approximating,chen2020approximation} were known.
	\item \textbf{Improved Kernels for W-weight separator and packing}: BCD directly gives a $3kW$-kernel for both of the problems improving over the previous best polynomial time kernels of size $9kW$~\cite{xiao2017linear} and $\bigO(kW^3)$~\cite{chen2012kernels}. 
		Especially, we get the same improvements for the unweighted versions \textsc{Component Order Connectivity} and \textsc{$T_r$-Packing}. 
	\item\textbf{Faster algorithms for Expansion}: Our algorithm for Balanced Expansion, also gives an alternative flow-based method for computing the standard (weighted) expansion. Our algorithm can compute a (weighted) expansion in $\bigO(|V||E|)$ surpassing the previous best runtimes of $\bigO(|V|^{1.5}|E|)$ and $\bigO(|E||V|^{1.5}W^{2.5}$)~\cite[Chapter 5.3]{fomin2019kernelization} (here $W$ is the largest weight) for unweighted and weighted expansion, respectively. In particular, for weighted expansion, our runtime does not depend on the weights and is the first algorithm that runs in time polynomial w.r.t.~the length of the input-encoding. 
		The improvement in runtime may turn out to be useful to speed up kernelizations for other problems.
\end{enumerate}

\subsection{Related work}




Both variants of \textsc{BCP} were first introduced for trees, where \textsc{Max-Min BCP} and  \textsc{Min-Max BCP} are introduced in~\cite{perl1981max} and~\cite{kundu1977linear}, respectively. For this restriction to trees, a linear time algorithm was provided for both variants in~\cite{frederickson1991optimal}. 
For both variants of \textsc{BCP}, a $\Delta_T$-approximation  is given in~\cite{borndorfer2019approximatin} where $\Delta_T$ is the maximal degree of an arbitrary spanning tree of the input graph; for \textsc{Max-Min BCP} the result holds only when the input is restricted to weights with $\max_{v\in V}w(v)\leq \frac{w\left(G\right)}{\Delta k}$. Also, although not explicitly stated, a $\left(1 + ln\left(k\right)\right)$-approximation in $\mathcal{O}\left(n^k\right)$ time for $\textsc{Min-Max BCP}_k$ follows from  the results in~\cite{chen2007almost}.
 With respect to lower bounds, it is known that there exists no approximation  for $\textsc{Max-Min BCP}$ with a ratio below~$6/5$, unless $\textsf{P}\neq \textsf{NP}$~\cite{chataigner2007approximation}.
For the unweighted case, i.e.~$w\equiv 1$, the best known result for \textsc{Min-Max BCP} is the $\frac{k}{2}$-approximation for every $k \geq 3$  given in~\cite{chen2019approximation}.

Balanced connected partitions are also studied for particular cases of $k$, denoted $\textsc{BCP}_k$. The restriction $\textsc{BCP}_2$, i.e.~balanced connected bipartition, is already be  \textbf{NP}-hard~\cite{camerini1983complexity}. On the positive side, a 
$\frac{4}{3}$-approximation for $\textsc{Max-Min BCP}_2$ is given in~\cite{chlebikova1996approximating}, and in~\cite{chen2019approximation} this result is used to derive a $\frac{5}{4}$-approximation for $\textsc{Min-Max BCP}_2$.
For tripartition, approximations for $\textsc{Max-Min BCP}_3$ and $\textsc{Min-Max BCP}_3$ with ratios $\frac{5}{3}$ and $\frac{3}{2}$, respectively, are given in~\cite{chen2020approximation}.
 	
$\textsc{BCP}$ in unit-weighted $k$-connected graphs can be seen as a special case of the Gy\H{o}ri-Lov{\'{a}}sz Theorem (independently given by  Gy\H{o}ri~\cite{gyori1976division} and Lov{\'{a}}sz~\cite{lovasz1977homology}). It states that for any $k$-connected graph $G=(V,E)$ and integers $n_1,\dots, n_k$ with $n_1 + \dots + n_k = |V|$,  there exists a connected partition $V_1,\dots,V_k$ of $V$ with $|V_i|=n_i$ for all $i\in [k]$. Moreover, it is possible to fix vertices $v_1, \dots, v_k$ and request $v_i\in V_i$  for all $i\in [k]$. The Gy\H{o}ri-Lov{\'{a}}sz Theorem is extended to weighted directed graphs in~\cite{chen2007almost} and Gy\H{o}ri's original proof is generalized to weighted undirected graphs in~\cite{chandran2018spanning}.
Polynomial algorithms to also compute such connected partitions are only known for the particular cases   $k=2,3,4$~\cite{suzuki1990linear,wada1993efficient,hoyer2019independent} and all $k\ge 5$ are still open. 
A restricted case of $\textsc{BCP}$ where the partitions are allowed to differ only by a size of one, has been studied from the FPT viewpoint~\cite{enciso2009makes}.

\textsc{$W$-weight Separator} occurs in the literature under different names. The unweighted version is studied under the names \textit{$p$-Size Separator} \cite{xiao2017linear} or \textit{$\ell$-Component Order Connectivity} (\textsc{COC}) \cite{drange2014computational, kumar20162}; where $p,\ell=W-1$ translate this to our definition of \textsc{$W$-weight Separator} with unit weights. In~\cite{drange2014computational} a weighted version of \textsc{COC} denoted by \textit{Weighted Component Order Connectivity} (\textsc{$w$COC}) is introduced. This problem differs from our \textsc{$W$-weight Separator} by searching for a set $S$ with  $w(S)\leq k$ instead of $|S|\leq k$.

Note that \textsc{$W$-weight Separator} with $W=2$ and unit weights yields the classical problem \textsc{Vertex Cover}. This in particular shows that~$W$ (alone) is not a suitable parameter from the FPT viewpoint. Further, \textsc{$W$-weight Separator} is \textbf{W[1]}-hard for parameter~$k$, even when restricted to split graphs~\cite{drange2014computational}. These lower bounds lead to studying parameterization by~$W+k$. Stated with $\ell=W-1$, a kernel of size $9k\ell$ is given in~\cite{xiao2017linear}. Also~\cite{kumar20162} derives a kernel of size $2k\ell$  in time $\bigO\left(|V|^\ell\right)$. Both of these results are for unit weights. An $\bigO\left(k\ell(k+\ell)^2\right)$ weight kernel for the related problem  \textsc{$w$-COC} is given in~\cite{drange2014computational}. 

For $W=3$ and unit weights, \textsc{$W$-weight Separator} corresponds to~\textsc{Vertex Cover $P_3$} or \textsc{3-path Vertex Cover} (see e.g.~\cite{vertexcoverp} and~\cite{pathvc}), first studied by Yannakakis~\cite{p2history} under the name~\textsc{Dissociation Number}. The best known kernel for this problem is of size $5k$ and given in~\cite{kernelp3k}.
\textsc{$W$-weight Packing} with unit weights is equivalent to  
\textsc{$T_r$-Packing} with $r=W+1$, where $T_r$ is a tree with at least~$r$ edges, as defined in~\cite{chen2012kernels}; note that any connected component with at least $W$ vertices has at least $r-1$ edges, and any tree with $r-1$ edges has exactly~$W$ vertices. The best known  kernel for this problem is of size  $\mathcal{O}\left(kW^3\right)$ by~\cite{chen2012kernels}.

\textsc{$W$-weight Packing} is also studied for particular values of $W$. The case  $W=2$ with unit weights is equivalent to the \textsc{Maximum Matching} problem; note that a matching of size $k$ can be derived from a solution $V_1,\dots, V_k$ for \textsc{$2$-weight Packing}  by choosing arbitrarily any edge in a set $V_i$ with $|V_i|>2$. In a similar way, the particular case of $W=3$ is a problem studied under the names \textsc{$P_2$-packing} or \textsc{Packing 3-Vertex Paths} (see e.g~\cite{ptwopack} and~\cite{threevp}). A $5k$~kernel for this problem is given in~\cite{p2packing}.

	\section{Balanced Expansion}\label{sec:be}
		In this section we introduce a balanced generalization of weighted expansions that we call \emph{balanced expansion}. Full proofs of the results in this section are given in \cref{sec:appbalexp}.

Balanced expansion extends the existing weighted expansion structures and
is one of the ingredients to derive our main BCD structure in the next section.
Like the weighted expansion, it is a structure on bipartite graphs. We write $G=(A\cup B, E, w)$ for bipartite vertex-weighted graphs, where $w\colon A\cup B\rightarrow \mathbb N$ is its weight function. See Figure~\ref{pic::MaxFlowNetwork2} for an illustration of this structure.

\begin{restatable}[\textit{balanced expansion}]{definition}{balancedExpansion}
	\label{definition::expanse}
	Let $G=\left(A \cup B, E, w\right)$ be a bipartite vertex-weighted graph, where $\wmax^B=\max_{b \in B} w\left(b\right)$.
	For $q \in \Nz$, a partition $A_1\cup A_2$ of $A$ and $f\colon B \to A$, the  tuple $\left(A_1,A_2,f,q\right)$ is called  a \expanse{} if:
	\begin{enumerate}
		\item $w\left(a\right) + w\left(f^{-1}\left(a\right)\right)\, \begin{cases}\geq q -\wmax^B +1, & a\in A_1\\ 
			\leq q +\wmax^B -1, & a\in A_2\end{cases}$  
		\item $f\left(b\right)\in N\left(b\right)$
		\item  $N\left(f^{-1}\left(A_1\right)\right)\subseteq A_1$	
	\end{enumerate}
\end{restatable}

Our main result of this section is the following theorem.
\begin{restatable}[\emph{\expanse}]{theorem}{balancedExpansionThm}
	\label{lemma::cycleCanceling}
	Consider a vertex-weighted bipartite graph $G=\left(A \cup B, E,w\right)$ with no isolates in $B$, and $q \geq \max_{b \in B} w\left(b\right) = \wmax^B$. A {\expanse} $\left(A_1,A_2,f,q\right)$ for $G$ can be computed in $\mathcal{O}\left(|V|\,|E|\right)$ time.	
	Furthermore, if  $w\left(A\right) + w\left(B\right) \geq q |A|$, then $A_1 \ne \varnothing$.
\end{restatable}

As intermediate structure we use a fractional version of the balanced expansion where we partially assign weights from vertices of $B$ to vertices of $A$ encoded as edge weights.


\begin{restatable}[\textit{\fexpanse}]{definition}{fracBalancedExp}
	\label{definition::Separator-Packing}
	Let $G=\left(A \cup B, E, w\right)$ be a bipartite vertex-weighted graph. For $q \in \Nz$, a partition $A_1\cup A_2$ of $A$ and $g\colon E \to \Nz$, the tuple $\left(A_1,A_2,g,q\right)$ is called  \fexpanse{ }if:
	\begin{enumerate}
		\item $w\left(a\right) + \sum_{b \in B} g\left(ab\right)\, \begin{cases}\geq \ q, & a\in A_1\\ 
			\leq \ q, & a\in A_2\end{cases}$  
		\item $\forall b \in B\colon \sum_{a \in A} g\left(ab\right) \leq w\left(b\right)$ (capacity)
		\item  $N\left(B_U \cup B_{A_1}\right) \subseteq A_1$ (separator)\\
		where $B_a :=\left\{b \in B\mid g\left(ab\right) > 0 \right\}$ for $a \in A$, $B_{A'} := \bigcup_{a \in A'} B_a$ for $A' \subseteq A$ and $B_U :=\left\{b \in B\mid \sum_{a \in A} g\left(ab\right) < w\left(b\right)\right\}$  
	\end{enumerate}
\end{restatable}

We prove a fractional version of our result first in the following lemma.

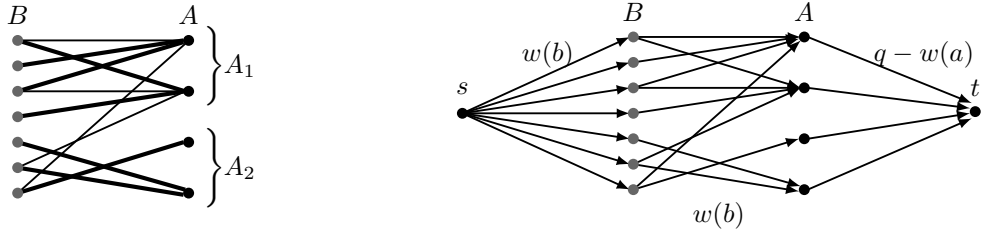
\begin{figure}
	\begin{center}

	\begin{tikzpicture}[scale=0.8,y=0.80pt, x=0.80pt, inner sep=0pt, outer sep=0pt]
			\path[draw=black, line width=0.25mm] (0,60)--(98,60);	
			\path[draw=black, line width=0.55mm] (0,60)--(98,30);
			\path[draw=black, line width=0.55mm] (0,45)--(97,60);			
			\path[draw=black, line width=0.55mm] (0,30)--(97,59);
			\path[draw=black, line width=0.25mm] (0,30)--(98,30);
			\path[draw=black, line width=0.55mm] (0,15)--(98,30);
			\path[draw=black, line width=0.55mm] (0,0)--(98,-28);
			\path[draw=black, line width=0.25mm] (0,-30)--(98,58);
			\path[draw=black, line width=0.55mm] (0,-30)--(97,0);
			\path[draw=black, line width=0.25mm] (0,-15)--(98,30);
			\path[draw=black, line width=0.55mm] (0,-15)--(97,-31);

			\path[fill=black!60!, line width=0.25mm] (0,60) circle (0.09cm);
			\path[fill=black!60!, line width=0.25mm] (0,45) circle (0.09cm);
			\path[fill=black!60!, line width=0.25mm] (0,30) circle (0.09cm);
			\path[fill=black!60!, line width=0.25mm] (0,15) circle (0.09cm);
			\path[fill=black!60!, line width=0.25mm] (0,0) circle (0.09cm);
			\path[fill=black!60!, line width=0.25mm] (0,-15) circle (0.09cm);			
			\path[fill=black!60!, line width=0.25mm] (0,-30) circle (0.09cm);
			\path[fill=black, line width=0.25mm] (100,60) circle (0.09cm);
			\path[fill=black, line width=0.25mm] (100,30) circle (0.09cm);
			\path[fill=black, line width=0.25mm] (100,0) circle (0.09cm);			
			\path[fill=black, line width=0.25mm] (100,-30) circle (0.09cm);
					
			\node at (0,75) {$B$};
			\node at (100,75) {$A$};
			\node at (115,45) {{\large $\Bigg \}$}};
			\node at (130,45) {$A_1$};
			\node at (115,-15) {{\large $\Bigg \}$}};
			\node at (130,-15) {$A_2$};
			\node at (50,-50) {};
			\end{tikzpicture} \hspace*{2.4cm}
			\begin{tikzpicture}[scale=0.8,y=0.80pt, x=0.80pt, inner sep=0pt, outer sep=0pt]

			\path[draw=black, line width=0.25mm, -latex] (-100,15)--(-2,60);
			\path[draw=black, line width=0.25mm, -latex] (-100,15)--(-2,30);
			\path[draw=black, line width=0.25mm, -latex] (-100,15)--(-2,0);
			\path[draw=black, line width=0.25mm, -latex] (-100,15)--(-2,-30);
			\path[draw=black, line width=0.25mm, -latex] (-100,15)--(-2,45);
			\path[draw=black, line width=0.25mm, -latex] (-100,15)--(-2,15);
			\path[draw=black, line width=0.25mm, -latex] (-100,15)--(-2,-15);

			\path[draw=black, line width=0.25mm, latex-] (198,21)--(102,60);
			\path[draw=black, line width=0.25mm, latex-] (197,18)--(102,30);
			\path[draw=black, line width=0.25mm, latex-] (197,15)--(102,0);
			\path[draw=black, line width=0.25mm, latex-] (198,12)--(102,-30);

			\path[draw=black, line width=0.25mm, -latex] (0,60)--(98,60);	
			\path[draw=black, line width=0.25mm, -latex] (0,60)--(98,30);
			\path[draw=black, line width=0.25mm, -latex] (0,45)--(97,60);			
			\path[draw=black, line width=0.25mm, -latex] (0,30)--(97,59);
			\path[draw=black, line width=0.25mm, -latex] (0,30)--(98,30);
			\path[draw=black, line width=0.25mm, -latex] (0,15)--(98,30);
			\path[draw=black, line width=0.25mm, -latex] (0,0)--(98,-28);
			\path[draw=black, line width=0.25mm, -latex] (0,-30)--(98,58);
			\path[draw=black, line width=0.25mm, -latex] (0,-30)--(97,0);
			\path[draw=black, line width=0.25mm, -latex] (0,-15)--(98,30);
			\path[draw=black, line width=0.25mm, -latex] (0,-15)--(97,-31);

			\path[fill=black!60!, line width=0.25mm] (0,60) circle (0.09cm);
			\path[fill=black!60!, line width=0.25mm] (0,45) circle (0.09cm);
			\path[fill=black!60!, line width=0.25mm] (0,30) circle (0.09cm);
			\path[fill=black!60!, line width=0.25mm] (0,15) circle (0.09cm);
			\path[fill=black!60!, line width=0.25mm] (0,0) circle (0.09cm);
			\path[fill=black!60!, line width=0.25mm] (0,-15) circle (0.09cm);			
			\path[fill=black!60!, line width=0.25mm] (0,-30) circle (0.09cm);
			\path[fill=black, line width=0.25mm] (100,60) circle (0.09cm);
			\path[fill=black, line width=0.25mm] (100,30) circle (0.09cm);
			\path[fill=black, line width=0.25mm] (100,0) circle (0.09cm);			
			\path[fill=black, line width=0.25mm] (100,-30) circle (0.09cm);
			\path[fill=black, line width=0.25mm] (200,16) circle (0.09cm);
			\path[fill=black, line width=0.25mm] (-100,15) circle (0.09cm);	
			
			\node at (-100,30) {$s$};
			\node at (200,30) {$t$};
			\node at (0,75) {$B$};
			\node at (100,75) {$A$};
			\node at (-50,50) {$w(b)$};
			\node at (170,50) {$q-w(a)$};
			\node at (50,-45) {$w(b)$};
		\end{tikzpicture} 
		
					\end{center}
	\caption{Left: Balanced expansion for $w(b)=1$ for all $b\in B$, $q=2$ and assignment $f$ depicted with bold edges.  Right: Flow network embedding of the graph on the left.  }
	\label{pic::MaxFlowNetwork2}
\end{figure}
\begin{restatable}[\emph{\fexpanse}]{lemma}{fracBalancedExpLemma}
	\label{lemma::separatorPacking}
	Given a vertex-weighted bipartite graph $G=\left(A \cup B, E,w\right)$ with no isolated vertices in $B$ and $q\in \Nz$, a \fexpanse{ }$\left(A_1,A_2,g,q\right)$ can be computed in $\mathcal{O}\left(|V|\,|E|\right)$.
	Also, if  $w\left(A\right) + w\left(B\right) \geq q |A|$, then $A_1 \ne \varnothing$.
\end{restatable}
\begin{sketch}
	The main idea is to embed $G$ to a capacitated flow network in a standard way (see Figure~\ref{pic::MaxFlowNetwork2}).
We construct a network $\Nw = (A \cup B \cup \left\{s,t\right\}, \overrightarrow{E}, c)$.
To obtain $\Nw$ from $G$, add source $s$ and sink $t$, and arcs $\overrightarrow{E}$ with a capacity function $c\colon\overrightarrow{E} \to \mathbb{N}$ defined as follows.
For every $b \in B$, add an arc $\overrightarrow{sb}$ with capacity $w\left(b\right)$ and for every $a \in A$, add an arc $\overrightarrow{at}$ with capacity $q - w\left(a\right)$.
Moreover, transform every edge $ab \in E$ to an arc $\overrightarrow{ba}$ with capacity $w\left(b\right)$. 

We compute a max flow $f\colon\dir{E} \to \mathbb{N}$
and define the saturated vertices $A' \subseteq A$ as $a \in A$ with $f(\dir{at}) = c(\dir{at})$. We now gradually build the sets $A_1$ and $A_2$. 
The vertices of $A'$ are potential vertices for $A_1$ while the unsaturated vertices are immediately added to $A_2$.
We define $F := \sum_{\dir{e} \in \delta^{-}(t)} f(\dir{e})$ as the flow value, where $\delta^{-}(v)$ denotes the incoming, and $\delta^{+}(v)$ the outgoing arcs for
$v \in V(G)$. 
The final selection of $A_1$ follows by individually increasing the capacity by one for each $\dir{at}$ for $a \in A'$, and checking whether the flow value increases by computing a new max flow $f_a$ with the increased capacity of $\dir{at}$.
Let $F_a$ be the flow value when the capacity of $\dir{at}$ is increased by one.
If $F_a > F$, then the vertex is added to $A_1$, otherwise it is added to $A_2$.
The intuition behind this selection can be explained as follows:
first observe that each $b \in B$ that has an edge $ba_2$ to some $a_2\in A_2$ is saturated, i.e.~$\sum_{\dir{e} \in \delta^{+}(b)} f(\dir{e}) = w(b)$.
Otherwise, we could route an additional unit of flow from $b$ to $a_2$ either in $f$ or in $f_{a_2}$, giving a contradiction to the fact that $a_2 \in A_2$. 
Consequently, every unsaturated $b \in B$ is adjacent only to $A_1$.
The second observation is that there are no $b \in B$ with $f(\dir{ba_1}) > 0$ and $ba_2 \in \dir{E}$ for $a_1 \in A_1$ and $a_2 \in A_2$. If such a $b$ exist, 
we show that we can route an extra unit of flow from $b$ to $a_2$ either in $f$ or $f_{a_2}$.
The idea is that we could reroute one unit of flow from $\dir{ba_1}$ to $\dir{b{a_2}}$ creating a vacuum for one unit of incoming flow in $a_1$.
Since $f_{a_1}$ routed one unit flow more than $f$, we could use a similar flow routing as in $f_{a_1}$ to fill this vacuum,
thus contradicting the maximality of either $f$ or $f_{a_2}$.
As a result, all vertices added to  $A_1$ have the desired exclusive neighborhood in $B$ encoded by $f$.
Finally, in order to derive $g$ we convert the flow arc values of $f$ to edge weights for $g$.
Note that the required upper bound on the assignment of $A_2$ follows from  the capacities of the arcs from $A$ to $t$, and the required lower bound on the assignment of $A_1$ follows from the vertices in $A_1$ being saturated.
Regarding running time, we remark that it is sufficient to find one max-flow $f$ at the beginning and then computing each $f_a$ with only one augmenting flow step.
\end{sketch}

\begin{sketch}[Proof Sketch of \cref{lemma::cycleCanceling}]
	Once we have the fractional balanced expansion $g$, our first step is modifying the edge weights $g$ such that the edge-weighted graph $G' := (V,\{ab \in E| g(a,b) > 0\},g)$ becomes a forest, without changing the sum $\sum_{b \in N(a)} g(a,b)$ for any $a \in A$ and at the same time ensuring that $\sum_{a \in N(b)} g(a,b) \leq w(b)$ for all $b \in B$.
	This is possible through a standard cycle canceling process.
	Now consider the trees in this forest.
	The trees intersecting $A_1$ are disjoint from the trees intersecting $A_2$ due to the separation property of the balanced fractional expansion.
	For a tree $T$ intersecting $A_1$, we allocate each $b\in V(T)\cap B$ completely to its parent in $T$.
	This way, any $a\in V(T)\cap A_1$ loses at most the assignment from its parent and hence its assignment decreases by at most $\wmax^B-1$  as required. 
	Now consider a tree $T$ intersecting $A_2$.
	If a $b\in V(T)\cap B$ is a leaf of $T$ its assignment has to be non-fractional, so it can be completely assigned to its parent $a$ and deleted from the tree.
	This way, all leafs can be assumed to be from $A_2$.
	We then allocate each $b\in V(T)\cap B$ to one of its children, and thus to every $a\in V(T)\cap A_2$  at most the assignment from its parent is added, and hence the assignment increases by at most $\wmax^B-1$ as required.
\end{sketch}
Before moving to BCD, we formally state the aforementioned implication of the results in this section on the runtime of computing (non-balanced) expansions. 


	\begin{restatable}[Weighted Expansion Lemma]{lemma}{WeightedExpansion}\label{lemma::WeightedExpansion}
		Let $G = \left(A \cup B,E\right)$ be a  bipartite graph without isolated vertices in $B$, $w\colon B\rightarrow \left\{1,\dots,W\right\}$, and $q\in \Nz$. A $q$-weighted expansion  $(f,H,C)$ in $G$
	  can be computed in time $\mathcal{O}\left(|A\cup B|\,|E|\right)$. Furthermore, if  $w\left(B\right) \geq q|A|$ then $H\neq \emptyset$.
\end{restatable}

	\section{Balanced Crown Decomposition}\label{sec:tec_intro}
		\def\hatH{\widehat{H}}
\def\hatC{\widehat{C}}
In this section we introduce our combination of balanced connected partition and crown decomposition that we call \emph{balanced crown decomposition}, formally defined as follows (see also Figure~\ref{figure::crd} for an illustration).  

\begin{definition}\label{def:lccd}
	A \textbf{$\lambda$-balanced crown decomposition} ($\lambda$-BCD) of a vertex-weighted graph $G=(V,E,w)$ is a tuple  $(\chrf)$, where $\{H,C,R\}$ is a partition of~$V$, the set $\mathfrak{R}$ is a partition of  $R$, and $f\colon \mathbb{CC}(C)\rightarrow H$ where $\mathbb{CC}(C)$ is the set of connected components of $G[C]$, such that:

\noindent
	$
\left.\parbox{0.65\textwidth}{
\begin{itemize}
\item[1.] \text{there are no edges from $C$ to $R$}
\item[2.] \text{$w(Q)<\lambda$ for each  $Q\in\mathbb{CC}(C)$ } 
\item[3.] \text{$f(Q) \in N(Q)$ for each $Q\in\mathbb{CC}(C)$ } 
\item[4.]  \text{$w(h)+w(f^{-1}(h))\geq \lambda$ for each $h\in H$	}	
\end{itemize}}
\right\}$ \hfill  \emph{(weighted crown dec.)}
\begin{itemize}
\item[5.] $G[R']$ is connected and $\lambda\leq w(R')\leq  3\lambda-3$ for each $R' \in \calR$. 	
	\end{itemize}
	\end{definition}
	Our novel contribution is condition 5, that gives a balanced connected partition of the \emph{body}.
	Without this condition, the structure is same as the weighted crown decomposition \cite{xiao2017linear}.
	Observe that if there is a connected component of weight less than~$\lambda$ in~$G$, then there is no $\lambda$-balanced crown decomposition for $G$. In the applications of BCD, such small components in the input are usually removed through some form of preprocessing. 
	
	We point out that the ratio $3$ between the upper and lower bound in condition~5 of BCD is not arbitrary, but the best possible, since we want to ensure the existence of this structure in case all connected components have weight at least $\lambda$. A simple tight example is a triangle with each vertex having a weight of $\lambda - 1$; here, $C = H = \varnothing$ is the only possibility   and hence $\calR = \{V\}$ is the only possible partition of $R=V$. 

\begin{figure}[h]
	\begin{center}
	\begin{tikzpicture}[y=0.8pt, x=0.8pt, yscale=-1.000000, xscale=1.000000, inner sep=0pt, outer sep=0pt]
	\path[draw=black, line width=0.25mm] (201.2069,28.7262) ellipse
	(0.2400cm and 0.2347cm);
	\path[draw=black, line width=0.25mm] (201.7738,66.5238) ellipse
	(0.2987cm and 0.2667cm);
	\path[draw=black, line width=0.25mm] (202.9077,99.4077) circle
	(0.1280cm);
	\path[draw=black, line width=0.25mm] (203.2857,132.1027) ellipse
	(0.2347cm and 0.2187cm);
	\path[fill=black, line width=0.25mm] (87.8140,45.1682) circle (0.09cm);
	\path[fill=black, line width=0.25mm] (88.5699,92.0372) circle
	(0.09cm);
	\path[fill=black, line width=0.25mm] (88.1920,68.9807) circle
	(0.09cm);
	\path[fill=black, line width=0.25mm] (88.9479,115.4717) circle
	(0.09cm);
	\path[draw=black, line width=0.25mm]
	(87.6250,44.6012) -- (192.7024,27.2143);
	\path[draw=black, line width=0.25mm]
	(87.6250,44.9792) -- (191.1905,63.8780);
	\path[draw=black, line width=0.25mm]
	(193.0804,30.9940) -- (88.0030,68.7917) -- (198.3720,98.2738);
	\path[draw=black, line width=0.25mm]
	(88.0030,68.7917) -- (197.2381,126.6220);
	\path[draw=black, line width=0.25mm]
	(191.5685,69.5476) -- (89.1369,91.8482) -- (194.9702,130.7798);
	\path[draw=black, line width=0.25mm]
	(198.3720,100.9196) -- (89.1369,115.2827) -- (194.9702,34.3958);
	\path[draw=black, line width=0.25mm]
	(88.7589,115.6607) -- (194.9702,134.5595);
	\path[draw=black, dashed, line width=0.25mm]
	(88.0030,68.7917) -- (87.6250,44.9792) .. controls (87.6250,44.9792) and
	(76.4256,58.9439) .. (76.3708,67.2662) .. controls (76.3109,76.3857) and
	(88.3810,91.8482) .. (88.3810,91.8482);
	\path[draw=black, dashed, line width=0.25mm]
	(89.1369,114.5268) .. controls (89.1369,114.5268) and (72.1188,100.7389) ..
	(72.1534,91.0916) .. controls (72.1853,82.1763) and (88.0030,69.5476) ..
	(88.0030,69.5476);
	\path[draw=black, dashed, line width=0.25mm] (-7, 120)-- (-20, 91);

	\path[draw=black, line width=0.25mm]
	(7,115)--(88.0030,91.4702);
	\path[draw=black, line width=0.25mm]
	(-13.2,70) -- (89.1369,115.6607);
	\path[draw=black, line width=0.25mm]
	(-13,50) -- (87.6250,44.2232);
	\path[draw=black, line width=0.25mm] (-20,60) ellipse
	(0.2cm and 0.9cm);
	\path[draw=black, line width=0.25mm] (-50,87) ellipse (0.3cm and 0.9cm);
	\path[draw=black, line width=0.25mm] (0,120) ellipse
	(0.2cm and 0.9cm);
	\path[draw=black, dashed,line width=0.25mm] (-39,90)-- (-5, 129);
	\path[draw=black, dashed,line width=0.25mm] (-39,80)-- (-27, 70);
	\node at (-25,10) {$\mathfrak{R}$};
	\node at (87.6250,22) {$H$};
	\node at (201.2069,6) {$C$};
	\node at (230.2069,66) {$<\lambda$};
	\node at (-20, 140) {$\lambda\leq$};
	\node at (35, 140) {$\leq 3\lambda - 3$};
	\end{tikzpicture}
\end{center}
\caption{$\lambda$-balanced crown decomposition.}\label{figure::crd}
\end{figure}
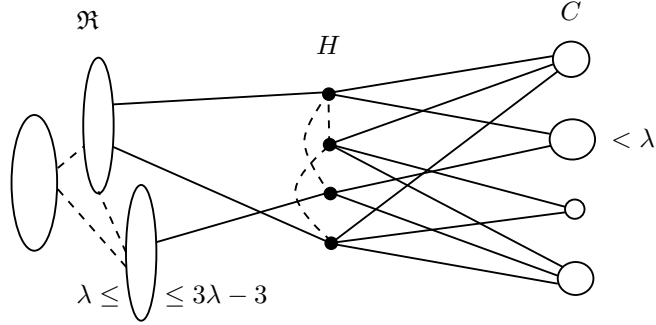

The main structural result of the paper is the following.
\begin{theorem}[Balanced Crown Decomposition Theorem]\label{theorem:lccd}
	Let $G = (V,E,w)$ be a vertex-weighted graph and $\lambda\in \N$, 
	such that each connected component in $G$ has weight at least $\lambda$.
	A $\lambda$-balanced crown decomposition $(C,H,\calR,f)$ of $G$ can be computed in $\mathcal{O}\left(k^2\,|V|\,|E|\right)$ time, where $k=|H|+|\calR|\le \min \left\{ w(G)/\lambda, |V| \right\}$.
\end{theorem}
The proof of this result is very technical and we therefore here only give a very high-level overview of the ideas.
For a full proof, see \cref{sec:lccd}. Observe that the condition $|H|+|\calR|\le \min \left\{ w(G)/\lambda, |V| \right\}$ holds, since $\left\{ \left\{ h \right\}\cup \finv(h):h\in H \right\}\cup \calR$ is a partition of the vertices with each part having weight at least $\lambda$. This bound is also used to track our progress in our BCD algorithm. 
We maintain a set $H$ that can be thought of as a \emph{potential head} (not necessarily a separator), a set of connected components of weight smaller than $\lambda$ (some of them assigned to vertices in $H$ by a partial assignment~$f$) which can be thought of as a \emph{potential crown}, and a remaining \emph{body} that is packed according to condition~5.

To easily talk about condition~5 in the following, we say $\calU$ is a \emph{connected packing} in $V' \subseteq V$, if for every $U \in \calU$ we have $U \subseteq V'$, $U$ induces a connected subgraph in $G$ and $\bigcap_{U \in \calU} U = \varnothing$.
We say $\calU$ is an \emph{$[a,b]$-connected packing} of $V'$ if $w(U) \in [a,b]$ for every $U \in \calU$ and that~$\calU$ is \emph{maximal} if the remaining graph does not have a connected component of weight at least~$a$.
Recall that we say $\calU$ is a {\cvp} or  $[a,b]$-CVP of $V'$ if additionally $\bigcup_{U \in \calU} U = V'$ holds, and observe that condition~5~asks for a $[\lambda, 3\lambda-3]$-CVP of the body~$R$.

\begin{sketch}[Proof Sketch of \cref{theorem:lccd}]
Let $G = (V,E,w)$ be a vertex-weighted graph and $\lambda\in \mathbb N$ such that each connected component of $G$ has weight at least $\lambda$.
We reduce $G$ by deleting all vertices of weight more than~$\lambda$ and all connected components of size smaller than~$\lambda$ that occur after this deletion. 
Suppose we have a $\lambda$-BCD for the reduced graph, then a $\lambda$-BCD of $G$ can be built by adding the deleted heavy vertices to the head, the deleted small components to the crown and assigning (by the function $f$ in the definition of $\lambda$-BCD)  each of these components arbitrarily to a heavy vertex it is adjacent to.
Thus, we can assume that every vertex has weight at most~$\lambda$.
See \cref{figure::codmain} for an illustration of the structures arising below.

We start with a maximal $[\lambda,2\lambda]$-connected packing of $G$ which is obtained greedily (slight deviation from the main proof to provide better intuition).
Let $\calR=\{R_1,R_2,\cdots \}$ be this packing, $C$ be the vertices not in the packing, and let $\mathbb{CC}(C)=\{C_1,C_2,\cdots\}$ be the connected components of $G[C]$. Note that by the maximality of the packing, $w(C_i)<\lambda$ for all $i$. 
Think of $\calR$ as the current body and $C$ as the current crown, and the head is empty in the beginning.
Note that at this point we do not ensure that there are no edges between crown and body.
If $C$ is empty then we already have a $\lambda$-BCD (with empty crown and head).
Also, if we can somehow assign  each $C_i$ to some adjacent $R_i$ such that each $R_i$ is assigned weight at most~$3\lambda$ (including its own weight), then we have also built a $\lambda$-BCD (with empty crown and head).
Assuming neither of these cases hold, there has to exist an $R_i$ such that its weight plus the weight of the neighborhood in the crown part is at least~$3\lambda$; recall that we assumed that all connected components of $G$ have weight at least~$\lambda$, so each~$C_j$ is connected to at least one component in~$\calR$.
We call the subgraph induced by~$R_i$ together with all $C_j$ that are connected to it the \emph{effective neighborhood of $R_i$}, and its weight the \emph{effective weight of $R_i$}.

In case we have not found a $\lambda$-BCD yet, we pick an $R_i$ with effective weight at least~$3\lambda$ and use the following fact derived from the famous results of Tarjan~\cite{Tarjan72,even1976computing}:
for any connected graph of total weight at least~$3\lambda$ and largest vertex weight at most~$\lambda$, we can efficiently 
either find a partition of it into two connected subgraphs of weight at least~$\lambda$ each, or find a cut-vertex that cuts the graph into components each of weight less than~$\lambda$. 
If the effective neighborhood is divided into two, we take each of the two parts into the body and remove $R_i$, thus increasing the body size by one. In the other case, that is, if we find a cut-vertex $c$, then we add $c$ to the head and the components of $R_i\setminus\left\{ c \right\}$ (each having weight less than~$\lambda$) to the crown $\mathbb{CC}(C)$.
We assign with a \emph{partial} function $f\colon \mathbb{CC}(C)\rightarrow H$  some of these components that we just added to $\mathbb{CC}(C)$ to $c$ such that $c$ is assigned a total weight of at least $3\lambda$ (including weight of $c$).
The reason for assigning $3\lambda$ when we only require $\lambda$ by the definition of BCD, will become clear in the following.
The new components added to the crown could have edges to the old components there and hence can merge with these.
If at any point it happens that there is a component of weight at least~$\lambda$ in the crown, then we immediately add it to the body.
This could cause some head vertex to loose some of its assignment, but since it had $3\lambda$ assigned to it, an  assignment of at least~$\lambda$ remains.
This is because we ensure that the part we move to the body can have  weight at most~$2\lambda$, as we move it immediately as the weight is at least~$\lambda$, and each addition is by steps of less than~$\lambda$.

We repeat this process of picking an effective neighborhood of an $R_i$ and dividing or cutting it.
We point out that when we calculate effective neighborhoods and weights, we do not consider the crown parts that are already assigned by~$f$.
This process continues until the effective weights of all sets $R_i$ are less than~$3\lambda$.
We claim that the reason why we have not arrived at a $\lambda$-BCD yet could only be that there are crown parts that do not have edges to the body (we call them \emph{private components}) and \emph{not} assigned by~$f$, while there are also crown parts that have edges to the body (non-private components) and assigned  by~$f$. Note that if there are no unassigned private components, we can merge all unassigned crown parts with some set in the body and create a $[\lambda,3\lambda]$-CVP given by the $R_i$'s and the sets $\{v\}\cup f^{-1}(v)$. 
Note that since the effective weights were lighter than $3\lambda$, the body parts after the merging are lighter than $3\lambda$.
Also, if~$f$ does not assign any non-private components, we can assign unassigned private components to arbitrary neighbors in $H$, and merge unassigned non-private components to body obtaining a $\lambda$-BCD.

We modify the assignment~$f$ to switch non-private with private components in the best way possible. For this we use the balanced expansion \cref{lemma::cycleCanceling}. We build the bipartite graph where the set $A$ are the head vertices, and the set $B$ are the private crown components each contracted into a single vertex.
\cref{lemma::cycleCanceling} with expansion parameter $2\lambda$ then gives us sets $A'\subseteq A$ and $B'\subseteq B$ and an assignment $f'$ such that $w(\{a'\}\cup f'^{-1}(a'))\geq \lambda$ for all $a'\in A'$ and $w(\{a\}\cup f'^{-1}(a))\leq 3\lambda$ for all $a\in A\setminus A'$, and the crown components in $B'$ are completely assigned to $A'$ and do not have neighbors in $A\setminus A'$. Note that since $B$ was the set of private components, the components $B'$ do not have neighbors in the body either. 
Now augment~$f'$ by assigning to $A\setminus A'$ also enough non-private components such that they have an assignment of at least $3\lambda$ each. This is possible since each vertex in $A\setminus A'$ has an assignment of $3\lambda$ by~$f$ which did not use any components from $B'$ (as there are no edges from $B'$ to $A\setminus A'$).
Note that this augmentation of~$f'$ needs to be done carefully since the private components could be assigned by the balanced expansion differently than by~$f$. 

By $f'$ all  private components are now assigned to the head, but there could still be non-private ones assigned as well. 
But now, if the effective weight of each $R_i$ is at most~$3\lambda$, we can add the unassigned crown parts to sets in $\calR$, and thus create a $\lambda$-BCD: $A'$ with its assignment by $f'$ are head and crown, and $\calR$ plus the sets $\{a\}\cup f'^{-1}(a)$ with $a\in A\setminus A'$ are a $[\lambda,3\lambda]$-CVP of the body. 
Thus, if we are not successful, there exists an $R_i$ with effective weight more than~$3\lambda$ and we continue by dividing or cutting it.
Note that we can proceed with $f'$ replacing $f$ although some head vertices (from what was $A'$ in the balanced expansion) might only have weight~$\lambda$ assigned to them (and not $3\lambda$), because the crown parts assigned to them are private  and hence do not interfere with the further process.


To analyse the run time, we estimate how often we divide or cut a set  $R_i$; note that each such step can be performed in $\mathcal O(|V||E|)$. Throughout our algorithm, the value $|H|+|\calR|$ is non-decreasing, and upper bounded by~$k$. Every time we divide some~$R_i$, we increase $|H|+|\calR|$, hence this happens at most~$k$ times. Every time we cut some~$R_i$ we increase $|H|$ by one.
Since~$|H|$ is also upper bounded by~$k$, and we are careful not to decrease $|H|$ with the balancing step in-between cut steps, we arrive at a total of at most $k^2$ divide or cut steps.

One pitfall here is that after applying the balanced expansion one might be tempted to just take $A'$ and its assignment via $f'$ into head and crown respectively, delete it, and start over on the rest of the graph.
The problem with this is that we are not guaranteed to find a non-empty set $A'$ (since the private components might not have weight at least $\lambda|A|$). The way we augment $f'$, we ensure that we retain the preliminary crown, head and body structure, and with this especially the value $|H|$, and can split up another~$R_i$ to either increase $|H|$ or $|H|+|\calR|$.
Further, the reason why \emph{we cannot use the standard weighted expansion lemma} here is that we would need a lower bound of at least $\lambda|A|$ on the weight of $B$ for this. We cannot ensure that the private components of the crown alone have a weight of at least~$\lambda|A|$, since we also used the non-private components for the assignment $f$.

One detail that we did not mention so far is that it is not possible to assign exactly~$3\lambda$ to each head vertex.
Since the step size we can guarantee is only $\lambda$, we might have to assign $(4\lambda-1)$ in order to get a value of at least~$3\lambda$.
Recall that  we assign a collection of components of weight less than $\lambda$. Without further work, this only yields an upper bound of $4\lambda$ instead of $3\lambda$ for the packing of the body, worsening the quality of our structure (we for example would only get a $4$- instead of a $3$-approximation for the BCP problems).
For this improvement from~4 to~3, we maintain a ``last component'' as a special assignment. (This is called  $g'$-assignment in the full proof).
The details of how we make use of this special second assignment is rather technical, and to some extent complicates the proof. 
If one is satisfied with a bound of $4\lambda$ for the body, this complication is not necessary. 

Another technical detail we skipped is that in the assignment~$f$ we maintain, the crown parts we map may not be whole sets in $\mathbb{CC}(C)$ (connected components induced by crown vertices). They are connected, but can be subgraphs of some  $C_i\in\mathbb{CC}(C)$. We call such subgraphs \emph{sub-components}.
Different sub-components of some $C_i$ can be assigned to different heads.
Also, for a $C_i$ some of its sub-components can be assigned while others are not. For our structure to converge to a  $\lambda$-BCD, sub-components have to be classified as private or non-private based on the set $\mathbb{CC}(C)$ they are a part of, so it can happen that a sub-component is non-private but has no edge to the body. Whenever we make the move from crown to body, we therefore have to do a merging of some sub-components such that for each  $C_i\in \mathbb{CC}(C)$ either all its sub-components are assigned to the head or none of them are. 
\end{sketch}

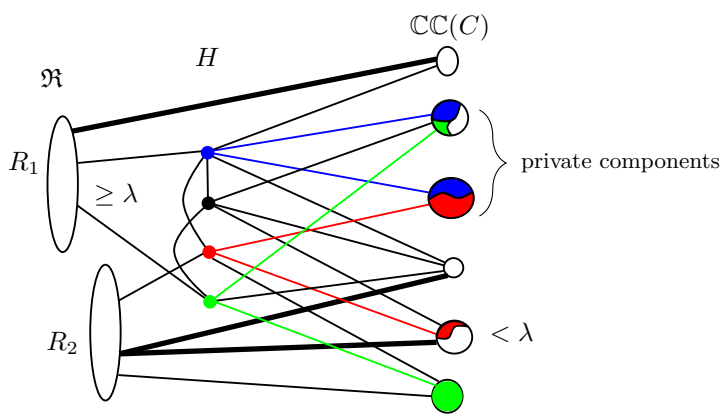
\begin{figure}[h]
	\begin{center}
		\begin{tikzpicture}[y=0.80pt, x=0.80pt, yscale=-1.000000, xscale=1.000000, inner sep=0pt, outer sep=0pt]
			\path[draw=black, line width=0.7mm] (510.2,103)--(356.5,140);
			\path[draw=blue, line width=0.25mm]
			(397.6250,44.6012) -- (502.7024,27.2143);
			\path[draw=blue, line width=0.25mm]
			(397.6250,44.9792) -- (501.1905,63.8780);
			\path[draw=black, line width=0.25mm]
			(396,45)--(509, 97);
			\path[draw=black, line width=0.25mm]
			(503.0804,30.9940) -- (398.0030,68.7917) -- (508.3720,98.2738);
			\path[draw=black, line width=0.25mm]
			(398.0030,68.7917) -- (507.2381,126.6220);
			\path[draw=red, line width=0.25mm]
			(501.5685,69.5476) -- (399.1369,91.8482) -- (504.9702,130.7798);
			\path[draw=black, line width=0.25mm]
			(400, 95)--(506, 153);
			\path[draw=black, line width=0.25mm]
			(508.3720,100.9196) -- (399.1369,115.2827);
			\path[draw=green, line width=0.25mm]
			(399.1369,115.2827)-- (504.9702,34.3958);
			\path[draw=black, line width=0.7mm]
			(357,140) -- (504.9702,134.5595);
			\path[draw=black, line width=0.25mm]
			(398.0030,68.7917) -- (397.6250,44.9792) .. controls (397.6250,44.9792) and
			(386.4256,58.9439) .. (386.3708,67.2662) .. controls (386.3109,76.3857) and
			(398.3810,91.8482) .. (398.3810,91.8482);
			\path[draw=black, line width=0.25mm]
			(399.1369,114.5268) .. controls (399.1369,114.5268) and (382.1188,100.7389) ..
			(382.1534,91.0916) .. controls (382.1853,82.1763) and (398.0030,69.5476) ..
			(398.0030,69.5476);
			\path[draw=black, line width=0.25mm] (502.6768,28.9893) .. controls (502.9135,29.6995) and (510.0151,31.2769) .. (512.8992,29.4556) .. controls (515.1433,28.0385) and (515.9904,22.1177) .. (515.9904,22.1177);
			\path[draw=black, line width=0.25mm] (501.1736,67.2121) .. controls (501.1736,67.2121) and (505.1954,64.1724) .. (507.5787,63.9860) .. controls (510.2840,63.7745) and (512.6221,66.6454) .. (515.3302,66.4734) .. controls (517.7018,66.3229) and (521.7882,63.4543) .. (521.7882,63.4543);
			\path[draw=black, line width=0.25mm] (512.0178,29.8482) .. controls (512.0178,29.8482) and (510.0419,31.8000) .. (509.9778,33.0692) .. controls (509.9051,34.5079) and (511.9104,36.9345) .. (511.9104,36.9345);
			\path[draw=black, line width=0.25mm] (504.9315,132.3843) .. controls (504.9315,132.3843) and (508.6573,133.3881) .. (510.2192,132.5690) .. controls (512.1499,131.5564) and (511.8207,128.0816) .. (513.7477,127.0620) .. controls (515.5479,126.1094) and (519.8556,127.2306) .. (519.8556,127.2306);
			\path[draw=black, line width=0.25mm]
			(356,115)--(398.0030,91.4702);
			\path[draw=black, line width=0.25mm]
			(336.8,70) -- (399.1369,115.6607);
			\path[draw=black, line width=0.25mm]
			(337,50) -- (397.6250,44.2232);
			\path[draw=black, line width=0.25mm] (330,60) ellipse
			(0.2cm and 0.9cm);
			\path[draw=black, line width=0.25mm] (350,130) ellipse
			(0.2cm and 0.9cm);
			\path[draw=black, line width=0.25mm] (510,2) ellipse
			(0.1400cm and 0.1900cm);
			\path [draw=black, line width=0.7mm] (505,1)--(334.5,35);
			\path [draw=black, line width=0.25mm] (505.5,4)--(397,45);
			\node at (312,50) {$R_1$};
			\node at (330,135) {$R_2$};
			\node at (325,10) {$\mathfrak R$};
			\node at (355, 65) {$\geq\lambda$};
			\node at (397.6250,0) {$H$};
			\node at (511.2069,-14) {$\mathbb{CC}(C)$};
			\node at (540.2069,130) {$<\lambda$};
			\draw [decorate,decoration={brace,amplitude=10pt,mirror,raise=4pt},yshift=0pt]
			(520,75) -- (520,25) node [black,midway,xshift=2cm] {\footnotesize
				private components};
			\path[draw=black, line width=0.25mm] (511.2069,28.7262) ellipse
			(0.2400cm and 0.2347cm);
			\path[draw=black, line width=0.25mm] (511.7738,66.5238) ellipse
			(0.2987cm and 0.2667cm);
			\path[draw=black, line width=0.25mm] (512.9077,99.4077) circle
			(0.1280cm);
			\path[draw=black, line width=0.25mm] (513.2857,132.1027) ellipse
			(0.2347cm and 0.2187cm);
			\path [draw=black, line width=0.25mm] (510, 160) ellipse 
			(0.2cm and 0.2187cm); 
			\path[draw=black, line width=0.25mm] (503,160)--(356, 150);
			\path[draw=green, line width=0.25mm] (505, 155)-- (400, 115);
			\path[fill=green, line width=0.25mm] (510,160) circle (0.2cm);
			\path[fill=blue, line width=0.25mm] (397.8140,45.1682) circle (0.09cm);
			\path[fill=red, line width=0.25mm] (398.5699,92.0372) circle
			(0.09cm);
			\path[fill=black, line width=0.25mm] (398.1920,68.9807) circle
			(0.09cm);
			\path[fill=green, line width=0.25mm] (398.9479,115.4717) circle
			(0.09cm);
			
			\path[scale=0.265,draw=black,fill=red,opacity=1,line width=0.460pt] (1911.7934,500.8181) .. controls (1908.9852,500.3657) and (1906.6082,499.9160) .. (1906.5110,499.8189) .. controls (1906.4138,499.7217) and (1906.5820,497.8442) .. (1906.8847,495.6468) .. controls (1908.7805,481.8836) and (1922.3867,470.7244) .. (1937.2721,470.7244) .. controls (1943.6197,470.7244) and (1952.6480,473.9307) .. (1956.8609,477.6811) .. controls (1958.1696,478.8461) and (1958.4226,479.2782) .. (1957.6451,479.0203) .. controls (1955.2450,478.2242) and (1943.3231,478.0190) .. (1940.7154,478.7289) .. controls (1936.7378,479.8118) and (1934.3690,482.7255) .. (1931.5658,489.9828) .. controls (1928.9052,496.8708) and (1927.5728,499.0432) .. (1925.2823,500.2277) .. controls (1922.8529,501.4840) and (1917.4066,501.7223) .. (1921.7934,500.8181) -- cycle;
			
			\path[scale=0.265,draw=black,fill=blue,opacity=1,line width=0.460pt] (1908.8647,112.6329) .. controls (1903.4410,111.6642) and (1899.1249,110.3203) .. (1898.3655,109.3637) .. controls (1897.6652,108.4814) and (1898.3789,102.5599) .. (1899.6450,98.7478) .. controls (1901.6511,92.7082) and (1907.3719,85.8141) .. (1913.2231,82.3851) .. controls (1919.4871,78.7142) and (1928.1992,77.1697) .. (1935.0449,78.5165) .. controls (1939.8348,79.4588) and (1946.4580,82.4405) .. (1946.4419,83.6471) .. controls (1946.4161,85.5770) and (1944.1944,94.5000) .. (1942.7102,98.6348) .. controls (1938.3940,110.6599) and (1933.7875,113.6234) .. (1919.5969,113.5044) .. controls (1915.2193,113.4761) and (1911.3898,113.0839) .. (1908.8647,112.6329) -- cycle;
			
			\path[scale=0.265,draw=black,fill=blue,opacity=1,line width=0.460pt] (1892.5706,247.9888) .. controls (1893.5229,239.1229) and (1898.3709,230.8089) .. (1905.9280,225.0813) .. controls (1918.7527,215.3616) and (1935.9588,213.9059) .. (1950.9339,221.2737) .. controls (1956.3759,223.9512) and (1962.9201,230.0252) .. (1966.0633,235.3159) .. controls (1968.8074,239.9349) and (1969.2152,239.2082) .. (1961.3753,243.6714) .. controls (1951.1978,249.4655) and (1946.0137,251.0823) .. (1941.0024,250.0253) .. controls (1939.4242,249.6925) and (1935.1422,247.9236) .. (1931.4869,246.0945) .. controls (1927.8316,244.2654) and (1923.3346,242.3310) .. (1921.4936,241.7957) .. controls (1915.2634,239.9846) and (1909.0785,241.7354) .. (1898.2928,248.3634) .. controls (1895.1612,250.2879) and (1892.4989,251.8625) .. (1892.3767,251.8625) .. controls (1892.2545,251.8625) and (1892.3417,250.1193) .. (1892.5706,247.9888) -- cycle;
			
			\path[scale=0.265,draw=black,fill=red,opacity=1,line width=0.460pt] (1922.6380,285.4352) .. controls (1907.0743,281.9585) and (1894.7851,270.2057) .. (1892.7881,256.8884) -- (1892.3867,254.2120) -- (1897.3317,250.9887) .. controls (1900.0514,249.2158) and (1904.5337,246.6925) .. (1907.2924,245.3812) .. controls (1911.7724,243.2517) and (1912.7677,242.9955) .. (1916.6122,242.9822) .. controls (1920.7971,242.9677) and (1921.1609,243.0848) .. (1929.7412,247.2098) .. controls (1937.3272,250.8567) and (1939.1058,251.4909) .. (1942.4098,251.7273) .. controls (1948.3427,252.1518) and (1954.5505,249.9338) .. (1964.7184,243.7566) .. controls (1966.7298,242.5346) and (1968.5872,241.6658) .. (1968.8461,241.8258) .. controls (1969.6148,242.3009) and (1970.5581,247.6379) .. (1970.5490,251.4603) .. controls (1970.5172,264.8148) and (1961.5734,277.2588) .. (1947.8582,283.0315) .. controls (1940.5489,286.1079) and (1930.1051,287.1033) .. (1922.6380,285.4352) -- cycle;	
			
			\path[scale=0.265,draw=black,fill=green,opacity=1,line width=0.460pt] (1924.2717,138.7525) .. controls (1917.5239,137.6251) and (1910.0664,133.5786) .. (1905.7963,128.7276) .. controls (1904.5635,127.3271) and (1902.6501,124.4380) .. (1901.5442,122.3075) .. controls (1899.5404,118.4470) and (1897.7486,112.7160) .. (1898.3556,112.1090) .. controls (1898.5313,111.9333) and (1900.5012,112.2956) .. (1902.7330,112.9141) .. controls (1908.7594,114.5841) and (1917.1662,115.4771) .. (1923.5492,115.1254) -- (1929.1993,114.8141) -- (1926.8821,118.1000) .. controls (1922.5060,124.3053) and (1922.9940,128.0168) .. (1929.4726,137.8026) -- (1930.5175,139.3808) -- (1928.8733,139.3137) .. controls (1927.9690,139.2768) and (1925.8982,139.0242) .. (1924.2717,138.7525) -- cycle;
		\end{tikzpicture}
	\end{center}
	\caption{Illustration of a possible intermediate stage in the proof of Theorem~\ref{theorem:lccd}. Colors represent the partial assignment~$f$, e.g., the two blue-colored sub-components are assigned to the blue vertex in $H$.
		Thick lines are edges that go from the sets of the body to their effective neighborhoods.}
	\label{figure::codmain}
\end{figure}

	\section{Applications of Balanced Crown Decomposition}
		In this section we present some applications of the balanced crown decomposition. The full proofs of the results of this section are given in Appendix~\ref{sec:appendix-applications}. 

For the problems \textsc{$W$-weight Separator} and \textsc{$W$-weight Packing}  we immediately get the following theorems  by reducing an instance $(G,k,W)$ by first finding a {\wccd} 
$(\chrf)$ 
of $G$,
and then applying the
standard crown reduction rule that removes the head $H$ and crown $C$ from $G$. 
We emphasize that the balanced connected partition of the body is crucial to obtain the kernel sizes. These are the first kernels for vertex-weighted graphs, while also improving the state-of-the-art results for the unweighted cases. 

\begin{restatable}{theorem}{WSepKernelThm}
	\textsc{$W$-weight Separator} admits a kernel of weight $3k\left(W-1\right)$. Furthermore, such a kernel can be computed in  time $\mathcal{O}\left(k^2|V|\,|E|\right)$.
\end{restatable}

\begin{restatable}{theorem}{WPackKernelThm}
	\textsc{$W$-weight Packing} admits a kernel of weight $3k\left(W-1\right)$. Furthermore, such a kernel can be computed in  time $\mathcal{O}\left(k^2|V|\,|E|\right)$.
\end{restatable}


For the optimization variant of the \textsc{$W$-weight Packing} problem, i.e.~maximizing the size of packing, the fact that the partition $\calR$ is a solution also gives a $3$-approximation;  to the best of our knowledge, the first approximation result for the problem.

\begin{restatable}{theorem}{WPackApxThm}
	A 3-approximation for the optimization problem of \textsc{$W$-weight Packing} can be computed in $\mathcal{O}\left({k^*}^2|V|\,|E|\right)$, where $k^*$ denotes the optimum value.
\end{restatable}

To better sketch the ideas for our results for the BCP problems, we denote by $I$-$\cvp_{k}$ for an interval $I$, a $\cvp$ with $k$ parts where each part has a weight in $I$.  We derive the following result for \textsc{Max-Min BCP}, which is the first constant approximation for this problem.

\begin{restatable}{theorem}{BCPMaxMinApxThm}
	\label{theorem::MaxMinAPX}
	A 3-approximation for the {\maxminbcp} problem can be computed in $\mathcal{O}\left(\log\left(X^*\right) k^2|V|\,|E|\right)$, where $X^*$ denotes the optimal value.
\end{restatable}
\begin{sketch}
	Let $(G,k)$ be an instance of {\maxminbcp}. 
	For any value $X$, using BCD, we show how to either obtain an $[X/3, \infty)$-{$\cvp_{k}$}, or report that $X>X^*$.
Once we have this procedure in hand, a binary search can be used to obtain an $[X^*/3, \infty)$-{$\cvp_{k}$}. 

We first obtain a {\ccd} {($\chrf$)} of $G$ with $\lambda=\lceil X/3\rceil$. 
If $|H|+|\calR|\ge k$, we output a $[X/3,\infty)$-$\cvp_{k}$ given by the body and the assignment to head vertices (if this gives more than $k$ sets, arbitrarily merge some until there are only $k$).
If $|H|+|\calR|<k$, then
we report that $X>X^*$.
To see that this is correct, assume towards contradiction that $X\le X^*$, and consider an optimal solution $\calS^* = \{S^*_1,\dots,S^*_k\}$. Then in the {\ccd} we computed, we know that $w(R) < X$ for every $R \in \calR$ and $w(C') < X$ for every $C' \in \comp(C)$.
Observe that then no $C' \in \comp(C)$ or a subset of it can be a set in $\calS^*$, since $w(S^*_i) \geq X^*\ge X$ for every $S_i^* \in \calS$.
From the separator properties of $H$ and that the fact that each $S_i^* \in \calS$ is connected, we obtain that any set in $\calS^*$ containing vertices from $C$ also has to contain at least one vertex from $H$. 
Thus, we can derive that the cardinality of $\calS^*_H = \{S^*_i \in \calS^*| S^*_i \cap (C \cup H) \neq \varnothing\}$ is at most $|H|$.
Also, $|\calS^*\setminus\calS^*_H|\le w(V(\calR))/ X^*\le w(V(\calR))/X\le |\calR|$.
Thus it follows that $|\calS^*|\le |H|+|\calR|<k$, a contradiction.
\end{sketch}

The last problem that we consider as application of the balanced crown decomposition is the {\minmaxbcp} problem,
where we also provide the first constant approximation result.
\begin{restatable}{theorem}{BCPMinMaxApxThm}
	A 3-approximation for the {\minmaxbcp} problem can be computed in time $\mathcal{O}\left(\log\left(X^* \right) |V|\,|E| \left(\log \log X^* \log\left(|V| w_{max} \right) + k^2 \right) \right)$, where $X^*$ denotes the optimum value and $w_{max} = \max_{v \in V} w(v)$  the maximum weight of a vertex.
\end{restatable}
\begin{sketch}
Achieving this requires several technical steps after having a balanced crown decomposition in hand, including a second use of the balanced expansion.
Let $(G,k)$ be an instance of  {\minmaxbcp} and let $\calS^* = \{S^*_1,\dots,S^*_k\}$ be an optimal solution.
Let $(\chrf)$ be a {\ccd} of $G$.
Similar to the Max-Min case we try to make a comparison between $\calS^*$ and the vertex decomposition $C,H$ and $V(\calR)$.
The main issue is that, in contrast to the Max-Min case, an optimal solution can (and sometimes has to) build more than $|H|$ sets from the vertices in $H\cup C$.
With the connectivity constraints, this means that some components in $G[C]$ are in fact a set in the optimal partition.
Hence, when computing an approximate solution from a balanced crown decomposition, we have to also choose some components from $G[C]$ to be sets, while others are combined with some vertex in $H$.
In order to make the decision of where to place the components in $G[C]$, we use a min-cost flow on a network that models the options for components in $G[C]$ to either be sets or be combined with some vertex in $H$.
A partial embedding of $\{S^*_i \in \calS^*| S^*_i \cap (C \cup H) \neq \varnothing\}$ to this cost-flow network allows a comparison with the resulting partition of $C \cup H$.  The balanced weight properties of $\calR$ then yield a comparison with the whole set $S^*$.
With the additional use of a min cost-flow network, our balanced crown structure can be used to estimate the optimal objective value and again enables a binary search for an approximate solution.
\end{sketch}

		\bibliography{literature}
		\appendix
	

\section{Preliminaries}
By $\Nz$ and  $\mathbb N$ we denote the natural numbers with and without zero, respectively. We also use $\left[k\right]$ to denote the set $\left\{1,\ldots,k\right\}$. A \emph{partial function} from $A$ to $B$ is a function $g\colon A'\to B$ where $A'\subseteq A$ and denote
such a partial function by $g\colon A'\pto B$ and $A'$ is the \emph{domain} of $g$. 
We often use a function $g$ that assigns vertex-sets to vertices, i.e. $g\colon 2^V\pto V$.
If $g(C)=v$ for some vertex-set $C$ and vertex $v$, then we call $C$  \emph{$g$-assigned} to $v$.
We also define the \emph{$g$-weight} of $v$ as $w^f[v]:=w(v)+w(\ginv(v))$.
We say that $C\in 2^V$ is \emph{$g$-assigned} if $C$ is in the domain of $g$ and \emph{$g$-unassigned} otherwise.
We define the \emph{$g$-neighborhood} of any vertex $v$ as $N^g[v]:=\left\{ v \right\}\cup \ginv(v)$.

\paragraph*{Graph theory terminology}
All the graphs that we refer to in this paper are simple and finite.
Consider a graph $G=(V,E)$.
For any subgraph $H$ of $G$, we use $V(H)$ and $E(H)$ to denote the vertices and edges of $H$ respectively. 
We denote an edge $e=\left\{u,v\right\}\in E$ by $uv$ and the neighborhood of a vertex $v\in V$ in $G$ by $N_G\left(v\right)=\left\{u\in V\mid uv\in E\right\}$. 
We may omit the subscript $G$ when the graph is clear from context.
Similarly we denote the neighborhood of a vertex set $V'\subseteq V$ in $G$ by $N_G\left(V'\right)$, that is $\bigcup_{v \in V'} N_G\left(v\right) \setminus V'$. Moreover, for a subset $A\subseteq V$ we denote the neighborhood of a vertex $v$ inside $A$ by $N_A(v)$, that is $N_G(v)\cap A$, and we denote by $N_A(B)$ the neighborhood of a vertex set $B\subseteq V$ inside $A$. 

We denote a \emph{vertex-weighted graph} as $G=\left(V,E,w\right)$ where $w$ is a function assigning integer weights to vertices $w\colon V\rightarrow \mathbb N$, and $V$ and $E$ are vertex and edge sets. We denote by $w_{\text{min}}$ and by $w_{\text{max}}$, $\text{min}_{v\in V}w\left(v\right)$ and $\text{max}_{v\in V} w\left(v\right)$ respectively. For any $V'\subseteq V$, we use $w(V')$ to denote the sum of weights of vertices in $V'$.
For a subgraph $H$ of $G$ we use $w\left(H\right)$ to denote $w(V(H))$, and refer to it as \emph{weight of the subgraph $H$}.
For a rooted tree $T$ and a vertex $x$ of $T$, we use $T(x)$ to denote the rooted subtree of $T$ rooted at $x$.

For $V'\subseteq V$ we denote by $G[V']$ the graph induced by $V'$, i.e. $G[V']=(V',E')$ with $E'=E\cap (V'\times V')$. For vertex-weighted graphs, induced subgraph inherits the vertex-weights. For $V'\subseteq V$ we also use $G-V'$ to denote $G[V\setminus V']$. Similarly, if $V'$ is a singleton $\left\{v\right\}$ we write $G-v$. 
With \emph{connected component of $G$}, we denote an inclusion maximal vertex set $V'\subset V$ such that $G[V']$ is connected.
For any $V'\subseteq V$, we denote by $\comp(V')$, the set of connected components of $G[V']$.
A vertex set $S\subset V$ whose removal separates $G$ into more than one connected component is called a \emph{separator} of $G$.
If $S=\{v\}$, then  we call $v$ a \textit{separator vertex}.

In the course of this paper we also use directed graphs in the form of networks. We denote a network $\Nw$ by $\left(V, \overrightarrow{E}, c\right)$, where $V$ is its vertex set, $\overrightarrow{E}$ the set of arcs, and $c\colon \overrightarrow{E}\rightarrow\mathbb N$ its capacity function. We denote an arc from vertex $u$ to vertex $v$ in a network by $\overrightarrow{uv}$, and denote by $\delta^+\left(v\right)$, $\delta^-\left(v\right)$, out- and in-coming arcs of vertex $v$.
The key characteristic of a network is that its vertex set contains the vertices $s$, $t$, referred to as \emph{source} and \emph{sink}, respectively. 
The source vertex is characterized by having only out-coming edges, while the sink vertex only has in-coming arcs. A \emph{flow} is denoted by $Y$, i.e. $Y =\left\{y_{\dir{e}} \in \mathbb{N}_0 \mid \dir{e} \in \dir{E}\right\}$, where $\dir{E}$ are the arcs in the corresponding network.

\paragraph*{Structural notions}
Let $\calU=\left\{ U_1, \dots, U_r \right\}$ be such that each $U_i\subseteq V$.  We call $\calU$  a \emph{connected packing} of $V$ if each $G[U_i]$ is connected, and the sets  in $\calU$ are pairwise disjoint. 
A connected packing $\calU$ is called \emph{connected vertex partition} (\cvp) of $V$, if $\calU$ is also a partition of $V$.
For any $\calU'\subseteq \calU$, we define $V\left(\mathcal{U}'\right) := \bigcup_{U' \in \mathcal{U}'} U'$, and the weight $w\left(\calU'\right):=w\left(V\left(\calU'\right)\right)$.
Let $\mathcal{I}$ be an interval.
If $\calU$ is a {\cvp} and $w\left(U_i\right) \in \mathcal{I}$ for all $i \in [r]$, then we say that $\mathcal{U}$ is a \emph{$\mathcal{I}$-connected vertex partition} ($\mathcal{I}$-\cvp) of $V$.
If $\calU$ is a {connected-packing} and $w\left(U_i\right) \in \mathcal{I}$ for all $i \in [r]$, then we say that $\mathcal{U}$ is a \emph{$\mathcal{I}$-connected packing} of $V$.
	For a vertex $v \in V$ and connected packing $\calU$, let $\calN_{\calU}(v)$ denote the set of all $U\in \calU$ such that $v\in N(U)$. 
	Similarly, for a vertex set $S$, let $\calN_{\calU}(S)$ denote the set of all $U\in \calU$ such that $S\cap N(C)\neq \varnothing$.

\paragraph*{Parameterized terminology}
We use the standard terminology for parameterized complexity, see e.g.~\cite{CFK+15,fomin2019kernelization} for details. A \emph{parameterized problem} $P$ is a decision problem equipped with a \emph{parameter}, i.e.~instances of $P$ are given as pairs $(I,k)$ where $k$ is the parameter. Such a parameterized problem is \emph{fixed-parameter tractable} if it can be solved in time $f(k)\cdot|I|^c$ for any  instance $(I,k)$ of $P$, where $f$ is a computable function and $c$ is a constant.

Of particular interest here are kernelizations - a formalization of preprocessing.
A \emph{kernelization} for a parameterized problem is an algorithm that  maps any instance $(I,k)$ of $P$ in polynomial time in $|I|$ to an instance $(I',k')$ of $P$ such that:
\begin{itemize} 
\item  $(I',k')$ is a yes-instance if and only if $(I,k)$ is a yes-instance,
\item $|I'|\leq g(k)$ and
\item $k'\leq g'(k)$, 
\end{itemize}
for some computable functions $g, g'$.
The instance  $(I',k')$ is called \emph{kernel} and  $g(k)$ is the \emph{kernel size}.  In case such a kernelization exists we say problem $P$ admits a  kernel of  size $g(k)$. It is known that a parameterized problem is fixed-parameter tractable if and only if it admits a kernel for some function $g$. 

We consider parameterization by \emph{combined parameters}, which formally means that our problems have instances of the form $(I,k_1,k_2)$ and we consider them as parameterized problem with parameter $k=k_1+k_2$. 

	\section{Balanced Expansion}
		\label{sec:appbalexp}

In this section we introduce a balanced generalization of weighted expansions we call \emph{balanced expansion}. It is one of the ingredients to derive our main BCD structure in the next section. 

With $G=\left(A \cup B, E, w\right)$, we denote bipartite vertex-weighted graphs, where $A$ and $B$ are independent sets in $G$ and $w\colon A\cup B\rightarrow \mathbb N$. Note that the weight bounds in our following structure depend on the maximum weight in $B$, denoted by $\wmax^B=\max_{b \in B} w\left(b\right)$.
	
\balancedExpansion*

An illustration of a balanced expansion is given in \cref{pic::MaxFlowNetwork}.
One can observe that this structure is a combination of a weighted $q$-expansion with a balanced partition
 of the remaining graph. 
 In the first step to compute this,
 we compute a fractional version of this structure. 
 Both the \expanse{} and its fractional version are substructures that we will employ later to find our balanced crown decomposition. Also, Lemma~\ref{lemma::WeightedExpansion} for efficiently computing a weighted expansion is a direct consequence of our routine to efficiently compute a \expanse, the proof can be found at the end of this section.

As fractional version of the \expanse, we formally consider the following structure.

\fracBalancedExp*
%
Note that such a \fexpanse{ }$\left(A_1,A_2,g,q\right)$ is indeed a fractional version: The function $g$ can be interpreted as a fractional assignment from $b \in B$ to $a \in N\left(b\right)$, and what we are ultimately aiming for in the non-fractional version is that  $g$ is indeed an assignment, i.e.~$g\left(ab\right)>0$ for some $a\in A$ and $b\in B$ implies $g\left(ab\right)=w\left(b\right)$ and $g\left(a'b\right)=0$ for all $a'\in A\setminus\left\{a\right\}$. 

Note that isolated vertices in~$B$ would imply that the assignment $f$ in the definition of a \expanse{ }cannot be a total function.  Further, such vertices may have a large weight without any merit for finding an \fexpanse{ } where $A_1\neq \varnothing$. 
We therefore restrict our further study to graphs without isolates in $B$; note that for the usual applications that exploit structures like \fexpanse, isolated vertices are removed in some form of preprocessing. For such graphs without isolates in $B$ we take advantage of the structure of maximal flows to derive the following result.
	
\fracBalancedExpLemma*	
	\begin{proof}
		Let $\ta$ be the set of vertices in $A$ with weight at least $q$ and let
		  $\widetilde{B} \subseteq B$ be the isolated vertices in $G-\widetilde{A}$, i.e.~${N}\left(\widetilde{B}\right) \subseteq \widetilde{A}$.  We add $\ta$ to $A_1$ and 
		for each $b\in \tb$, pick an arbitrary $a\in N\left(b\right)\subseteq \ta$	and set $g\left(ab\right) = w\left(b\right)$.
		We obtain $w\left(a\right) + \sum_{b \in B} g\left(ab\right) \geq q$ for all $a \in \widetilde{A}$
		 and $G - \left(\widetilde{A} \cup \widetilde{B}\right)$ has no isolated vertex $b \in B \setminus \widetilde{B}$.
		Since we end up in the same situation after removing $\widetilde{A}$ and $\widetilde{B}$, we may assume that $\max_{a \in A} w\left(a\right) < q$ and we have no isolated vertex in $b \in B$, i.e.~to simplify the notation we assume $A = A \setminus \widetilde{A}$ and $B = B \setminus \widetilde{B}$.
		Note that, if $w\left(A \setminus \widetilde{A}\right) + w\left(B \setminus \widetilde{B}\right) < q |A|$, then  $\widetilde{A}\neq \varnothing$ and thus $A_1 \ne \varnothing$.

\begin{figure}
	\begin{center}
		
		\begin{tikzpicture}[scale=0.8,y=0.80pt, x=0.80pt, inner sep=0pt, outer sep=0pt]
			\path[draw=black, line width=0.25mm] (0,60)--(98,60);	
			\path[draw=black, line width=0.55mm] (0,60)--(98,30);
			\path[draw=black, line width=0.55mm] (0,45)--(97,60);			
			\path[draw=black, line width=0.55mm] (0,30)--(97,59);
			\path[draw=black, line width=0.25mm] (0,30)--(98,30);
			\path[draw=black, line width=0.55mm] (0,15)--(98,30);
			\path[draw=black, line width=0.55mm] (0,0)--(98,-28);
			\path[draw=black, line width=0.25mm] (0,-30)--(98,58);
			\path[draw=black, line width=0.55mm] (0,-30)--(97,0);
			\path[draw=black, line width=0.25mm] (0,-15)--(98,30);
			\path[draw=black, line width=0.55mm] (0,-15)--(97,-31);

			\path[fill=black!60!, line width=0.25mm] (0,60) circle (0.09cm);
			\path[fill=black!60!, line width=0.25mm] (0,45) circle (0.09cm);
			\path[fill=black!60!, line width=0.25mm] (0,30) circle (0.09cm);
			\path[fill=black!60!, line width=0.25mm] (0,15) circle (0.09cm);
			\path[fill=black!60!, line width=0.25mm] (0,0) circle (0.09cm);
			\path[fill=black!60!, line width=0.25mm] (0,-15) circle (0.09cm);			
			\path[fill=black!60!, line width=0.25mm] (0,-30) circle (0.09cm);
			\path[fill=black, line width=0.25mm] (100,60) circle (0.09cm);
			\path[fill=black, line width=0.25mm] (100,30) circle (0.09cm);
			\path[fill=black, line width=0.25mm] (100,0) circle (0.09cm);			
			\path[fill=black, line width=0.25mm] (100,-30) circle (0.09cm);
			
			\node at (0,75) {$B$};
			\node at (100,75) {$A$};
			\node at (115,45) {{\large $\Bigg \}$}};
			\node at (130,45) {$A_1$};
			\node at (115,-15) {{\large $\Bigg \}$}};
			\node at (130,-15) {$A_2$};
			\node at (50,-50) {};
		\end{tikzpicture} \hspace*{2.4cm}
		\begin{tikzpicture}[scale=0.8,y=0.80pt, x=0.80pt, inner sep=0pt, outer sep=0pt]

			\path[draw=black, line width=0.25mm, -latex] (-100,15)--(-2,60);
			\path[draw=black, line width=0.25mm, -latex] (-100,15)--(-2,30);
			\path[draw=black, line width=0.25mm, -latex] (-100,15)--(-2,0);
			\path[draw=black, line width=0.25mm, -latex] (-100,15)--(-2,-30);
			\path[draw=black, line width=0.25mm, -latex] (-100,15)--(-2,45);
			\path[draw=black, line width=0.25mm, -latex] (-100,15)--(-2,15);
			\path[draw=black, line width=0.25mm, -latex] (-100,15)--(-2,-15);

			\path[draw=black, line width=0.25mm, latex-] (198,21)--(102,60);
			\path[draw=black, line width=0.25mm, latex-] (197,18)--(102,30);
			\path[draw=black, line width=0.25mm, latex-] (197,15)--(102,0);
			\path[draw=black, line width=0.25mm, latex-] (198,12)--(102,-30);

			\path[draw=black, line width=0.25mm, -latex] (0,60)--(98,60);	
			\path[draw=black, line width=0.25mm, -latex] (0,60)--(98,30);
			\path[draw=black, line width=0.25mm, -latex] (0,45)--(97,60);			
			\path[draw=black, line width=0.25mm, -latex] (0,30)--(97,59);
			\path[draw=black, line width=0.25mm, -latex] (0,30)--(98,30);
			\path[draw=black, line width=0.25mm, -latex] (0,15)--(98,30);
			\path[draw=black, line width=0.25mm, -latex] (0,0)--(98,-28);
			\path[draw=black, line width=0.25mm, -latex] (0,-30)--(98,58);
			\path[draw=black, line width=0.25mm, -latex] (0,-30)--(97,0);
			\path[draw=black, line width=0.25mm, -latex] (0,-15)--(98,30);
			\path[draw=black, line width=0.25mm, -latex] (0,-15)--(97,-31);

			\path[fill=black!60!, line width=0.25mm] (0,60) circle (0.09cm);
			\path[fill=black!60!, line width=0.25mm] (0,45) circle (0.09cm);
			\path[fill=black!60!, line width=0.25mm] (0,30) circle (0.09cm);
			\path[fill=black!60!, line width=0.25mm] (0,15) circle (0.09cm);
			\path[fill=black!60!, line width=0.25mm] (0,0) circle (0.09cm);
			\path[fill=black!60!, line width=0.25mm] (0,-15) circle (0.09cm);			
			\path[fill=black!60!, line width=0.25mm] (0,-30) circle (0.09cm);
			\path[fill=black, line width=0.25mm] (100,60) circle (0.09cm);
			\path[fill=black, line width=0.25mm] (100,30) circle (0.09cm);
			\path[fill=black, line width=0.25mm] (100,0) circle (0.09cm);			
			\path[fill=black, line width=0.25mm] (100,-30) circle (0.09cm);
			\path[fill=black, line width=0.25mm] (200,16) circle (0.09cm);
			\path[fill=black, line width=0.25mm] (-100,15) circle (0.09cm);	
			
			\node at (-100,30) {$s$};
			\node at (200,30) {$t$};
			\node at (0,75) {$B$};
			\node at (100,75) {$A$};
			\node at (-50,50) {$w(b)$};
			\node at (170,50) {$q-w(a)$};
			\node at (50,-45) {$w(b)$};
		\end{tikzpicture} 
		
	\end{center}
	\caption{Left: Balanced expansion for $w(b)=1$ for all $b\in B$, $q=2$ and assignment $f$ depicted with bold edges.  Right: Flow network construction in Lemma \ref{lemma::separatorPacking}.}
	\label{pic::MaxFlowNetwork}
\end{figure}
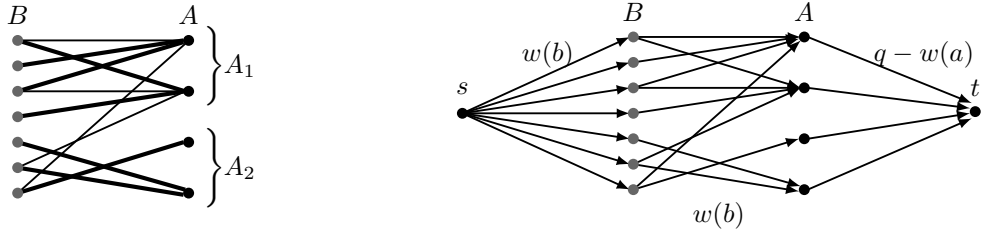

			We construct a network $\Nw = \left(A \cup B \cup \left\{s,t\right\}, \overrightarrow{E}, c\right)$ through the following procedure. 	In order to obtain $\Nw$ from $G$, add the vertices $s,t$, $s$ as source and $t$ as sink, and arcs $\overrightarrow{E}$ with a capacity function $c\colon\overrightarrow{E} \to \mathbb{N}$ defined as follows.
		Connect every $b \in B$ through an arc $\overrightarrow{sb}$ with capacity $w\left(b\right)$ and connect every $a \in A$ through an arc $\overrightarrow{at}$ with capacity $q - w\left(a\right)$. 
		Note, that $q-w\left(a\right)>0$ by $\max_{a \in A} w\left(a\right) < q$.
		Moreover, transform every edge $ab \in E$ to an arc $\overrightarrow{ba}$ with capacity $w\left(b\right)$ (see Figure~\ref{pic::MaxFlowNetwork}).
		Now, we use the max-flow algorithm for integral capacities by Orlin~\cite{orlin2013max} running in $\mathcal{O}\left(|V|\,|E|\right)$ to send a maximum flow from $s$ to $t$.
			Let $Y = \left\{y_{\overrightarrow{e}} \in \mathbb{N}_0 \mid \overrightarrow{e} \in \overrightarrow{E}\right\}$ be the resulting maximum flow and $Y^* := \sum_{a \in A} f\left(\overrightarrow{at}\right)$ be the total flow which is sent from $s$ to $t$.
		
		We set $g\left(ab\right) = y_{\dir{ba}}$ for all $\dir{ba} \in \dir{E}$.
		The capacities $c\left(\overrightarrow{ba}\right) = w\left(b\right)$ for all $b \in B$ and the flow conservation ensure that $\sum_{a \in A} g\left(ab\right) \leq w\left(b\right)$ for every $b \in B$.  
		
		Recall   the notation $B_a := \left\{b \in B\mid g\left(ab\right) > 0 \right\}$ for every $a \in A$, $B_{A'} := \bigcup_{a \in A'} B_a$ for a subset $A' \subseteq A$ and $B_U := \left\{b \in B\mid \sum_{a \in A} g\left(ab\right) < w\left(b\right)\right\}$ from Definition~\ref{definition::Separator-Packing}.	If $y_{\overrightarrow{at}} = c\left(\overrightarrow{at}\right)$ for all $a \in A$,
		then we set $A_1 = A$ and $A_2 = \varnothing$.
		In this case, it follows that  $w\left(a\right) + \sum_{b \in B} g\left(ab\right)  = w\left(a\right) + c\left(\overrightarrow{at}\right) = w\left(a\right) + q - w\left(a\right) = q$ for all $a \in A_1$, which satisfies the lower bound constraint given for the vertices in $A_1$.
		Since $A_2 = \varnothing$, we obtain $B_U \cup B_{A_1} = B$.
		Therefore, $N\left(B_U \cup B_{A_1}\right) = N\left(B\right) \subseteq A = A_1$.
		Thus, $\left(A_1,A_2=\varnothing,g,q\right)$ satisfies the conditions of a \fexpanse.
		Also, in this case we have $A_1 \neq \varnothing$.
		
		If $y_{\overrightarrow{sb}} = c\left(\overrightarrow{sb}\right) = w\left(b\right)$ for all $b \in B$ and there exists an $a \in A$ with $y_{\overrightarrow{at}} < c\left(\overrightarrow{at}\right)$, then we set $A_2 = A$ and $A_1 = \varnothing$.
		The flow conservation guarantees that every vertex $b \in B$ is saturated, i.e.  $B_U=\varnothing$.
		Since $A_1 = \varnothing$, we obtain $B_{A_1} = \varnothing$ and thus, ${N}\left(B_U \cup B_{A_1}\right) =  \varnothing = A_1$ satisfying the \emph{separator} condition.
		Moreover, we satisfy the upper bound constraint for all vertices in $A_2$.
		That is, $w\left(a\right) + \sum_{b \in B} g\left(ab\right) \leq w\left(a\right) + c\left(\overrightarrow{at}\right) = w\left(a\right) + q - w\left(a\right) = q$ for all $a \in A_2$.
		Thus, $\left(A_1=\varnothing,A_2,g,q\right)$ satisfies the conditions of a \fexpanse.
		Furthermore, in this case it follows that $w\left(A\right) + w\left(B\right) < \sum_{a \in A} \left(w\left(a\right) + c\left(\overrightarrow{at}\right)\right) = \sum_{a \in A} \left(w\left(a\right) + q - w\left(a\right)\right) = \sum_{a \in A} q = q|A|$, satisfying the second part of the lemma. 
		
		It only remains to consider the case where there exists an	$a \in A$ with $y_{\overrightarrow{at}} < c\left(\overrightarrow{at}\right)$ and  also a $b \in B$ with $y_{\overrightarrow{sb}} < c\left(\overrightarrow{sb}\right)$.
		Let $A'$ be the set of saturated vertices in $A$, i.e., $y_{\overrightarrow{at}} = c\left(\overrightarrow{at}\right)$ for all $a \in A'$.
		Set $A_1 = \varnothing$ and $A_2 = A \setminus A'$.
		We use the following procedure to add vertices of $A'$ either to $A_1$ or to $A_2$.
		First, observe that from the flow conservation for all $a \in A'$ we obtain $w\left(a\right) + \sum_{b \in B} g\left(ab\right) = w\left(a\right) + c\left(\overrightarrow{a_1t}\right) = w\left(a\right) + q - w\left(a_1\right) = q$.
		Thus, the lower/upper bound condition on a vertex of $A'$ are satisfied regardless of whether we place it in $A_1$ or $A_2$.

		To decide where to place a vertex $a \in A'$, we use properties of our computed maximum flow on $\Nw$. For more details regarding the properties of maximal flows we refer to the well-known Ford-Fulkerson-Algorithm \cite{ford2009maximal}.
		We construct an unweighted residual network $R= \left(V, \overrightarrow{E}_r \right)$ resulting from $Y$ as follows. Denote the residual flow by $r\left(\overrightarrow{e}\right):= c\left(\overrightarrow{e}\right) - y_{\overrightarrow{e}}$ for all $\overrightarrow{e} \in \overrightarrow{E}$.
		We introduce arcs in the reverse direction and set $r\left(\overleftarrow{e}\right) = y_{\overrightarrow{e}}$ for all $\overrightarrow{e} \in \overrightarrow{E}$.
		Now, we construct the residual graph $R$ by introducing an arc $\dir{e}$ for every strictly positive $r\left(\overrightarrow{e}\right)$.
		(Note that usually the arising residual weights are also assigned to the arcs in $R$ for the classical definition of a residual graph, but we do not need them for the following steps.)
		By the definition of the residual network, we know that if there exists an $s$-$t$ path in $R$, called usually augmenting path, then the flow in $Y$ can be altered to a flow of value $Y^*+1$ without violating any capacity constraint.
		Thus, by the optimality of $Y^*$ we obtain that no $s$-$t$ path exists in $R$.
		Moreover, we point out that $y_{\overrightarrow{at}} = c\left(\overrightarrow{at}\right)$ for all $a \in A'$ which leads to $\overrightarrow{at} \notin \overrightarrow{E}_r$ for those arcs.
				For each $a \in A'$, we consider the network  $R_a$ derived from $R$ by adding an arc $\overrightarrow{at}$.  Now we add $a$ to $A_1$ if an $s$-$t$ path exists in $R_a$, otherwise we add $a$ to $A_2$. Constructing the residual graph and checking whether an $s$-$t$ path exists in each $R_a$ for every $a \in A'$ can be realized in time $\mathcal{O}\left(|A'|\,|E|\right)$.

		It remains to show that $A_1$ is the desired separator. Consider any $a\in A'$.	If there exists an $s$-$t$ path in $R_a$, then this path has to use the added arc~$\overrightarrow{at}$; recall that $R$ did not contain such a path.	Looking at the original network $\Nw$, such an $s$-$t$ path in $R_a$ means that  if we increase the capacity $c\left(\overrightarrow{at}\right)$ in the network $\Nw$ by one, then we may increase the flow from $Y^*$ to $Y^* + 1$. For simplicity, denote by $N_a$ the network derived from $\Nw$ by increasing the capacity $c\left(\overrightarrow{at}\right)$ by one, and denote by $Y_a$ this maximum flow  from $s$ to $t$ in $N_a$ and by and $Y^*_a$ its value. Moreover, we denote the flow in $Y_a$ trough an arc $\dir{e}$ by $y^a_{\dir{e}}$. Note that   $y^a_{\overrightarrow{at}} = c\left(\overrightarrow{at}\right) + 1$ and objective value $Y_a^* = Y^* + 1$ if $a\in A_1$, and  $Y_a=Y$ with $Y_a^* = Y^*$ if $a\in A_2$.		
		
		We first prove that ${N}\left(A_2\right) \cap B_{A_1} = \varnothing$.
		Suppose there exists an edge $a_2b$ in $G$, where $a_2 \in A_2$ and $b \in B_{A_1}$.
		By definition of $B_{A_1}$, there exists an $a_1 \in A_1$ with $y_{\overrightarrow{ba_1}} > 0$ and by the definition of $A_1$ we have $Y^*_{a_1} = Y^* + 1$ with $y^{a_1}_{\overrightarrow{a_1t}} = c\left(\overrightarrow{a_1t}\right) + 1$.
		We claim that there exists a max flow $Y_{a_1}$ in $N_{a_1}$ such that $y^{a_1}_{\overrightarrow{ba_1}} > 0$.
		The flow value $y^{a_1}_{\overrightarrow{ba_1}}$ is smaller than $y_{\overrightarrow{ba_1}}$ only if $\overrightarrow{a_1b}$ is in the augmenting $s$-$t$ path $P$ for $Y$ in $N_{a_1}$.
		However, once we have reached $a_1$ in a path in $R_{a_1}$ we can extend this path with the arc $\overrightarrow{a_1t}$ to obtain an augmenting path.
		Thus, we can find a max flow $Y_{a_1}$ in $N_{a_1}$ that does not use $\overrightarrow{a_1b}$ in $R_{a_1}$ to improve the original flow $Y$ in $N_{a_1}$ and this ensures $y^{a_1}_{\overrightarrow{ba_1}} \geq y_{\overrightarrow{ba_1}} > 0$.		
		Now, we distinguish the two cases: $a_2 \notin A' \cap A_2$ or $a_2 \in A' \cap A_2$.
\begin{itemize}
\item 	 $a_2 \notin A' \cap A_2$: In this case we have $y_{\overrightarrow{a_2t}} < c\left(\overrightarrow{a_2t}\right)$ and $y^{a_1}_{\overrightarrow{a_2t}} < c\left(\overrightarrow{a_2t}\right)$, because the flow increases by one in $Y_{a_1}$ and only at the arc $\overrightarrow{a_1t}$ for the arcs in $\left\{ \overrightarrow{at}\mid a \in \delta^-\left(t\right)\right\}$ (recall that $Y$ was a maximum flow in $\Nw$, so an augmenting path for $N_{a_1}$ has to use the only new capacity on the arc $\overrightarrow{a_1t}$).		
		So, by $y^{a_1}_{\overrightarrow{ba_1}} > 0$ and $y^{a_1}_{\overrightarrow{a_2t}} < c\left(\overrightarrow{a_2t}\right)$ we may reduce the flow $y^{a_1}_{\overrightarrow{ba_1}}$ by one and could send one unit flow via the path $\overrightarrow{sb}, \overrightarrow{ba_2}, \overrightarrow{a_2t}$ without decreasing $Y^*_{a_1}$.
		Note, $y^{a_1}_{\overrightarrow{ba_1}} = c\left(\overrightarrow{ba_1}\right) + 1$ and after the new just mentioned routing operation we obtain $y^{a_1}_{\overrightarrow{ba_1}} = c\left(\overrightarrow{ba_1}\right)$.
		This yields a flow of value $Y^*_{a_1}$ that satisfies the capacity constraints in the original network $\Nw$,  which contradicts the optimality of $Y^*$, because $Y^*_{a_1} > Y^*$.
		
\item $a_2 \in A' \cap A_2$: In this case we have $y_{\overrightarrow{a_2t}} = c\left(\overrightarrow{a_2t}\right)$.
		Since there is no $s$-$t$-path in $R_a$, increasing the capacity by one at $\overrightarrow{a_2t}$ does not allow for a better flow, i.e. $Y^*_{a_2} = Y^*$ is a maximum flow in $N_a$.
		Again, we use the same argument, that we may reduce the flow $y^{a_1}_{\overrightarrow{ba_1}}$ by one and could send one unit of flow via the path $\overrightarrow{sb}, \overrightarrow{ba_2}, \overrightarrow{a_2t}$ without decreasing $Y^*_{a_1}$.
		As result, if the capacity $\overrightarrow{a_2t}$ is increased by one, we may find a maximal flow $Y^* + 1$ from $s$ to $t$ in $N_a$ which contradicts the optimality of $Y^*_{a_2}$.
		\end{itemize}	
		In both cases we show that  $b \in B_{A_1}$ cannot be in the neighborhood of an $a_2 \in A_2$.
		Thus, we obtain $N\left(A_2\right) \cap B_{A_1} = \varnothing$ and therefore ${N}\left(B_{A_1}\right) \subseteq A_1$.
		
		For the separator property of $A_1$, we also need to show that ${N}\left(B_U\right) \subseteq A_1$.
		Recall that $B_U = \left\{b \in B\mid \sum_{a \in A} g\left(ab\right) < w\left(b\right)\right\}$ are the unsaturated vertices in $B$.
		Thus, for every $b \in B_U$ we have $\sum_{a \in A} y_{\dir{ba}} < c\left(\dir{ba}\right) = w\left(b\right)$.
		Suppose there exists an arc with $\overrightarrow{ba_2}$ for $a_2 \in A_2$ and $b \in B_U$.
		If $y_{\overrightarrow{a_2 t}} < c\left(\overrightarrow{a_2 t}\right)$, then it is easy to see that we can improve $Y^*$, which would be a contradiction to the maximality of $Y$.
		If $y_{\overrightarrow{a_2 t}} = c\left(\overrightarrow{a_2 t}\right)$, then $a_2 \in A'$.
		Consequently, $Y^*_{a_2} = Y^*$  and $y^{a_2}_{\overrightarrow{a_2t}} = c\left(\overrightarrow{a_2t}\right)$ by the assigning of the vertices in $A'$ to $A_1$ or $A_2$.
		By definition, the capacity of $\dir{a_2t}$ in $N_{a_2}$ is $c\left(\overrightarrow{a_2t}\right) + 1$ and there is an arc $\overrightarrow{ba_2}$, where $\sum_{a \in A} y_{\dir{ba}} < c\left(\dir{ba}\right)$. Hence we may improve $Y_{a_2}^*$ by sending one unit of flow via the path $\overrightarrow{sb}, \overrightarrow{ba_2}, \overrightarrow{a_2t}$ without violating any capacity constraint in $N_a$ which is a contradiction to the optimality of $Y_{a_2}^*$.
		As result, ${N}\left(A_2\right) \cap B_U = \varnothing$ and therefore, ${N}\left(B_U\right) \subseteq A_1$.
		
		Finally, we have ${N}\left(B_{A_1}\right) \subseteq A_1$ and ${N}\left(B_U\right) \subseteq A_1$.
		Hence, ${N}\left(B_U \cup  B_{A_1}\right) \subseteq A_1$.
		As result, $\left(A_1,A_2,g,q\right)$ is a \fexpanse.
		
		Regarding the second part of the Lemma for the investigated last case we point out that no $b \in B$ is isolated.
		That is, if there exists a $b \in B$ with $y_{\overrightarrow{sb}} < c\left(\overrightarrow{sb}\right)$, then $N\left(b\right) \subseteq A'$.
		Otherwise, it would be a contradiction to the optimality $Y^*$, since $y_{\overrightarrow{at}} < c\left(\overrightarrow{at}\right)$ for all $a \in A \setminus A'$.
		Thus, in the case of $y_{\overrightarrow{sb}} < c\left(\overrightarrow{sb}\right)$ for a vertex $b \in B$ there exists at least one $a \in A'$ such that there exist an $s$-$t$ path in $R_a$ and consequently, we obtain $A_1 \ne \varnothing$.
		If $w\left(A\right) + w\left(B\right) \geq q|A| = \sum_{a \in A} \left(w\left(a\right) + c\left(\overrightarrow{at}\right)\right)$ and $Y$ does not satisfy $y_{\overrightarrow{at}} = c\left(\overrightarrow{at}\right) = q - w\left(a\right)$ for all $a \in A$, then there has to exist a $b \in B$ with $y_{\overrightarrow{sb}} < c\left(\overrightarrow{sb}\right)$ and therefore, we obtain $A_1 \ne \varnothing$.
		Note, that $y_{\overrightarrow{at}} = c\left(\overrightarrow{at}\right)$ for all $a \in A$ leads to $A_1 = A$, which we explained above as one of the extremal cases.		
		
		Lastly, we analyze the running time.
		Constructing $\Nw$ and computing the max-flow $Y$ in this network can be done in  $\mathcal{O}\left(|V|\,|E|\right)$.
		Defining $g$ requires time in $\mathcal{O}\left(|A|\,|B|\right)\subseteq\mathcal{O}\left(|V|\,|E|\right)$.
		As already argued above, the possible improvement of the flow $Y$ in $N_a$ for all $a \in A'$ together can be checked in time $\mathcal{O}\left(|A'|\,|E|\right)$.
		In total, we obtain a running time in $\mathcal{O}\left(|V|\,|E|\right)$.
	\end{proof}

Equipped with this procedure to efficiently compute a \fexpanse, we now round the function $g$ to a proper assignment $f$ from $B$ to $A$. The following rounding procedure to derive a {\expanse} from a \fexpanse{ }uses the well-known technique of cycle canceling, as used e.g. in~\cite{kumar20162}. While the lower bound for the weight of the vertices assigned to $A_1$ is more or less straight forward, the upper bound  for the weight of the vertices assigned to $A_2$ requires a different, more careful assigning strategy.

\balancedExpansionThm*

	\begin{proof}
		Recall the definitions  $B_{a} = \left\{b \in B\mid g\left(ab\right) > 0\right\}$  for $a \in A$  and $B_{A_i} = \bigcup_{a \in A_i} B_a$ for $i=1,2$. 
		Since  $\left(A_1,A_2,g,q\right)$ is a \fexpanse{ }we obtain   $\sum_{a \in A} g\left(ab\right) = 0$ for each $b \in B \setminus \left(B_{A_1} \cup B_{A_2}\right)$.
		Thus, these vertices are part of the unsaturated vertices $B_U = \left\{b \in B| \sum_{a \in A} g\left(ab\right) < w\left(b\right)\right\}$ and are connected to vertices in $A_1$.
		At first, we set  $f\left(b\right) = a$ for an arbitrary $b \in N\left(a\right)$ for each $b \in B \setminus \left(B_{A_1} \cup B_{A_2}\right)$. The vertices in $A_1$ have no upper bound condition.
		Thus, by this procedure we do not violate any of the desired conditions.
		To simplify the notation we set $B = B_{A_1} \cup B_{A_2}$.
		Note, that for the original set $B$ the vertices $B \setminus \left(B_{A_1} \cup B_{A_2}\right)$ are already assigned in a way that can not yield problems with $f$ being a \expanse.
		
		To determine the further assignments, we first construct an edge-weighted bipartite graph $G' = \left(A \cup B, E'\right)$ resulting from $g$, where $E'$ consists of the edges $ab\in E$ with $g\left(ab\right) > 0$, so $G'$ is a subgraph of $G$.
		We set $w_{ab} = g\left(ab\right)$ as edge weight for every $ab \in E'$.
		From the \fexpanse{ }$\left(A_1,A_2,g,q\right)$ we obtain the properties $\sum_{a \in N\left(b\right)} w_{ab} \leq w\left(b\right)$ for every $b \in B$, $w\left(a\right) + \sum_{b \in N\left(a\right)} w_{ab} \geq q$ for all $a \in A_1$ and $w\left(a\right) + \sum_{b \in N\left(a\right)} w_{ab} \leq q$ for all $a \in A_2$.		
		
		The next step is to modify the graph $G'$ and the edge weights $w_{ab}$ until $G'$ becomes acyclic without changing the size of the sum $w\left(a\right) + \sum_{b \in N\left(a\right)} w_{ab}$ for every $a \in A$ in $G'$
		and maintaining $\sum_{a \in N\left(b\right)} w_{ab} \leq w\left(b\right)$ for every $b \in B$.

		\begin{figure}[h]
			\begin{center}
				\def \globalscale {1.000000}
				\begin{tikzpicture}[y=0.80pt, x=0.80pt]
				\node at (71, 40) {$B$};
				\node at (30, -40) {$A_1$};

				\path[draw=black, line width=0.25mm] (0,30)--(15,-30);
				\path[draw=black, line width=0.25mm] (0,30)-- (45,-30);
				\path [draw=black, line width=0.25mm] (30, 30)--(15,-30);
				\path [draw=black, line width=0.25mm] (60,30)--(45,-30);
				
				\node at (110,-40) {$A_2$};
				\path[fill=black, line width=0.25mm] (80,30) circle (0.1cm);
				\path[fill=black, line width=0.25mm] (110, 30) circle (0.1cm);
				\path[fill=black, line width=0.25mm] (140,30) circle (0.1cm);
				\path[fill=black, line width=0.25mm] (80,-30) circle (0.1cm);
				\path[fill=black, line width=0.25mm] (110, -30) circle (0.1cm);
				\path[fill=black, line width=0.025mm] (140,-30) circle (0.1cm);
				
				\path[draw=black, line width=0.25mm] (80,30)--(110,-30);
				\path[draw=red, line width=0.25mm] (110,30)--(80,-30);
				\path[draw=red, line width=0.25mm] (80,-30)--(140,30);
				\path[draw=red, line width=0.25mm] (110,30)--(140,-30);
				\path [draw=red, line width=0.25mm] (140, 30)--(140,-30);
				
				\path[fill=black, line width=0.25mm] (0,30) circle (0.1cm);
				\path[fill=black, line width=0.25mm] (30,30) circle (0.1cm);
				\path[fill=black, line width=0.25mm] (60,30) circle (0.1cm);
				\path[fill=black, line width=0.25mm] (15,-30) circle (0.1cm);
				\path[fill=black, line width=0.25mm] (45, -30) circle (0.1cm);
				
				\path[fill=black, line width=0.25mm] (80,30) circle (0.1cm);
				\path[fill=black, line width=0.25mm] (110, 30) circle (0.1cm);
				\path[fill=black, line width=0.25mm] (140,30) circle (0.1cm);
				\path[fill=black, line width=0.25mm] (80,-30) circle (0.1cm);
				\path[fill=black, line width=0.25mm] (110, -30) circle (0.1cm);
				\path[fill=black, line width=0.025mm] (140,-30) circle (0.1cm);
				
				\node at (145, 0) {$7$};
				\node at (123, -10) {$2$};
				\node at (110, -10) {$5$};
				\node at (85,0) {$3$};
				
				\node at (271, 40) {$B$};
				\node at (230, -40) {$A_1$};

				\path[draw=black, line width=0.25mm] (200,30)--(215,-30);
				\path[draw=black, line width=0.25mm] (200,30)-- (245,-30);
				\path [draw=black, line width=0.25mm] (230, 30)--(215,-30);
				\path [draw=black, line width=0.25mm] (260,30)--(245,-30);		
				
				\node at (310,-40) {$A_2$};

				\path[draw=black, line width=0.25mm] (280,30)--(310,-30);
				\path[draw=red, line width=0.25mm] (310,30)--(280,-30);
				\path[draw=red, line width=0.25mm] (280,-30)--(340,30);
				\path[draw=red, dashed, line width=0.25mm] (310,30)--(340,-30);
				\path [draw=red, line width=0.25mm] (340, 30)--(340,-30);
				
				\path[fill=black, line width=0.25mm] (200,30) circle (0.1cm);
				\path[fill=black, line width=0.25mm] (230,30) circle (0.1cm);
				\path[fill=black, line width=0.25mm] (260,30) circle (0.1cm);
				\path[fill=black, line width=0.25mm] (215,-30) circle (0.1cm);
				\path[fill=black, line width=0.25mm] (245, -30) circle (0.1cm);
				
				\path[fill=black, line width=0.25mm] (280,30) circle (0.1cm);
				\path[fill=black, line width=0.25mm] (310, 30) circle (0.1cm);
				\path[fill=black, line width=0.25mm] (340,30) circle (0.1cm);
				\path[fill=black, line width=0.25mm] (280,-30) circle (0.1cm);
				\path[fill=black, line width=0.25mm] (310, -30) circle (0.1cm);
				\path[fill=black, line width=0.025mm] (340,-30) circle (0.1cm);
				
				\node at (285,0) {$5$};
				\node at (323, -10) {$0$};
				\node at (310, -10) {$3$};
				\node at (345, 0) {$9$};
				
				\end{tikzpicture}
			\end{center}
			\caption{Example of one iteration of the cycle canceling procedure.}
		\label{figure::cycle-cancel}
	\end{figure}
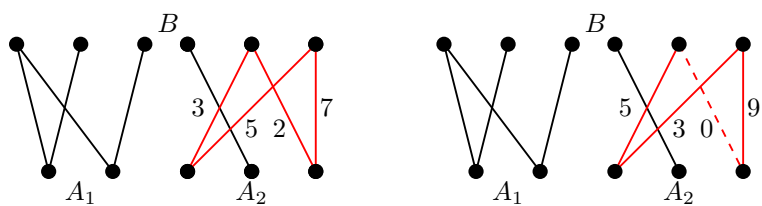 
		Suppose there exists a cycle $C \subseteq G'$.
		We pick an edge $ab \in E\left(C\right)$ with minimum weight $\min_{ab \in E\left(C\right)} w_{ab} = x$.
		We remove $C$ from $G'$ as follows. First, we reduce the weight of edge $ab$ by $x$. Then, traversing the cycle $C$ from $b$ in direction of its neighbor that is not $a$, we alternate between increasing and decreasing the weight of the visited edge by $x$ until we reach $a$.
		In the end we obtain at least one zero weight edge, namely $ab$, and we remove all edges of weight zero from $G'$. An example of such a cycle canceling step is given in Figure~\ref{figure::cycle-cancel}.
		Since $G'$ is bipartite, the number of edges in the cycle is even.
		Thus, we will maintain the weight $w\left(a\right) + \sum_{b \in N\left(a\right)} w_{ab}$ for all $a \in A$ during this process.
		Furthermore, we transfer weight on an edge incident to a $b \in B$ to another edge, which in turn is also incident to this vertex $b$.
		As result, the condition $\sum_{a \in N\left(b\right)} w_{ab} \leq w\left(b\right)$ for all $b \in B$ is still satisfied.
		We repeat this process until $G'$ becomes a forest.
		
		The next step is to determine $f\colon B \to A$ via the individual trees in the resulting forest.
		For this purpose we root every tree of the forest by a vertex $a \in A$.
		We know that $A_1$ separates $B_{A_1} \cup B_U$ from $B_{A_2} \cup A_2$ in $G$ and that vertices in $A_1$ do not have edges to $B_{A_2}$ in $G'$, because $E\left(G'\right)$ consists only of edges $ab \in E$ with $g\left(ab\right) > 0$ (which by definition puts $b\in B_{A_1}$ if $a\in A_1$).
		Since we only delete edges in the cycle canceling process, this fact does not change.
		Therefore, the vertices of $A_1$ and $A_2$ cannot be in the same tree. 
		
		We now define $f$ for the vertices in trees containing vertices from $A_1$ stepwise as follows.
		Let $T$ be a tree in the resulting acyclic graph that is rooted at an $a_r \in A_1$.
		We have $w\left(a\right) + \sum_{b \in N\left(a\right)} w_{ab} \geq q$ for every $a \in A_1$ in the tree.
		Since $G'$ is a bipartite graph, every child of $a \in A_1$ is in $b \in B$ and vice versa.
		We assign $f\left(b\right) = a$ for every child vertex $b \in B$ to its father $a \in A_1$ in $T$.
		Thus, for $a_r$ we obtain $w\left(a_r\right) + w\left(f^{-1}\left(a_r\right)\right) \geq w\left(a_r\right) + \sum_{b \in N\left(a_r\right)} w_{ab} \geq  q$.
		For all other vertices $a \in A_1$ that are in the tree, we only lose its father vertex $b$. Since all weights are integer, it follows that $w_{ab} \leq \wmax^B - 1$ for all weights $w_{ab}$ of those edges; note that $q \geq \wmax^B$ and $b$ in this case has at least two neighbors in $T$.
		As result, we obtain $w\left(a\right) + w\left(f^{-1}\left(a\right)\right) \geq w\left(a\right) + \sum_{b \in N\left(a\right)} w_{ab} - \left(w_{\max}^B-1\right) = q - \left(w_{\max}^B-1\right)$.
		
		Now, we construct the assignment $f$ for vertices in trees containing vertices from $A_2$.
		Let $T$ be a tree in the resulting acyclic graph that is rooted at an $a_r \in A_2$.
		Moreover, let $D \subset V\left(T\right)$ be the leaves, $D' = D \cap B$ and $A' \subseteq A$ the corresponding father vertices of $b \in D'$.
		At first, we assign every leaf $b \in D'$ to its corresponding father $a \in A'$ in $T$, so $f\left(b\right) = a$, and then remove those leaves from $T$.
		From the properties of the \fexpanse{ }$\left(A_1,A_2,g,q\right)$ we know that every vertex $b \in B_{A_2}$ is saturated, i.e. $\sum_{a \in N\left(b\right)} w_{ab} = w\left(b\right)$.
		Thus, for every leaf $b \in D'$ we obtain $w_{ab} = w\left(b\right)$, where $a$ is the father vertex of $b$. For this partial definition of $f$ we observe the following.
		If $a \in A'$ is not the root of $T$, then $w\left(a\right) + w\left(f^{-1}\left(a\right)\right) = w\left(a\right) + \sum_{b \in f^{-1}\left(a\right)} w\left(b\right) \leq w\left(a\right) + \sum_{b \in N\left(a\right)} w_{ab} - 1 \leq q - 1$.
		In the case $a \in A'$ is the root we obtain $w\left(a\right) + w\left(f^{-1}\left(a\right)\right) \leq q$.
		After this first round of assignments, every leaf is in the vertex set $A_2$.
		We continue assigning every father vertex $b$ to an arbitrary child vertex $a$.
		This way, we add at most one vertex $b$ for every $A_2 \setminus \left\{a_r\right\}$ in the tree $T$, which has a father vertex in $T$.
		As result, we obtain $w\left(a\right) + w\left(f^{-1}\left(a\right)\right) \leq q + w_{\max}^B - 1$ for every $a \in A_2$ in the tree $T$.
		
		Overall, we assign every $b \in B$ to an $a \in A$ via an existing edge in the corresponding tree. Since we only removed edges to obtain the forest from $G'$ and $E\left(G'\right) \subseteq E\left(G\right)$, the assignment of every $b \in B$ satisfies $f\left(b\right) \in N\left(b\right)$ in $G$. 
		
		Finally, we analyze the running time.
		The cycle canceling process is applied at most $|E|$ times and that each step can executed in time $\mathcal{O}\left(|V|\right)$ in the following way.
		We compute a spanning forest $T$ in $G'$, which can be done in time $\mathcal{O}\left(|E|\right)$ by breadth first search. Then, for each $e \in E\left(G'\right) \setminus E\left(T\right)$, we add $e$ to $T$ and check if this yields a cycle in $T$. If this produces a cycle, we cancel it by the process described above and delete  the arising zero edges from both $E\left(G'\right)$ and $E\left(T\right)$. This searching for and deleting can be done in $\mathcal O\left(|E\left(T\right)|\right) = O\left(|V\left(T\right)|\right)$ for each cycle.
		Since there are at most $|E|$ edges in $E\left(G'\right) \setminus E\left(T\right)$, this process can be performed in time $\mathcal{O}\left(|V|\,|E|\right)$ in total. 
	\end{proof}

Before moving to finding a balanced crown decomposition in a general graph in the next section, we formally state the aforementioned implications of the results in this section on (non-balanced) expansions. 
A \emph{weighted $q$-expansion} as defined in~\cite{fomin2019kernelization} (Chapter 5.3) 
on $G=(A\cup B,E,w)$ 
is a function $f\colon C\rightarrow H$ where $H\subseteq A$ and $C\subseteq B$, such that $f\left(b\right)\in N\left(b\right)$ for each $b\in C$ and $w\left(f^{-1}\left(a\right)\right)\geq q-W+1$ for each $a\in H$,
where $W$ is the largest vertex weight in $B$.
%
\WeightedExpansion*
\begin{proof}
First of all, we define a weight function $w'$ on $A\cup B$ by $w'\left(b\right)=w\left(b\right)$ for $b\in B$ and by setting $w'\left(a\right)=1$ for all $a\in A$. This way, we arrive at a vertex-weighted bipartite graph $G'=\left(A\cup B,E,w'\right)$ without isolated vertices in $B$ as considered in Lemma~\ref{lemma::separatorPacking} and Theorem~\ref{lemma::cycleCanceling}.
 Observe that with this setting on the vertices in $A$, the condition $w\left(B\right) \geq q|A|$ yields $w'\left(A\right) + w'\left(B\right) \geq \left(q+1\right)|A|$, hence we can compute for  $q'=q+1$ by  Lemma~\ref{lemma::separatorPacking} a \fexpanse{ }$\left(A_1,A_2,g,q'\right)$ for $G'$  with $A_1\neq \varnothing$. 

To apply \cref{lemma::cycleCanceling}, we need the restriction $q' \geq \max_{b \in B} w\left(b\right) = \wmax^B$, where $\wmax^B=W$ here. However, if this property is not satisfied, it follows that $q'\leq W-1$, hence $q=q'-1\leq W-2$, which turns the requirement of the $q$-expansion to $w\left(f^{-1}\left(a\right)\right)\geq -1$. Hence if $q'\leq W-1$, by choosing $H=A$ and $C=B$ any arbitrary function that maps $f\left(b\right)\in N\left(b\right)$ (note that $B$ contains no isolates) gives the claimed $q$-expansion (with  $N\left(C\right)\subseteq H$ trivially satisfied). Such an arbitrary function can be computed in $\mathcal{O}\left(|V|\right)$.

 If $q \geq \max_{b \in B} w\left(b\right) = \wmax^B$, we can use \cref{lemma::cycleCanceling}  to compute in time $\mathcal{O}\left(|V|\,|E|\right)$  an assignment function $f$ such that  $\left(A_1,A_2,f,q'\right)$ is an \expanse{} for $G'$. In particular, with $H:=A_1$ and $C:=f^{-1}\left(A_1\right)$, this function restricted to $C$ gives the desired  $q$-expansion with  $N\left(C\right)\subseteq H$, since:\begin{itemize}
 \item for each $a\in H=A_1$ condition 1 of a \expanse{} gives $w'\left(a\right)+w'\left(f^{-1}\left(a\right)\right)\geq q'- \wmax^B+1$ which translates to  $w\left(f^{-1}\left(a\right)\right)\geq q- \wmax^B+1$ by the definition of $q'$ and $w'$,
 \item condition 2 of a \expanse{} yields $f\left(b\right)\in N\left(b\right)$,
 \item condition 3 of a \expanse{} yields $N\left(f^{-1}\left(A_1\right)\right)\subseteq A_1$, i.e. $N\left(C\right)\subseteq H$.
 \end{itemize}
The claimed running time follows by Lemma~\ref{lemma::separatorPacking} and Theorem~\ref{lemma::cycleCanceling}.
\end{proof}

	\section{Balanced Crown Decomposition}
		\label{sec:lccd}

This section is devoted to proving Theorem~\ref{theorem:lccd}. Before giving the detailed steps to prove this theorem, we give a few definitions followed by a slightly informal road-map that explains the structure of the main algorithm.

Here we give an outline of the algorithm {\algCrown} for finding a {\ccd} ($\chrfstar$) of $G$.
We need to first define some intermediate structures that we construct during the algorithm.
Let $H \subseteq V$,
let $\calC$ be a connected packing
of $V \setminus H$, and $\calR$ be a {\cvp} of $V \setminus \left(H \cup V\left(\calC\right)\right)$.
Furthermore, let $g: \calC \pto H$ be a partial function, i.e.~not any $C \in \calC$ is assigned to some $h \in H$, denoted as $"\pto"$.
We say that $(\chrg)$ is a \textbf{$\lambda$-component decomposition} ({\cod}) of $G$ if (see \cref{figure::cod}):
\begin{enumerate}
	\item $w(Q)< \lambda$ for all $Q\in \comp(V(\calC))$ \label{prop:Qlambda}
	\hfill (component condition)
	\item $g(C) \in N(C)$ for every $C \in \ginv\left(H\right) $  \label{prop:gnbh}
	\hfill ($g$-neighbor condition)
	\item $\calR$ is a $[\lambda, \infty)$-{\cvp} of $V \setminus\left(H \cup V\left(\calC\right)\right)$ \label{prop:Rweight}
	\hfill ($\calR$-condition)
\end{enumerate}

Let $(\chrg)$ be a {\cod}.
Let $\calQ := \comp(V(\calC))$.
Note that each $C\in \calC$ is completely contained in some $Q \in \calQ$ because $\calC$ is a {\cvp} of $C$.
We call $C \in \calC$ a \emph{sub-component} (of the connected-component $Q$). 
We denote the set of sub-components of $Q$ by $\calC(Q)$. 
Conversely, the unique connected component $Q \in \calQ$ of which $C$ is a part is denoted by $\calQ(C)$.
Note that each sub-component $C\in \calC$ has $w(C) \leq w(\calQ(C)) < \lambda$ due to condition (1) in the definition of {\cod}.
Let $\qpriv\subseteq \calQ$ (see \cref{figure::cod}) denote the set of all $Q\in\calQ$ that do not have edges to $V\left(\calR\right)$, 
and let $\cCpriv$ denote the set of all sub-components $C\in \calC$ with $\calQ(C)\in \qpriv$.
We call $\qpriv$, the \emph{private components} and $\cCpriv$, the \emph{private sub-components}.
Let $g':H \to\calC\cup \left\{ \varnothing \right\}$ be a function. 
For any $h\in H$, 
we define its \emph{$g$-weight} as $w^g[h]:=w(\ginv(h))$
and its \emph{$\left(g+g'\right)$-weight} as $\wgg[h]:=w^g[h]+w\left(g'(h)\right)$.
For an vertex $v \in V$ we denote by $\calN_{\calC}(v)$ the set of all $C \in \calC$ such that $\{v\} \cap N(C) \neq \varnothing$.
Similarly, for a vertex set $S \subseteq V$ we denote by $\calN_{\calC}(S)$ the set of all $C \in \calC$ such that $S \cap N(C) \neq \varnothing$.
We say that $(\chrgg)$ is a \textbf{Partially Balanced {\cod} ({\pbcod})} of $G$ if $(\chrg)$ is a {\cod} of $G$ and for all $h\in H$:

\begin{enumerate}
	\item if $g'(h)\neq \varnothing$, then $g'(h) \in \calN_{\calC}(h) \setminus \ginv(h)$,
	\hfill ($g'$-neighbor condition)
	\label{item:gdnbh}
	\item $w^g[h]\in [2\lambda -1,3\lambda - 3]$
	\hfill ($g$-weight condition)
	\label{item:wgh}
	
	\item if $\ginv(h)\nsubseteq \cCpriv$, then $w^{g+g'}[h]\ge 3\lambda-2$, and
	\hfill ($(g+g')$-weight condition)
	\label{item:gg1}
	
	\item if $w^{g}(h)<2.5 (\lambda-1)$, then $w(C)\ge 0.5 (\lambda-1)$ for every $C\in \ginv(h)\setminus\cCpriv$.
	\hfill ($0.5 (\lambda-1)$-condition)
	\label{item:gg2}
	
\end{enumerate}

For an $R \in \calR$, we define the \emph{unassigned neighborhood} $\nu(R):=\ncalC(R)\setminus \ginv(H)$, the \emph{effective neighborhood} $\nef[R]:=R\cup V\left(\nu(R)\right)$
and the \emph{effective weight} $w^{\bar{g}}(R) := w(\nef[R])$.
For an $h \in H$ we define $N^g[h] := \{h\} \cup V\left(\ginv(h)\right)$.
Lastly, we say that a {\pbcod} $(\chrgg)$ is a \textbf{Fully Balanced {\cod} (\fbcod)} if $w^{\bar{g}}(R) \leq 3(\lambda - 1)$ for every $R \in \calR$.
\begin{figure}[h]
	\begin{center}
		\begin{tikzpicture}[y=0.80pt, x=0.80pt, yscale=-1.000000, xscale=1.000000, inner sep=0pt, outer sep=0pt]
			\path[draw=black, line width=0.7mm] (510.2,103)--(356.5,140);
			\path[draw=blue, line width=0.25mm]
			(397.6250,44.6012) -- (502.7024,27.2143);
			\path[draw=blue, line width=0.25mm]
			(397.6250,44.9792) -- (501.1905,63.8780);
			\path[draw=black, line width=0.25mm]
			(396,45)--(509, 97);
			\path[draw=black, line width=0.25mm]
			(503.0804,30.9940) -- (398.0030,68.7917) -- (508.3720,98.2738);
			\path[draw=black, line width=0.25mm]
			(398.0030,68.7917) -- (507.2381,126.6220);
			\path[draw=red, line width=0.25mm]
			(501.5685,69.5476) -- (399.1369,91.8482) -- (504.9702,130.7798);
			\path[draw=black, line width=0.25mm]
			(400, 95)--(506, 153);
			\path[draw=black, line width=0.25mm]
			(508.3720,100.9196) -- (399.1369,115.2827);
			\path[draw=green, line width=0.25mm]
			(399.1369,115.2827)-- (504.9702,34.3958);
			\path[draw=black, line width=0.7mm]
			(357,140) -- (504.9702,134.5595);
			\path[draw=black, line width=0.25mm]
			(398.0030,68.7917) -- (397.6250,44.9792) .. controls (397.6250,44.9792) and
			(386.4256,58.9439) .. (386.3708,67.2662) .. controls (386.3109,76.3857) and
			(398.3810,91.8482) .. (398.3810,91.8482);
			\path[draw=black, line width=0.25mm]
			(399.1369,114.5268) .. controls (399.1369,114.5268) and (382.1188,100.7389) ..
			(382.1534,91.0916) .. controls (382.1853,82.1763) and (398.0030,69.5476) ..
			(398.0030,69.5476);
			\path[draw=black, line width=0.25mm] (502.6768,28.9893) .. controls (502.9135,29.6995) and (510.0151,31.2769) .. (512.8992,29.4556) .. controls (515.1433,28.0385) and (515.9904,22.1177) .. (515.9904,22.1177);
			\path[draw=black, line width=0.25mm] (501.1736,67.2121) .. controls (501.1736,67.2121) and (505.1954,64.1724) .. (507.5787,63.9860) .. controls (510.2840,63.7745) and (512.6221,66.6454) .. (515.3302,66.4734) .. controls (517.7018,66.3229) and (521.7882,63.4543) .. (521.7882,63.4543);
			\path[draw=black, line width=0.25mm] (512.0178,29.8482) .. controls (512.0178,29.8482) and (510.0419,31.8000) .. (509.9778,33.0692) .. controls (509.9051,34.5079) and (511.9104,36.9345) .. (511.9104,36.9345);
			\path[draw=black, line width=0.25mm] (504.9315,132.3843) .. controls (504.9315,132.3843) and (508.6573,133.3881) .. (510.2192,132.5690) .. controls (512.1499,131.5564) and (511.8207,128.0816) .. (513.7477,127.0620) .. controls (515.5479,126.1094) and (519.8556,127.2306) .. (519.8556,127.2306);
			\path[draw=black, line width=0.25mm]
			(356,115)--(398.0030,91.4702);
			\path[draw=black, line width=0.25mm]
			(336.8,70) -- (399.1369,115.6607);
			\path[draw=black, line width=0.25mm]
			(337,50) -- (397.6250,44.2232);
			\path[draw=black, line width=0.25mm] (330,60) ellipse
			(0.2cm and 0.9cm);
			\path[draw=black, line width=0.25mm] (350,130) ellipse
			(0.2cm and 0.9cm);
			\path[draw=black, line width=0.25mm] (510,2) ellipse
			(0.1400cm and 0.1900cm);
			\path [draw=black, line width=0.7mm] (505,1)--(334.5,35);
			\path [draw=black, line width=0.25mm] (505.5,4)--(397,45);
			\node at (312,50) {$R_1$};
			\node at (330,135) {$R_2$};
			\node at (325,10) {$\mathfrak R$};
			\node at (355, 65) {$\geq\lambda$};
			\node at (397.6250,0) {$H$};
			\node at (511.2069,-14) {$\mathcal Q$};
			\node at (540.2069,130) {$<\lambda$};
			\draw [decorate,decoration={brace,amplitude=10pt,mirror,raise=4pt},yshift=0pt]
			(520,75) -- (520,25) node [black,midway,xshift=0.8cm] {\footnotesize
				$\qquad\mathcal Q^{priv}$};
			\path[draw=black, line width=0.25mm] (511.2069,28.7262) ellipse
			(0.2400cm and 0.2347cm);
			\path[draw=black, line width=0.25mm] (511.7738,66.5238) ellipse
			(0.2987cm and 0.2667cm);
			\path[draw=black, line width=0.25mm] (512.9077,99.4077) circle
			(0.1280cm);
			\path[draw=black, line width=0.25mm] (513.2857,132.1027) ellipse
			(0.2347cm and 0.2187cm);
			\path [draw=black, line width=0.25mm] (510, 160) ellipse 
			(0.2cm and 0.2187cm); 
			\path[draw=black, line width=0.25mm] (503,160)--(356, 150);
			\path[draw=green, line width=0.25mm] (505, 155)-- (400, 115);
			\path[fill=green, line width=0.25mm] (510,160) circle (0.2cm);
			\path[fill=blue, line width=0.25mm] (397.8140,45.1682) circle (0.09cm);
			\path[fill=red, line width=0.25mm] (398.5699,92.0372) circle
			(0.09cm);
			\path[fill=black, line width=0.25mm] (398.1920,68.9807) circle
			(0.09cm);
			\path[fill=green, line width=0.25mm] (398.9479,115.4717) circle
			(0.09cm);
			
			\path[scale=0.265,draw=black,fill=red,opacity=1,line width=0.460pt] (1911.7934,500.8181) .. controls (1908.9852,500.3657) and (1906.6082,499.9160) .. (1906.5110,499.8189) .. controls (1906.4138,499.7217) and (1906.5820,497.8442) .. (1906.8847,495.6468) .. controls (1908.7805,481.8836) and (1922.3867,470.7244) .. (1937.2721,470.7244) .. controls (1943.6197,470.7244) and (1952.6480,473.9307) .. (1956.8609,477.6811) .. controls (1958.1696,478.8461) and (1958.4226,479.2782) .. (1957.6451,479.0203) .. controls (1955.2450,478.2242) and (1943.3231,478.0190) .. (1940.7154,478.7289) .. controls (1936.7378,479.8118) and (1934.3690,482.7255) .. (1931.5658,489.9828) .. controls (1928.9052,496.8708) and (1927.5728,499.0432) .. (1925.2823,500.2277) .. controls (1922.8529,501.4840) and (1917.4066,501.7223) .. (1921.7934,500.8181) -- cycle;
			
			\path[scale=0.265,draw=black,fill=blue,opacity=1,line width=0.460pt] (1908.8647,112.6329) .. controls (1903.4410,111.6642) and (1899.1249,110.3203) .. (1898.3655,109.3637) .. controls (1897.6652,108.4814) and (1898.3789,102.5599) .. (1899.6450,98.7478) .. controls (1901.6511,92.7082) and (1907.3719,85.8141) .. (1913.2231,82.3851) .. controls (1919.4871,78.7142) and (1928.1992,77.1697) .. (1935.0449,78.5165) .. controls (1939.8348,79.4588) and (1946.4580,82.4405) .. (1946.4419,83.6471) .. controls (1946.4161,85.5770) and (1944.1944,94.5000) .. (1942.7102,98.6348) .. controls (1938.3940,110.6599) and (1933.7875,113.6234) .. (1919.5969,113.5044) .. controls (1915.2193,113.4761) and (1911.3898,113.0839) .. (1908.8647,112.6329) -- cycle;
			
			\path[scale=0.265,draw=black,fill=blue,opacity=1,line width=0.460pt] (1892.5706,247.9888) .. controls (1893.5229,239.1229) and (1898.3709,230.8089) .. (1905.9280,225.0813) .. controls (1918.7527,215.3616) and (1935.9588,213.9059) .. (1950.9339,221.2737) .. controls (1956.3759,223.9512) and (1962.9201,230.0252) .. (1966.0633,235.3159) .. controls (1968.8074,239.9349) and (1969.2152,239.2082) .. (1961.3753,243.6714) .. controls (1951.1978,249.4655) and (1946.0137,251.0823) .. (1941.0024,250.0253) .. controls (1939.4242,249.6925) and (1935.1422,247.9236) .. (1931.4869,246.0945) .. controls (1927.8316,244.2654) and (1923.3346,242.3310) .. (1921.4936,241.7957) .. controls (1915.2634,239.9846) and (1909.0785,241.7354) .. (1898.2928,248.3634) .. controls (1895.1612,250.2879) and (1892.4989,251.8625) .. (1892.3767,251.8625) .. controls (1892.2545,251.8625) and (1892.3417,250.1193) .. (1892.5706,247.9888) -- cycle;
			
			\path[scale=0.265,draw=black,fill=red,opacity=1,line width=0.460pt] (1922.6380,285.4352) .. controls (1907.0743,281.9585) and (1894.7851,270.2057) .. (1892.7881,256.8884) -- (1892.3867,254.2120) -- (1897.3317,250.9887) .. controls (1900.0514,249.2158) and (1904.5337,246.6925) .. (1907.2924,245.3812) .. controls (1911.7724,243.2517) and (1912.7677,242.9955) .. (1916.6122,242.9822) .. controls (1920.7971,242.9677) and (1921.1609,243.0848) .. (1929.7412,247.2098) .. controls (1937.3272,250.8567) and (1939.1058,251.4909) .. (1942.4098,251.7273) .. controls (1948.3427,252.1518) and (1954.5505,249.9338) .. (1964.7184,243.7566) .. controls (1966.7298,242.5346) and (1968.5872,241.6658) .. (1968.8461,241.8258) .. controls (1969.6148,242.3009) and (1970.5581,247.6379) .. (1970.5490,251.4603) .. controls (1970.5172,264.8148) and (1961.5734,277.2588) .. (1947.8582,283.0315) .. controls (1940.5489,286.1079) and (1930.1051,287.1033) .. (1922.6380,285.4352) -- cycle;	
			
			\path[scale=0.265,draw=black,fill=green,opacity=1,line width=0.460pt] (1924.2717,138.7525) .. controls (1917.5239,137.6251) and (1910.0664,133.5786) .. (1905.7963,128.7276) .. controls (1904.5635,127.3271) and (1902.6501,124.4380) .. (1901.5442,122.3075) .. controls (1899.5404,118.4470) and (1897.7486,112.7160) .. (1898.3556,112.1090) .. controls (1898.5313,111.9333) and (1900.5012,112.2956) .. (1902.7330,112.9141) .. controls (1908.7594,114.5841) and (1917.1662,115.4771) .. (1923.5492,115.1254) -- (1929.1993,114.8141) -- (1926.8821,118.1000) .. controls (1922.5060,124.3053) and (1922.9940,128.0168) .. (1929.4726,137.8026) -- (1930.5175,139.3808) -- (1928.8733,139.3137) .. controls (1927.9690,139.2768) and (1925.8982,139.0242) .. (1924.2717,138.7525) -- cycle;
		\end{tikzpicture}
	\end{center}
	\caption{A demonstration of a {\cod} ($\chrg$). 
		Here, $\calQ$ is the set of connected components of $G[V(\calC)]$.
		The further divisions of sets in $Q$ are the sub-components, i.e the elements of $\calC$.
		The assignment $g$ is represented by the colors, e.g., the two blue-colored sub-components are assigned to the blue vertex in $H$.
		The thick lines are edges that go from $R_1$ and $R_2$ to their unassigned neighborhoods $\nu(R_1)$ and $\nu(R_2)$ respectively. The second and third sets are in $\qpriv$ and all the sub-components of those two components are in $\cCpriv$.}
	\label{figure::cod}
\end{figure}
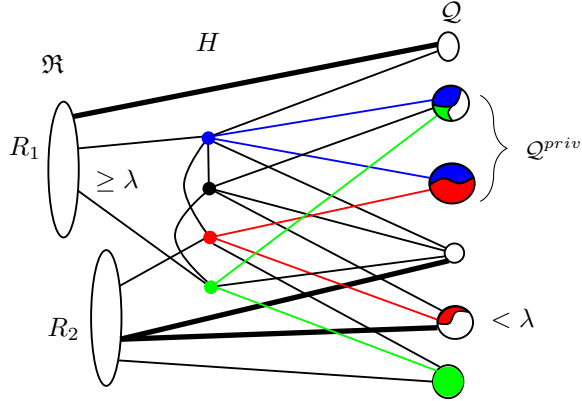

\subsection{Outline of Algorithm {\algCrown}}
\label{sub:algorithm_overview}
Note that the algorithm stated here is slightly informal at times, and is intended as a road-map to piece together the rigorous technical description that follows.

During the algorithm we always maintain $(C^*,H^*,f^*)$ 
to be a Weighted Crown Decomposition (WCD), i.e. it
satisfies conditions 1-4 of a {\ccd} (\cref{def:lccd}).
We determine $\calR^*$ at the final step of the algorithm.
Each step will be described in further detail in the next sections.

\begin{enumerate}[label=(\Roman*)]
	\item
	\label{algFindDCr::1}	
	\textbf{Removing heavy vertices:}
	Let $\hath$ be the set of all vertices that have weight at least $\lambda$.
	Let $\hatQ$ be the set of
	connected components in $\comp(V \setminus \hath)$ that have weight less than $\lambda$ and let $\hatc = V( \hatQ )$.
	We initialize $f^*$ by $f^*$-assigning each $Q \in \hatQ$ to an arbitrary $h \in N_{\hath}(Q)$. 
	Note that such an $h$ exists because there are no connected components in $G$ having weight less than $\lambda$.
	We initialize $H^* := \hath$, $C^* := \hatc$ and $\calR^* := \varnothing$.
	Note that the $f^*$-weight of each $h \in \hath$ is at least $\lambda$ as $w(h)\ge \lambda$.
	Thus, we have that $(C^*,H^*,f^*)$ is a WCD.
	
	\item
	\label{algFindDCr::2}
	\textbf{Initialization:}
	We use $G'$ to denote the current working graph throughout the algorithm, 
	and initialize it as $G':=G - (C^* \cup H^*)$.
	Note that each connected component in $G'$ weighs at least $\lambda$.
	We always maintain this invariant for $G'$ during the algorithm.
	We maintain a {\pbcod} $L := (\chrgg)$ of $G'$, 
	initialized as $H:=\varnothing$, $\calR:= \mathbb{CC}(V(G'))$, $\calC:=\varnothing$, and $g$ and $g'$ as the empty function. 
	We define $|H^*|+|H|+|\calR|$ as the \emph{outer index} and maintain it to be non-decreasing throughout the algorithm.
	An \emph{outer iteration} is the period during which the outer index remains the same.
	Similarly, we define $|H^*|+|H|$ as the \emph{inner index} and maintain it to be non-decreasing within an outer iteration.
	We say that we get progress whenever the outer index or the inner index increases.
	It is not difficult to derive that the outer and inner indices can be at most $k$ (see \cref{lem:k2it}), from the properties of {\pbcod}.

	\item
	\label{step:divorcut}
	\textbf{Divide or Cut-vertex:}
	As long as the {\pbcod} ($\chrgg$) is not an {\fbcod}, we pick an $R\in \calR$ such that its effective weight $w^{\bar{g}}(R) > 3(\lambda-1)$,
	and we run what we call the \emph{divide or cut-vertex} routine on the effective neighborhood $\nef[R]$.
	In this routine, we either find a CVP $\{R_1,R_2\}$ of $V(\nef[R])$ such that each of $R_1$ and $R_2$ weighs at least $\lambda$, or we find a vertex $h \in R \subseteq V(\nef[R])$ such that all connected components of $G'\left[V\left(\nef[R]\right)\right]- h$ have weight less than $\lambda$.
	
	\begin{enumerate}
		\item
		\label{overview:divide}
		In the former case, i.e.~in the case we divide $R$, we update $\calR$ to be $(\calR \setminus \{R\}) \cup \{R_1,R_2\}$,
		thereby increasing the outer index.
		We update $\calC$ to be $\calC \setminus \nu(R)$. 
		Note that $\nu(R) \cap \ginv(H) = \varnothing$ and hence the $g$-weight of each $h\in H$ remains the same.
		But it is possible that some $h\in H$ have now lost their $g'(h)$.
		So we do a \emph{cleanup} of $H$ as follows:
		For every $h \in H$, we add $N^g[h]$ to $\calR$,
		reset $H=\varnothing$, and $g$ and $g'$ as empty functions. 
		Note that before the cleanup, we had for each $h \in H$ that $w(N^g[h]) \geq 2 \lambda - 1$ and $G[N^g[h]]$ is connected.
		Moreover, $N^g[h]$ were disjoint for distinct $h$ and also disjoint from sets in $\calR$.
		Consequently, after the cleanup we obtain a {\pbcod} ($\chrgg$) of $G'$ such that the outer index is increased by one.
		
		\item
		In the latter case, i.e.~in the case where we find a cut-vertex $h$, we remove $R$ from $\calR$ and add $h$ to $H$, thus not changing the outer index and increasing the inner index.
		In order to maintain that $L$  is still a {\pbcod}, we add $\mathbb{CC}(R \setminus \{h\})$ to $\calC$.
		(Recall that each of these components have weight less than $\lambda$ as guaranteed by the divide or cut-vertex routine).
		Next, we $g$-assign some sub-components $\calC' \subset \calC$ to $h$ and assign another sub-component $C \in \calC \setminus \calC'$ as $g'(h)$ in such a way that the requirements for $g$-weight and $(g+g')$-weight in the definition of {\pbcod} are satisfied.		
		
		One problem that might occur is that now components in $\comp(V(\calC))$ have weight greater or equal to $\lambda$, which violates that ($\chrgg$) is a {\pbcod}.
		(This can happen because $R$ may have an edge to some $C\in \calC$ that is not in $\nu(R)$, i.e., $C$ has been already $g$-assigned to some other vertex in $H$).
		In this case, we add part of this big component in $\comp(V(\calC))$ as a new set to $\calR$ (thereby increasing outer index) and then do \emph{cleanup} of $H$. Here we do a more careful cleanup than Step~\ref{overview:divide}, crucially using the weight bounds regarding $g,g'$ for reconstructing a subgraph for each $h\in H$ with the required weight, to be added to $\calR$.
	\end{enumerate}
	
	Thus, as long as ($\chrgg$) is not an {\fbcod}, we get progress. 
	So, now assume that ($\chrgg$) is an {\fbcod}.
	For further details and formal description of the step, see Section~\ref{sec:divcut}.
	
	\item
	\label{step:bip}
	\textbf{Balanced Expansion:}
	Now, we construct an auxiliary bipartite graph $B$ with $H$ and $\qpriv$ as the two vertex sets.
	Recall that $\qpriv := \{Q \in \mathbb{CC}(V(\calC)) \mid N(Q) \subseteq H\}$.
	We put an edge in $B$ from an $h\in H$ to $Q\in \qpriv$ if and only if there is an edge in $G'$ from $Q$ to $h$.
	Note that 
	no $Q \in \qpriv$ can be an isolated vertex in $B$
	because if it is, then $Q$ is a connected component of $G'$ 
	having weight less than $\lambda$,
	but $G'$ does not have such connected components.
	As a result, the bipartite graph $B$ satisfies the precondition of a {\expanse}~\ref{definition::expanse}, which is our novel variation of weighted expansion.
	Applying Theorem \ref{lemma::cycleCanceling} on $B$, we find a partition of $H$ into $\hcrown$ (possibly empty) and $\hrem$ (possibly empty), and a function $f:\qpriv\rightarrow H$ that maps every $Q\in \qpriv$ to a neighbor of it in the bipartite graph $B$ such that: 
	\begin{enumerate}
		\item there are no edges from $f^{-1}(\hcrown)$ to $\hrem$ in $B$.
		\item for all $h\in \hcrown$, the $f$-weight of $h$ is at least $\lambda$ and
		\item for all $h\in \hrem$, the $f$-weight of $h$ is at most $3\lambda- 3$.
	\end{enumerate}
	Let $\Qcrown := \finv(H_1)$ and $C_1 = V(\Qcrown)$.
	Note that $\Qcrown = \mathbb{CC}(C_1)$.
	We update $H^*=H^*\cup H_1$ and $C^*=C^*\cup C_1$ and take $f^*:=f$ for the sets in $\Qcrown$. 
	
	From the definition of $\qpriv$ and properties of $\hcrown,\hrem$ and $f$, we have that $C_1$ does not have edges to $V(G') \setminus \hcrown$ in $G$, thus ensuring the separator property regarding $C_1$ and $H_1$ for the desired {\ccd}. 
	By condition (b) above, we obtain that the $f$-weight of each vertex $h \in H_1$ is at least $\lambda$, as required by a {\ccd}.
	
	We also update $G' = G' - (C_1 \cup H_1)$, $H = H \setminus H_1$, and $\calC = \calC \setminus \bigcup_{Q \in \Qcrown} \calC(Q)$.
	
	We show that after the deletion of $(C_1 \cup H_1)$, each connected component in $G'$ weighs at least $\lambda$, thus maintaining the invariant for $G'$. 
	The sub-components now in $\calC$ that were previously $g$-assigned to some vertex in $H_1$, become $g$-unassigned now.
	Note that since $H_2$ does not have edges to $C_1$,
	the $g$-neighborhood of all vertices in $H_2$ remains the same, and
	all the $g'$-assignments remain same.
	Thus, $L$ is still a {\pbcod} of $G'$ but might not be an {\fbcod} anymore because the effective neighborhoods of some $R\in \calR$ can increase as previously $g$-assigned vertices now become $g$-unassigned. 
	It is also easy to see that the value of the outer index and the inner index stay unchanged during this step.
	For further details and formal proofs, see Section~\ref{sec::step4}.
	
	\item
	\label{step:cggmod}
	\textbf{Private assignment:}
	We now modify $\calC,g$ and $g'$ using $f \colon \qpriv \to H$ such that:
	(a)$(\chrgg)$ is still a {\pbcod}, and
	crucially, (b) all sub-components in $\cCpriv$ are $g$-assigned
	Roughly speaking, the function $f$ from Step~\ref{step:bip} helps us to modify $g$ and $g'$ in a way so as to satisfy condition (b) without violating condition (a). 
	Let us say $g_0,g_0'$ are the $g$ and $g'$ at the start of the step. 
	The idea is to start with $g$ as the same assignment as $f$. We know $f$ satisfies that each $\cCpriv$ is $f$-assigned and the $f$-weight of any $h\in H$ is at most $3\lambda-3$.
	But some of the vertices in $H$ may not satisfy the required lower bound of $2\lambda-1$ (and the {\splitcond}) just from the $f$-weight.
	So, we fill up this deficit by using the $g$-assignment (and $g'$) while making sure that the $f$-assigned components remain $g$-assigned.
	Intuitively, we are able to fill up the $g$-assignment satisfying the weight conditions because $g$ satisfied those conditions before.
	Every time some $h\in H$ does not satisfy the conditions, there is a gap in the $g$-assignment of $h$ compared to its $g_0$-assignment, so there should have been some sub-component that was assigned in this gap in $g_0$, so we move that sub-component into the $g$-assignment of $h$ and fill the gap.
	We do not cycle because, once something is moved to its original place in $g$, then it is not moved anymore.
	We need to do this carefully so that no $f$-assigned components are pushed out of the $g$-assignment.
	For further details and formal description, see Section~\ref{sec::step5}.
	
	\item
	\label{step:merge-unassigned}
	\textbf{Merge unassigned subcomponents:}
	We now modify $\calC,g$ and $g'$ such that:
	(a)$(\chrgg)$ is still a {\pbcod}, 
	(b) all sub-components in $\cCpriv$ are $g$-assigned,
	and additionally, (c) for each $Q\in \comp(V(\calC))$, either all sub-components of $Q$ are $g$-assigned, or $Q$ itself is a single sub-component in $\calC$. 
	The condition (c) is ensured by the merging of the unassigned sub-components, sometimes merging into the assigned components.
	We need to do this carefully so that the weight conditions of the {\pbcod} are not violated.
	For further details and formal description, see Section~\ref{sec::step6}.
	
	\item
	\label{algFindDCr::6}
	\textbf{Final Step:}
	If $L$ is now not an {\fbcod}, then we {go to} Step \ref{step:divorcut}. 
	Recall that we showed in Step \ref{step:divorcut} that we get progress
	if $L$ is not an {\fbcod}.
	Otherwise, i.e.~if $L$ is an {\fbcod}, we push the whole $G'$ to $\calR^*$ as follows, and terminate the algorithm.
	Let $\calR=\{R_1,R_2,\dots, R_r\}$.
	We define $R^*_i:=\nef[R_i]\setminus \left(R^*_1\cup R^*_2\cup \dots R^*_{i-1}\right)$.
	We add each $R^*_i$ to $\calR^*$.
	Also, for each $h\in H$, we add $N^g[h]$ to $\calR^*$.
	
	In order to see that we indeed add all of $V(G')$ to $V(\calR^*)$ in this step, observe that after Step~\ref{step:cggmod}, we have that each $g$-unassigned sub-component has an edge to $V(\calR)$.
	Note that the vertex sets in $\calR^*$ are pairwise disjoint.
	
	The sets of the created $\calR^*$ have weight between $\lambda$ and $3\lambda-3$ as required by {\ccd} because: each $R^*_i\subseteq \nef[R_i]$ and $w(\nef[R_i]) \in [\lambda, 3\lambda-3]$ as $L$ is an {\fbcod}; for each $h\in H$, we have $w(N^g[h])=w^{g}[h] \in [2\lambda-1, 3\lambda-3]$ by the properties of the {\pbcod}.
\end{enumerate} 

We now give the detailed formal descriptions and proofs of the steps of this road-map.
\subsection{Step \uproman{1}: Removing Heavy Vertices}
\label{sub:steps_1_and_2}
The algorithm for Step~\ref{algFindDCr::1} is completely given in the description in Section~\ref{sub:algorithm_overview},
so we do not repeat it here.
We prove here the related lemmas that are required for correctness.

\begin{lemma}
	\label{lemma:weightedCrown}
	At the end of Step~\ref{algFindDCr::1} of Algorithm {\algCrown}, $(C^*,H^*,f^*)$ is a WCD of $G$.
\end{lemma}
\begin{proof}
	At the end of Step~\ref{algFindDCr::1}, we have $C^*=\hatC$, $H^*=\hatH$ and that $f^*$ is the function assigning each $C\in \hatC$ to some arbitrary $h \in N_{\hatH}(C)$. 
	Recall that $\widehat{H}$ is the set of vertices $h\in V$ with $w(h) \geq \lambda$ and $\widehat{C}$ are the vertices in components of $G - \widehat{H}$ with weight less than $\lambda$.
	Note that an $h\in N_{\hatH}(C)$ exists because there are no connected components in $G$ having weight less than $\lambda$.
	It is clear that $\hatH$ separates $\widehat{C}$ from the rest of the graph.
	By definition of $\hatH$, every $h\in \hatH$ has weight at least $\lambda$ and hence $w^{f^*}(h)\ge \lambda$.
	By definition, each $C\in \hatC$ has less than $\lambda$ weight.
	Thus $(C^*,H^*,f^*)$ is a WCD of $G$.
\end{proof}

Next, we show that the minimum weight condition on the connected components of $G'$ holds after step~\ref{algFindDCr::1}.

\begin{lemma}
	\label{lemma::step1MinWeight}
	After step~\ref{algFindDCr::1} of algorithm {\algCrown}, each connected component of $G' = G - (C^* \cup H^*)$ has a weight of at least $\lambda$.
\end{lemma}
\begin{proof}
	Observe that the components in $G'$ have weight of at least $\lambda$ by choice of $\widehat{C}$.
\end{proof}

\subsection{Step \uproman{2}: Initialization}
\label{sub:step2}
The algorithm for Step~\ref{algFindDCr::2} is completely given in the description in Section~\ref{sub:algorithm_overview},
so we do not repeat it here.
We prove here the following lemma that is required for correctness.
\begin{lemma}
	\label{lem:initpbcod}
	In the algorithm {\algCrown}, at the end of Step~\ref{algFindDCr::2}, $(\chrgg)$ is a {\pbcod} of $G'$.
\end{lemma}
\begin{proof}
	After Step~\ref{algFindDCr::1} each connected component in $G'$ weighs at least $\lambda$ by Lemma \ref{lemma::step1MinWeight}.
	The set $\calR$ corresponds to these components and hence {\Rcond} is satisfied.
	The {\gnbrcond} is satisfied as $g$ is the empty function.
	The {\compcond} is satisfied as $\calC=\varnothing$.
	All the 4 conditions in the definition of {\pbcod} are satisfied using that $H=\varnothing$.
\end{proof}

\subsection{Step \uproman{3}: Divide or Cut-vertex}
\label{sec:divcut}

If the {\pbcod} $L$ is not an {\fbcod}, we first pick an $\divR\in \calR$ such that $w^{\bar{g}}(\divR)\ge 3\lambda-2$ and apply the {\algDorC} routine on $\nef[\divR]$.
Depending on the output of this routine, we then run either the {\algDivide} or the {\algCut} procedure, both of which modify $L$.
In both cases we show that $L$ is still a {\pbcod} and either the outer index increases or the inner index increases with outer index remaining same, thereby giving progress.
\label{sub:step_3}
\subsubsection{Divide or Cut-vertex routine}
The {\algDorC} routine takes as input a vertex-weighted connected graph $G=(V,E,w)$ and an integer $\lambda>\wmax$, and outputs either a $[\lambda, \infty)$-{\cvp} $\left\{ V_1,V_2 \right\}$ of $V$ or a cut-vertex $x$ such that every connected component of $G-\left\{ x \right\}$ has less weight than $\lambda$.
We define such a vertex $x$ as \emph{$\lambda$-cut-vertex}.

Our algorithm uses the following theorem for $2$-connected graphs provided by Tarjan~\cite{even1976computing}.

\begin{theorem}[\cite{even1976computing}]
	\label{thm::vertexRanking}
	Let $G=(V,E)$ be a $2$-connected graph and $st \in E$.
	There is an ordering $f:V \to [|V|]$, such that $f(s) = 1$, $f(t) = |V|$ and for each $v \in V \setminus \{s,t\}$ there exist $uv,vw \in E$, such that $f(u) < f(v) < f(w)$ (\textit{rank property}).
	Furthermore, such an ordering is computable in $\mathcal{O}(|E|)$.
\end{theorem}

From Theorem \ref{thm::vertexRanking} we deduce a corollary for vertex-weighted $2$-connected graphs. 

\begin{corollary}
	\label{lemma::2Con2Sets}		
	Let $G = (V,E,w)$ be a vertex-weighted $2$-connected graph and let $\lambda > w_{\max}$ be an integer.
	If $w(G) > 3 (\lambda-1)$, then we find a $[\lambda,\infty)$-{\cvp} $\{V_1,V_2\}$ of $V$.
	Furthermore, this is realizable in $\mathcal{O}(|E|)$.
\end{corollary}
\begin{proof}
	Take an arbitrary $st \in E$ and compute an ordering $f:V \to [|V|]$ for the vertices $V$ such in Theorem \ref{thm::vertexRanking}.
	Let $k$ be the largest index, such that $\sum_{i=1}^{k-1} w(f^{-1}(i)) < \lambda$ and $\sum_{i=1}^{k} w(f^{-1}(i)) \geq \lambda$.
	Set $V_1 = \bigcup_{i=1}^k f^{-1}(i)$ and $V_2 = V \setminus V_1$.
	The rank property of $f$ provides that $G[V_1]$ and $G[V_2]$ are connected.
	From $w(f^{-1}(k)) < \lambda$ and $\sum_{i=1}^{k-1} w(f^{-1}(i)) < \lambda$ we obtain $w(V_1) \leq 2 (\lambda-1)$.
	Thus, $w(V) - w(V_1) \geq 3(\lambda - 1) + 1 - 2(\lambda-1) = \lambda$ and therefore, $w(V_2) \geq \lambda$.
	
	Realizing the ordering in Theorem \ref{thm::vertexRanking} requires time in $\mathcal{O}(|E|)$.
	To find the desired partitions, afterwards, requires time in $\mathcal{O}(|V|)$.
	As a result, we need a total running time in $\mathcal{O}(|E|)$.
\end{proof}

Finally, the following lemma gives the required {\algDorC} procedure.
\begin{lemma}
	\label{lemma::divideInto2Comp}
	Let $G = (V,E,w)$ be a connected vertex-weighted graph and let $\lambda > w_{\max}$ be an integer.
	If $w(G) > 3(\lambda - 1)$, then either we find a $\lambda$-cut-vertex $x \in V$ or a  $[\lambda,\infty)$-{\cvp} $\{V_1,V_2\}$ of $V$.
	Furthermore, this is realizable in time $\mathcal{O}(|V|\,|E|)$.
\end{lemma}

\begin{proof}
	If $G$ is $2$-connected, then by Corollary~\ref{lemma::2Con2Sets} we obtain the desired vertex sets $V_1$ and $V_2$.
	Otherwise, there exists a separator vertex $v \in V$.
	Let $\T = \{T_1, \dots, T_n\}$ be the connected components with $T_i \subset V$ of $G - v$.
	Moreover, let $\mathcal{T}_1 := \{T_i \in \T \mid w(T_i) \geq \lambda\}$ and let $\mathcal{T}_2 := \{T_i \in \T \mid w(T_i) < \lambda\}$.
	We work through the edge cases first.
	If $|\mathcal{T}_1| \geq 2$, then  we take two vertex sets $V_1,V_2$ of $\mathcal{T}_1$ and add the remaining ones $V \setminus (V_1 \cup V_2)$ arbitrarily to $V_1$ or $V_2$ to obtain the desired $[\lambda, \infty)$-{\cvp} of $V$.
	If $|\mathcal{T}_1| = 0$, then $x=v$ is the desired $\lambda$-cut-vertex.
	If $|\mathcal{T}_1| = 1$ and $w(v) + w\left(V\left(\mathcal{T}_2\right)\right) \geq \lambda$, then we take the only vertex set in $\T_1$ as $V_1$ and $V_2 = \{v\} \cup V\left(\mathcal{T}_2\right)$.
	Note, that $v$ is connected to each $T \in \mathcal{T}_2$.
	Thus, in all three edge cases we are done.
	
	In the remaining case, $|\mathcal{T}_1| = 1$ and $w(v) + w\left(V\left(\mathcal{T}_2\right)\right) < \lambda$, we contract the vertex set $\{v\} \cup V\left(\mathcal{T}_2\right)$ to $v$ and update $w(v)$ to $w(v) + w\left(V\left(\mathcal{T}_2\right)\right)$.	
	Since $w(v) < \lambda$, the constructed graph still satisfies $\wmax < \lambda$.
	Also, the graph remains connected and the total weight does not change, so we can restart the procedure with the contracted graph.
	Note that since the number of vertices decreases by at least one each iteration, this procedure terminates successfully.
	
	Finally, we analyze the running time.
	To simplify this, we assume that one copy of $G$ can be obtained in $\mathcal{O}(|E|)$.
	We use in each iteration the algorithm of Tarjan~\cite{Tarjan72}, that provides either that the graph is two $2$-connected, or a separator vertex $v \in V$.
	This costs $\bigO\left(|V| + |E|\right)$.
	If the graph is $2$-connected, then using Corollary~\ref{lemma::2Con2Sets}, the desired $[\lambda,\infty)$-{\cvp} can be computed in $\bigO\left(|E|\right)$.
	Otherwise, we consider $G - v$ and check firstly the edge cases.
	This runs in time $\bigO\left(|E|\right)$, since identifying the connected components is necessary. 
	In case we need to contract vertices we just delete $V\left(\T_2\right)$ from the graph and update only the weight of the separator vertex $v$.
	Remembering the contracted vertices can be achieved using a container that we are able to manage in time $\bigO(|V|)$ in each iteration individually as we have $|V|$ vertices in the graph.
	Consequently, each iteration runs in time $\bigO\left(|E|\right)$
	Lastly, if we do not find the desired $[\lambda,\infty)$-{\cvp} of $V$ we reduce the vertex set of the graph by at least one in each iteration.
	As a result, we obtain a final running time in $\bigO\left(|V|\,|E|\right)$, since we have at most $|V|$ iterations.
\end{proof}

	\subsubsection{The Divide Case}
	If {\algDorC} returns a division of $\nef[\divR]$ into $R_1$ and $R_2$, we do the procedure {\algDivide} as follows:
	\begin{enumerate}
		\item	
			Update $\calR$ to be $(\calR \setminus \{\divR\}) \cup \{R_1,R_2\}$,
		and $\calC$ to be $\calC \setminus \mathcal{N}_u(\divR)$. 
	\item
		\label{divide:cleanup}
		For every $h \in H$, remove the sub-components in $N^g[h]$ from $\calC$ and add $N^g[h]$ to $\calR$.
		Set $H=\varnothing$, and $g$, $g'$ as empty functions.
	\end{enumerate}

	\begin{lemma}
		\label{lem:divpbcod}
		$L=(\chrgg)$ is still a {\pbcod} after the {\algDivide} procedure.	
	\end{lemma}
	\begin{proof}
		First of all we prove that $(\chrg)$ satisfies the $3$ conditions for {\cod}:
		\begin{enumerate}
			\item
		Since we only delete elements (sub-components) from $\calC$ and do not add anything to $\calC$ in {\algDivide}, we have that every component in $\comp(V(\calC))$ still has weight less than $\lambda$.
	\item Since $g$ is the empty function, it holds trivially that $g(C) \in N_H(C)$ for every $g$-assigned $C\in \calC$.
	\item We have that $\calR$ is still a CVP of $G'-H-V(\calC)$ because
		any vertex that was removed from $H$ or $V(\calC)$ during {\algDivide} was added to $V(\calR)$ in {\algDivide} and all the sets added to $V(\calR)$ in {\algDivide} are pairwise disjoint.
		In order to see that $\calR$ is a $[\lambda,\infty]$-CVP, observe that each set added to $\calR$ has weight at least $\lambda$. The sets $R_1$ and $R_2$ have weight at least $\lambda$ as given by the Divide or Cut-vertex procedure. 
		Each  $N^g[h]$ added during Step~\ref{divide:cleanup} of {\algDivide} has weight at least $\lambda$ as the $g$-weight of each $h\in H$ was at least $2\lambda$ before the {\algDivide} procedure by property of {\pbcod}.
		\end{enumerate}
		The 4 conditions in the definition of {\pbcod} are trivially satisfied as $H$ is empty.
	\end{proof}

	\begin{lemma}
		\label{lem:DivideIndex}
		The outer index increases during {\algDivide}.
	\end{lemma}
	\begin{proof}
		Let $\hold$ and $\Rold$ be the $H$ and $\calR$ respectively before the {\algDivide} procedure.
		It is clear from the {\algDivide} procedure that $|\calR|=|\hold|+|\Rold|+1$.
		This is because we replace $\divR\in \Rold$ by two sets $R_1$ and $R_2$ in $\calR$, all sets in $\Rold\setminus \{\divR\}$ remain in $\calR$, and for each $h\in \hold$ a set is added to $\calR$.
		Also, observe that $H^*$ remains unchanged.
		Thus, the outer index increases by one.
	\end{proof}

	\begin{lemma}
		\label{lemma::timeDecide}
		{\algDivide} runs in $\bigO	(|V|\,|E|)$ time.
	\end{lemma}
	\begin{proof}
	It is easy to see that each step runs in $\bigO(|V|\,|E|)$ time.	
	\end{proof}
	\subsubsection{The Cut-vertex Case}
	In the case where {\algDorC} finds a $\lambda$-cut-vertex $h_c$ of $\nef[\divR]$, 
	we apply the {\algCut} procedure on $L=(\chrgg)$ as follows.
	Note that $h_c$ is in fact in $\divR$ as otherwise there is a connected component in $G[\nef[\divR]]-h_c$ that completely contains $\divR$ and hence has weight at least $\lambda$. 

	\begin{enumerate}
		\item
			\label{step:cp}
			Let $C_1,\dots,C_p$ be the connected components of $G[\nef[\divR]]-\{h_c\}$ in \emph{decreasing} order of weight.
			Let $i$ be the smallest number such that $w(h_c)+w(C_1\cup \dots \cup C_i)\ge 3\lambda-2$. (Note that such an $i$ exists because $w(h_c)+w(C_1\cup\dots \cup C_p)=w(\nef[\divR])\ge 3\lambda-2$).

		\item
			\label{step:cutassign1}
			Let $\calC_1:= \left(\calC\setminus\nu(\divR)\right)\cup\left\{ C_1,\dots,C_p \right\}$, $\calR_1:= \calR\setminus \{ \divR \}$, and $H_1:= H\cup \{h_c\}$. 

			Let $g:\calC_1\pto H_1$ be defined as:
			for all $g$-assigned $C\in \calC$, let $g_1(C):=g(C)$ (note that all $g$-assigned $C\in \calC$ are also in $\calC_1$ by definition of $\calC_1$); 
			let $g_1(C_j):=h_c$ for all $j\in [i-1]$. 

			Let $g_1':H_1\to \calC_1$ be defined as:
			for all $h\in H$, if $g'(h)\in \left(\calC\setminus \nu(\divR)\right)\cup \left\{ \varnothing \right\}$ then let $g_1'(h):=g'(h)$;
			otherwise, i.e.~if $g'(h)\in \nu(\divR)$ then let $g_1'(h)$ be the $C_j$ such that $g'(h)\subseteq C_j$ (such a $C_j$ exists by the definitions of $\nu(\divR)$ and $C_1,\dots, C_p$); 
			let $g_1'(h_c):=C_i$. 

		\item
			\label{step:callcutclean}
			If there exists a connected component $\hatq$ of $G[V(\calC_1)]$ with $w(\hatq)\ge \lambda$, then perform the {\algCutClean} procedure given below, otherwise update $L\leftarrow (\chrggone)$.
\end{enumerate}

			The {\algCutClean} procedure is as follows:
			\begin{enumerate}
				\item
					\label{step:tree}
					Find a subset $Q'$ of $\hatq$ such that $G[Q']$ is connected, and $w(Q')\in [\lambda,2\lambda-2]$ as follows:
					let $\conQ(V_{\conQ},E_{\conQ})$ be the graph with
				    vertex set $V_{\conQ}:=\left\{ C\in \calC_1: C\subseteq Q \right\}$,
					and there an edge between sub-components $C$ and $C'$ in $E_\conQ$ exists, if and only if 
					there is an edge between $C$ and $C'$ in $G$. 
					For any subgraph $\conQ'$ of $\conQ$, let $V(\conQ')$ denote the set of all $v\in V$ such that 
					$v$ is contained in some vertex of $\conQ'$, and let $w(\conQ'):=w(V(\conQ'))$.
					Find a spanning tree $T$ of $\conQ$.
					As long as $T$ has a leaf $C$ such that $w(T- \{C\})\ge \lambda$, 
					update $T\leftarrow T-\{C\}$.
					Let $T'$ be the final $T$ when no longer such leaves can be found.
					Let $\calC'\subseteq V_{\conQ}$ be the set of vertices of $\conQ$ that are spanned by $T'$.
					Let $Q':=V(\calC')$. 
					(We show in \cref{lem:treeweight} that indeed $w(\calC')=w(Q')\in [\lambda,2\lambda-2]$).

				\item
					\label{step:cutassign2}
				Let	
					$\calC_2:= \calC_1\setminus \calC'$,
					$\calR_2:= \calR_1\cup \{Q'\}$,
					and
					$H_2:= H_1$.

					Let $g_2:\calC_2\pto H_2$ be defined as: for all $g_1$-assigned $C\in \calC_2$,
					let $g_2(C):=g_1(C)$.
					Let $g_2':H_2\to \calC_2$ be defined as: 
					for all $h\in H_2$, if $g_1'(h)\in \calC_2$ then
					let $g_2'(h):=g_1'(h)$ and $g_2'(h):=\varnothing$ otherwise.

				\item 
					\label{step:cutdeficiency}
					If $w^{g_2}[h]\ge \lambda$ for all $h\in H_2$, then: 
					let $\calC_h:=\g2inv(h)$ for each $h\in H_2$.
						
					Otherwise, i.e., if there is some $h_d\in H_2$ ($d$ stands for \emph{deficiency}) such that $w^{g_2}[h_d]< \lambda$ (we prove in \cref{lem:hd} below that there is at most one such $h_d$), then:
					let 
					$\calC_{h_d}:=\g2inv(h_d)\cup \{g'_2(h_d)\}$ 
					(we show in \cref{lem:hd} that $g'_2(h_d)\neq \varnothing$),
					and for all $h\in H_2\setminus \left\{ h_d \right\}$, 
					let 
					$\calC_h:=\g2inv(h)\setminus \left\{ g'_2(h_d) \right\}$ .
					Let $R_h:=\left\{ h \right\}\cup V(\calC_h)$ for each $h\in H_2$.

				\item
					\label{step:cutassign3}
				Let	
			$\calQ_3$ be the set of connected components of $G[V\left(\calC_2\setminus \left( \bigcup_{h\in H_2}\calC_h\right)\right)]$.
			Let $\calC_3:=\left\{ Q\in \calQ_3:w(Q)<\lambda \right\}$,
					$\calR_3:= \calR_2\cup \left\{ R_h:h\in H_2 \right\}\cup
			\left\{ Q\in \calQ_3:w(Q)\ge\lambda \right\}$,
					$H_3:= \varnothing$, and $g_3$, $g_3'$ be empty functions.

					Update $L\leftarrow (\chrggthr)$
				\end{enumerate}
				Lemmas \ref{lem:nocutcleanup}, \ref{lem:CutPbcod}, \ref{lem:CutIndex} and \ref{lem:cutruntime} summarize the correctness and running time of this section. We defer the proofs of Lemmas \ref{lem:nocutcleanup}, and \ref{lem:CutPbcod} to later as we need to prove some auxiliary lemmas for this.
\begin{lemma}
	\label{lem:nocutcleanup}
	If {\algCutClean} is not called then $(\chrggone)$ is a {\pbcod}.
\end{lemma}
	\begin{lemma}
		\label{lem:CutPbcod}
		If {\algCutClean} is called, then $(\chrggthr)$ is a {\pbcod}. 
	\end{lemma}

	\begin{lemma}
		\label{lem:CutIndex}
		During {\algCut}, either the outer index increases, or the outer index remains the same and the inner index increases.
	\end{lemma}
	\begin{proof}
		First of all, notice that $H^*$ is not modified during {\algCut}. 
		It is clear by construction that $|H_1|=|H|+1$ and $|\calR_1|=|\calR|-1$.	
		So, if {\algCutClean} is not called then we have that the outer index remains same and inner index increases during {\algCut}.

		So, now consider the case when {\algCutClean} is called.
		By construction, we have $H_2=H_1$, $|\calR_2|=|\calR_1|+1$,
		and $|\calR_3|\ge |\calR_2|+|H_2|$.
		This implies $|\calR_3|\ge |H_1|+|\calR_1|+1$.
		Thus the outer index increases.
	\end{proof}
\begin{lemma}
	\label{lem:cutruntime}
	{\algCut} runs in $\bigO(|V|\,|E|)$ time.
\end{lemma}
\begin{proof}
	It is not difficult to see that each step can be implemented in $\bigO(|V|\,|E|)$ time.
\end{proof}
\begin{lemma}
	\label{lem:cprivsupset}
 $\cCpriv_1\supseteq\cCpriv$.
\end{lemma}
\begin{proof}
Suppose there exists a $C\in \cCpriv$ that is not in $\cCpriv_1$.
Since $C\in \cCpriv$, $C$ does not have an edge to $R$ and hence is not in $\nu(\divR)$.
 This implies that $C\in \calC_1$.
 Let $Q_1$ be the connected component of $G[V(\calC_1)]$ that contains $C$,
 and let $Q$ be the connected component of $G[V(\calC)]$ that contains $C$.
 Since $V\left( \calC_1 \right)\supseteq V\left( \calC \right)$, we have that $Q_1\supseteq Q$.
 Since $C\notin \cCpriv_1$, we have that $Q_1$ has 
an edge to some $R\in \calR_1$.
 But $\calR_1\subset\calR$ by construction and hence $R\in \calR$.
 Then we know that $Q$ does not have an edge to $R$ as $C\in \cCpriv$.
 This implies $Q_1\supsetneq Q$ and hence $Q_1\cap \nef[\divR]\neq \varnothing$.
 But then $N(Q)\cap \nef[\divR]\neq \varnothing$.
 This implies $C\notin \cCpriv$, thus a contradiction.
\end{proof}

\begin{lemma}
	\label{lem:chrgg1}
	\begin{enumerate}
		\item
			\label{item:almostpbcod}
	$(\chrgone)$ satisfies conditions \ref{prop:gnbh} and \ref{prop:Rweight} of {\cod} and $(\chrggone)$ satisfies 
	conditions \ref{item:gdnbh}, \ref{item:wgh}, \ref{item:gg1}, and \ref{item:gg2} of {\pbcod}.
		\item
			Each $C\in \calC_1$ weighs less than $\lambda$ (a weak version of {\compcond} of {\cod}).
			\label{item:clambda}
		\item For all $h\in H_1$, if $\ginvone(h)\nsubseteq \cCpriv$, then $\wggone[h]\ge 3\lambda-2$
		(this is stronger than the $(g+g')$-weight condition for $(\chrggone)$ as we have $\cCpriv$ and not $\cCpriv_1$ in the statement and $\cCpriv\subseteq\cCpriv_1$ by \cref{lem:cprivsupset}).
	\label{item:gg1strong}
\item For all $h\in H_1$, if $w^{g_1}(h)<2.5 (\lambda-1)$, then $w(C)\ge 0.5 (\lambda-1)$ for every $C\in \ginvone(h)\setminus\cCpriv$
	(this is stronger than the $0.5 (\lambda-1)$-condition for $(\chrggone)$ as we have $\cCpriv$ and not $\cCpriv_1$ in the statement and $\cCpriv\subseteq\cCpriv_1$ by \cref{lem:cprivsupset}).
	\label{item:gg2strong}
	\end{enumerate}
\end{lemma}
\begin{proof}
	First we prove condition \ref{prop:gnbh} of {\cod} ($g$-neighbor condition) for $(\chrgone)$.
	Consider any $g_1$-assigned $C$ in $\calC_1$.
	If $g_1(C)=g(C)$ then the condition is satisfied for $C$, by applying condition \ref{prop:gnbh} of {\cod} for $(\chrg)$.
	So assume $g_1(C)\neq g(C)$.
	Then by construction of $g_1$, we know that $C\in \left\{ C_1,\dots.,C_{i-1} \right\}$ and $g_1(C)=h_c$ (see step \ref{step:cp} of {\algCut}). 
	Also $h_c \in N(C_j)$ for all $j\in [p]$ by construction of $h_c$ and $C_j$. 
	Thus the condition is satisfied.

	Now, we prove condition \ref{prop:Rweight} of {\cod} ($\calR$-condition) for $(\chrgone)$.
	Indeed, it is clear from the construction that $\calR_1$ is a $[\lambda,\infty)$-{\cvp} of $G-H_1-V(\calC_1)$ by using that 
	$\calR$ is a $[\lambda,\infty)$-{\cvp} of $G-H-V(\calC)$.

Next, we prove the 4 conditions of {\pbcod} for $(\chrggone)$.
	Consider $h\in H=H_1\setminus \left\{ h_c \right\}$.
	We have that $\ginvone(h)=\ginv(h)$ and $g_1'(h)\supseteq g'(h)$
	from the definitions of $g_1$ and $g_1'$.
	We now prove the 4 conditions of {\pbcod} for $(\chrggone)$ for $h\in H = H_1 \setminus \{h_c\}$. We also prove statements \ref{item:gg1strong} and \ref{item:gg2strong} of the lemma for $h\in H$ along the way.
    \begin{enumerate}
		\item $g$-neighbor condition: Suppose $g'_1(h)\neq \varnothing$. 
			We know that $g_1'(h)\supseteq g'(h)$.
			We have that $g'(h)$ has an edge to $h$ using the $g'$-neighbor condition of $(\chrgg)$.
			Thus $g_1'(h)$ has an edge to $h$.
			We also know that $g'_1(h)\in \calC_1$ by the definition of $g'_1$.
			So, $g_1'(h)\in \calN_{\calC_1}(h)$.
			Hence it only remains to prove that $g_1'(h)\notin \ginvone(h)$.
			Suppose 
			 $g_1'(h)\in \ginvone(h)$.
			 But $\ginvone(h)=\ginv(h)$ and hence $g_1'(h)\in \ginv(h)$.
			 This implies that $g_1'(h)\in \calC$ and we get $g'(h)=g_1'(h)$ by using that $g_1'(h)\supseteq g'(h)$.
			 Thus, $g'(h)\in \ginv(h)$, a contradiction to the $g'-$neighbor condition of $(\chrgg)$.
		 \item $g$-weight condition: follows from $\ginvone(h)=\ginv(h)$ and the $g$-weight condition of $(\chrgg)$.
		 \item $(g+g')$-weight condition: 
			 we prove the more general statement \ref{item:gg1strong} of the lemma.
			 Suppose $\ginvone(h)\nsubseteq \cCpriv$.
			 Since $\ginvone(h)=\ginv(h)$, we have $\ginv(h)\nsubseteq \cCpriv$.
			 Then by the 
			 $(g+g')$-weight condition of $(\chrgg)$, it follows that 
			 $\wgg[h]\ge 3\lambda-2$.
		  	Using $\ginvone(h)=\ginv(h)$, and $g_1'(h)\supseteq g'(h)$, we get that
			 $\wggone[h]\ge 3\lambda-2$.
		 \item $0.5(\lambda-1)$ condition: 
			 we prove the more general statement \ref{item:gg2strong} of the lemma.
			 Suppose $w^{g_1}[h]< 2.5(\lambda-1)$.
			 Since $\ginvone(h)=\ginv(h)$, we have 
			 $w^{g}[h]< 2.5(\lambda-1)$.
			 Consider a $C\in \ginvone(h)\setminus \cCpriv$.
			 Since $\ginvone(h)=\ginv(h)$, 
			 we have $C\in \ginv(h)\setminus \cCpriv$.
			 Then $w(C)\ge 0.5(\lambda-1)$ using the {\splitcond} of $(\chrgg)$. 
    \end{enumerate}

	Now, we prove the 4 conditions of {\pbcod} for $(\chrgg)$ and the statements \ref{item:gg1strong} and \ref{item:gg2strong} for $h_c$.
	We know that $g'_1(h_c)=C_i$ and $\ginvone(h_c)=\left\{ C_1,\dots,C_{i-1} \right\}$ (see step~\ref{step:cp} of {\algCut}).
	\begin{enumerate}
		\item	{\gdnbrcond}:
	From the definitions of $C_1,\dots, C_p$ and $h_c$, it is clear that $C_i$ has an edge to $h_c$ and that 
	$\{C_1,\dots,C_{i-1}\}$ is disjoint from $C_i$. 
\item 
	{\gwcond}:
	By the selection of $i$ in Step~\ref{step:cp} of {\algCut}, we have that $w^{g_1}[h_c]=w(C_1)+\dots +w(C_{i-1})\le 3\lambda-3$.
	Also, $w^{g_1}[h_c]=w(C_1)+\dots +w(C_{i-1})\ge \left(w(C_1)+\dots +w(C_i)\right)-w(C_i)\ge (3\lambda-2)-(\lambda-1)=2\lambda-1$, where we used $w(C_i)\le \lambda-1$ as given by the {\algDorC} procedure.
\item
	{\ggwcond} and statement \ref{item:gg1strong} of lemma:
	$\wggone[h_c]=w(C_1)+\dots +w(C_i)\ge 3\lambda-2$ by the selection of $i$.
\item
	{\splitcond} and statement \ref{item:gg2strong} of lemma:
	If $w^{g_1}[h_c]<2.5(\lambda-1)$, then by the selection of $i$, we have $w(C_i)>(3\lambda-2)-2.5(\lambda-1)>0.5(\lambda-1)$.
	Then, $w(C_1),\dots,w(C_{i-1})>0.5(\lambda-1)$ as $C_1,\dots, C_p$ were sorted in decreasing order of weight.
	Thus, $w(C)\ge 0.5(\lambda-1)$ for each $C\in \ginvone(h_c)$.
	\end{enumerate}

	Finally, we prove statement \ref{item:clambda} of the lemma. If $C\in \calC$ then this is clear by the {\compcond} of $(\chrg)$. 
	So assume $C\notin \calC$.
	This implies $C=C_j$ for some $j\in [p]$ by the construction of $\calC_1$.
	Then we know $w(C)<\lambda$ from the {\algDorC} procedure.
\end{proof}

\begin{pfof}{\cref{lem:nocutcleanup}}
	By statement \ref{item:almostpbcod} of \cref{lem:chrgg1}, it follows that we only need to prove the {\compcond}.	
	But if the {\compcond} is not fulfilled by $(\chrggone)$ then Step \ref{step:callcutclean} of {\algCut} would call {\algCutClean}, thus a contradiction.
\end{pfof}
\begin{lemma}
	\label{lem:treeweight}
	$w(\calC')=w(Q')\in [\lambda,2\lambda-2]$.	(See Step~\ref{step:tree} of {\algCutClean} for the definitions of $\calC'$ and $\calQ'$).
\end{lemma}
\begin{proof}
	We have that $w(\calC')=w(Q')=w(T')$ from Step~\ref{step:tree} of {\algCutClean}.
	$T'$ is a tree, thus connected.
	So it is sufficient to prove that $w(T')\in [\lambda,2\lambda-2]$.
					It is clear that $w(T')\ge \lambda$ by construction. 
					It only remains to prove that $w(T')\le 2\lambda-2$.
					Suppose $w(T')\ge 2\lambda-1$.
					Now, consider any leaf $C$ of $T'$. 
					We have $w(C)\ge\lambda$, because otherwise $w(T'-\left\{ C \right\})\ge \lambda$, a contradiction to the construction of $T'$. 
					We know that $C\in \calC_1$ as $C$ is a vertex of $\conQ$.
					But then $w(C)<\lambda$ by \cref{lem:chrgg1}.
					Thus we have a contradiction.
\end{proof}

\begin{lemma}
	\label{lem:hdcpriv}
	For $h_d$ as in Step~\ref{step:cutdeficiency} of {\algCutClean},
	we have $\ginvone(h_d)\nsubseteq \cCpriv$.
\end{lemma}
\begin{proof}
    Suppose $\ginvone(h_d)\subseteq\cCpriv$. 
	This means that any connected component of $G[V(\calC)]$ that intersects $V(\ginvone(h_d))$, does not intersect $\nef(\divR)$.
	This implies that $V(\ginvone(h_d))$ is disjoint from the component $\hatq$ in Step~\ref{step:tree} of {\algCutClean},
	and hence is also disjoint from $Q'$. 
	This implies that $\ginvone(h_d)$ is disjoint from $\calC'$.
	However by the definition of $h_d$, we have	that
	$w\left( \ginvone(h_d)\cap \calC'\right)\ge \lambda$ as $w\left( \ginvone(h_d) \right)\ge 2 \lambda - 1$, thus a contradiction.
\end{proof}
\begin{lemma}
	\label{lem:hd}
	There exists at most one $h_d\in H_2$ such that $w^{g_2}(h_d)< \lambda$.
	Also, for such an $h_d$, we have $g_2'(h_d)\neq \varnothing$, and for all $h\in H_2\setminus\left\{ h_d \right\}$, we have that
	$w\left(\ginvone(h)\cap \calC'\right)< 0.5(\lambda-1)$.
\end{lemma}
\begin{proof}
	Consider any $h_d\in H_2$ such that $w^{g_2}(h_d)<\lambda$.
	Since $H_2=H_1$, $h_d\in H_1$.
	By \cref{lem:chrgg1}, we get that $w^{g_1}(h_d)\ge 2\lambda-1$.
	Thus, $w^{g_1}(h_d)-w^{g_2}(h_d)\ge \lambda$. 
	This implies that $w\left( \ginvone(h_d)\cap \calC'\right)\ge \lambda$, 
	by the definition of $g_2$.
	Now, consider a leaf $C$ of the tree $T'$ in Step~\ref{step:tree} of {\algCut}.
	We have that $C\in\ginvone(h_d)$ because otherwise $T'-\{C\}$ contains all of $\ginvone(h_d)\cap \calC'$ and hence has a weight of at least $\lambda$, a contradiction to the construction of $T'$.
	Thus $g_1(C)=h_d$. 
	Since $g_1(C)$ is unique, $h_d$ is also unique, proving the first part of the lemma.

	As we showed above, every leaf of $T'$ is in $\ginv(h_d)$.
	Since $T'$ has at least $2$ leaves and $w(\ginvone(h_d)\cap \calC')\ge \lambda$, it follows that $T'$ has at least one leaf $C$ in $\ginvone(h_d)$ such that $w\left(\left( \ginvone(h_d)\cap \calC' \right)\setminus \left\{ C \right\}\right)\ge 0.5\lambda$.
	Now, if there is an $h\in H_2\setminus \left\{ h_d \right\}$ such that $\ginvone(h)\cap \calC'\ge 0.5(\lambda-1)$, then we get that $w(T'-\left\{ C \right\})\ge 0.5\lambda+0.5(\lambda-1)=\lambda-0.5$. 
	This implies $w(T'-\left\{ C \right\})\ge \lambda$ by integrality.
	But this is a contradiction to the construction of $T'$.

	It only remains to prove that $g_2'(h_d)\neq\varnothing$ in order to conclude the proof of the lemma. 
	Suppose $g_2'(h_d)=\varnothing$.
	This implies $g_1'(h_d)\notin \calC_2=\calC_1\setminus\calC'$ by the definition of $g_2'$.
	Then either $g_1'(h_d)=\varnothing$ or $g_1'(h_d)\in \calC'$.

	Case 1. $g_1'(h_d)=\varnothing$: this implies that $g'(h_d)=\varnothing$, by definition of $g_1'$.
	But then $\ginv(h_d)\subseteq\cCpriv$ using the {\gwcond} and the {\ggwcond} of the {\pbcod} $(\chrgg)$.
	This is a contradiction to \cref{lem:hdcpriv}.

	Case 2. 
	$g_1'(h_d)\in \calC'$:
then we have $w(\calC')\ge w\left(\{g_1'(h_d)\}\cup \left(\ginvone(h_d)\cap \calC'\right)\right)$.
	We know that $g_1'(h_d)$ is disjoint from $\ginvone(h_d)$ 
	by the {\gdnbrcond} of $(\chrggone)$ given by \cref{lem:chrgg1}.
	Thus we have $w(\calC')\ge w\left( g_1'(h_d) \right)+w(\ginvone(h_d)\cap \calC')= w\left( g_1'(h_d)\right) 
	+w^{g_1}[h_d]-
	w^{g_2}[h_d]
	=\wggone[h_d]-w^{g_2}[h_d]$.
	Since $w^{g_2}[h_d]\le \lambda-1$ and $w(\calC')\le 2\lambda-2$ (by \cref{lem:treeweight}), we have that 
	$\wggone[h_d] \le 3\lambda-3$.
	This implies that 
	$\wgg[h_d] \le 3\lambda-3$ by the definition of $g_1$ and $g'_1$.
	But then $\ginvone(h_d)\subseteq \cCpriv$ 
	from the statement \ref{item:gg1strong} of \cref{lem:chrgg1}.
	This is
	a contradiction to \cref{lem:hdcpriv}.
\end{proof}
\begin{lemma}
	\label{lem:g2dlambda}
	For all $h\in H_2$, $w(g_2'(h))<\lambda$.	
\end{lemma}
\begin{proof}
	By definition of $g_2'$, either $g_2'(h)=\varnothing$ or $g_2'(h)=g_1'(h)$.
	In the former case, clearly $w(g_2'(h))=0<\lambda$. 
	In the latter case we have $w(g_2'(h))=w(g_1'(h))<\lambda$ by \cref{lem:chrgg1}.
\end{proof}

\begin{pfof}{\cref{lem:CutPbcod}}
First, we prove that $(\chrgthr)$ satisfies the 3 conditions in the definition of {\cod}:
\begin{enumerate}
	\item
		{\compcond}:
		From the construction of $\calC_3$,
		it is clear that all connected components in $G[V(\calC_3)]$ have weight less than $\lambda$. 
	\item 
		{\gnbrcond}:
		follows since $g_3$ is the empty function.
	\item {\Rcond}: We need to prove that $w(R)\ge \lambda$ for all $R\in \calR_3$.
		Recall that 
					$\calR_3= \calR_2\cup \left\{ R_h:h\in H_2 \right\}\cup
			\left\{ Q\in \calQ_3:w(Q)\ge\lambda \right\}$.

			Case 1. $R\in \calR_2$:
			Recall that $\calR_2=\calR_1\cup \left\{ Q' \right\}$.

			Case 1.1. $R\in \calR_1$: 
			By {\Rcond} of {\cod} $(\chrgone)$ given by \cref{lem:chrgg1}, we know that $w(R)\ge \lambda$.

			Case 1.2. $R=Q'$: by \cref{lem:treeweight}, $w(Q')\ge \lambda$.

			Case 2. $R=R_h$ for some $h\in H_2 $:

			Case 2.1. If $w^{g_2}[h]\ge \lambda$ for all $h\in H_2$:
			then $R_h=N^{g_2}[h]$ by construction.
			Thus, $w(R_h)=w^{g_2}[h]\ge \lambda$, by assumption of Case 2.1.

			Case 2.2. 
					If there is some $h_d\in H_2$ such that $w^{g_2}[h_d]< \lambda$ :

			Case 2.2.1. $h=h_d$:
			we know that $R_{h_{d}}=\left\{ h_d \right\}\cup\ginvtwo(h_d)\cup \{g_2'(h_d)\}$ by the construction of $R_{h_d}$. 
			By the construction of $g_2$ we get that, $\ginvtwo(h_d)\cup \{g_2'(h_d)\}=\left(\ginvone(h_d)\cup \left\{ g_1'(h_d) \right\}\right)\setminus \calC'$. 
			This implies that $w(R_{h_d})\ge \wggone[h_d]-w(\calC')$.
			We have $\wggone[h_d]\ge 3\lambda-2$ using the {\ggwcond} of $(\chrggone)$ from \cref{lem:chrgg1} and $w(\calC')\le 2\lambda-2$ from \cref{lem:treeweight}. 
			Thus, $w(R_{h_d})\ge \lambda$.
				
			Case 2.2.2. $h\neq h_d$:
			Recall that in this case $R_h=\left\{ h \right\}\cup \g2inv(h)\setminus \left\{ g'_2(h_d) \right\}$.
		Thus, $w(R_h)\ge w^{g_2}[h]-w(g_2'(h_d))=w^{g_1}[h]-w\left(\ginvone(h)\cap\calC'\right)-w(g_2'(h_d))$.
		Suppose for the sake of contradiction that $w(R_h)\le \lambda-1$.
			From \cref{lem:g2dlambda}, we have $w(g_2'(h_d))\le (\lambda-1)$.
			Thus, we have $w^{g_1}[h]-w\left(\ginvone(h)\cap\calC'\right)\le 2(\lambda-1)$.
			From \cref{lem:hd}, we have that 
			$w\left(\ginvone(h)\cap \calC'\right)< 0.5(\lambda-1)$.
			Thus, we get $w^{g_1}[h]< 2.5(\lambda-1)$. 
			Then, by statement \ref{item:gg2strong} of \cref{lem:chrgg1},
	for each $C\in \ginvone(h)\setminus \cCpriv$,
	we have $w(C)\ge 0.5(\lambda-1)$.
	Since $w\left(\ginvone(h)\cap \calC'\right)< 0.5(\lambda-1)$, this implies that if $C\in \ginvone(h)\cap \calC'$, then $C\in \ginvone(h)\cap\calC'\cap\cCpriv$.
	But $\calC'\cap \cCpriv=\varnothing$ by the constructions of $\cCpriv$ and $\calC'$. 
	Thus $\ginvone(h)\cap\calC'=\varnothing$.
	Plugging this in 
	$w^{g_1}[h]-w\left(\ginvone(h)\cap\calC'\right)\le 2(\lambda-1)$, we get that $w^{g_1}[h]\le 2(\lambda-1)$,
	a contradiction to {\gwcond} of $(\chrggone)$ given by Lemma~\ref{lem:chrgg1}.

			Case 3. 
			$R\in\left\{ Q\in \calQ_3:w(Q)\ge\lambda \right\}$: in this case it is clear that $w(R)\ge \lambda$.
\end{enumerate}
The 4 conditions in the definition of {\pbcod} are trivially satisfied by $(\chrggthr)$ as $H_3=\varnothing$.
\end{pfof}

\subsection{Step \uproman{4}: Balanced Expansion}
\label{sec::step4}

The algorithm for Step~\ref{step:bip} is completely given in the description in Section~\ref{sub:algorithm_overview},
so we do not repeat it here.
We prove here the related lemmas that are required for correctness.
\begin{lemma}
	\label{lemma:weightedCrown2}
	After the execution of the step~\ref{step:bip} during algorithm {\algCrown}, $(C^*,H^*,f^*)$ remains a weighted crown decomposition of $G$.
\end{lemma}
\begin{proof}
	Consider the $r$-th execution of step~\ref{step:bip}.
	We do induction on $r$.
	The induction base case is true because of Lemma~\ref{lemma:weightedCrown} in Step~\ref{algFindDCr::1}.
	Let $C^*_0$, $H^*_0$ and $f^*_0$ be the $C^*,H^*$ and $f^*$ before the execution of the step. 
	By induction hypothesis we can assume that $(C^*_0,H^*_0,f^*_0)$ is a weighted crown decomposition (Note that step \ref{step:bip} is the only place we modify $C^*,H^*$ and $f^*$ apart from once at the start in step~\ref{algFindDCr::1}). 
	We know that $C^*=C^*_0\cup C_1$ and $H^*=H^*_0\cup H_1$ (see Step~\ref{step:bip} of {\algCrown}).
	Also, $f^*$ on $C^*_0$ is same as $f^*_0$, and
	$f^*$ on $C_1$ is same as $f$ on $C_1$, where $f$ is the assignment returned by the balanced expansion in Step~\ref{step:bip}.
	By the properties of balanced expansion, $f$ satisfies:
	\begin{enumerate}
		\item[a)] there are no edges from $f^{-1}(\hcrown)=C_1$ to $\hrem$ in $B$.
	\item[b)] $f(C)\in N(C)$ for all $C\in \comp(C_1)$
		\item[c)] for all $h\in \hcrown$, the $f$-weight of $h$ is at least $\lambda$ and
	\end{enumerate}
		We know that there are no edges from $C^*_0$ to $C_1$ as $H^*$ separates $C^*$ from $G-H^*-C^*$. 
		Thus all connected components in $C^*=C^*_0\cup C_1$ are 
		either connected components of $C_1$ or that of $C^*_0$.
		We now prove the conditions 1-4 in Definition~\ref{def:lccd} for $(C^*,H^*,f^*)$, thereby showing that it is a WCD.
	\begin{enumerate}
		\item We know there are no edges from $C^*_0$ to $V(G')$ as $(C^*_0,H^*_0,f^*_0)$ is a WCD.
		By a) above and using that $C_1$ contains only vertices in private components, we have that $C_1$ has no edges to $G'-\hcrown$.
		Hence, $C^*=C^*_0\cup C_1$ has no edges to $G'-\hcrown=G-H^*$.
	\item We deduce that each connected component of $C_1$ has weight less than $\lambda$ using that $(\chrgg)$ is a {\pbcod}.
		Also, each connected component of $C^*_0$ has weight less than $\lambda$ because $(C^*_0,H^*_0,f^*_0)$ is a WCD. 
		Thus, all connected components in $C^*=C^*_0\cup C_1$ have weight less than $\lambda$.
	\item By b) above we have that $f^*(C)\in N(C)$ for all $C\in \comp(C_1)$.
	Also, we have that $f^*(C)=f^*_0(C)\in N(C)$ for all $C\in \comp(C^*_0)$ as $(C^*_0,H^*_0,f^*_0)$ is a WCD.
\item By c) above we get that the $f^*$-weight of $h$ is at least $\lambda$ for all $h\in H_1$. 
Also, we have that $f^*$-weight of $h$ is at least $\lambda$ for all $h\in H^*$, using that $(C^*_0,H^*_0,f^*_0)$ is a WCD. 
\qedhere
	\end{enumerate}
\end{proof}

For the correctness of step~\ref{step:bip} it remains to show that after the deletion of $C_1$ and $H_1$ we still have a {\pbcod} of the resulting graph in hand.
Moreover, we need to ensure that the remaining graph has no connected component of weight less than $\lambda$, thus ensuring the necessary condition on the working graph $G'$.

\begin{lemma}
	\label{lemma::Step4Connected}
	In the algorithm {\algCrown} in step \ref{step:bip}, $\left(\calC \setminus \calC\left(\calQ_1\right), H \setminus H_1, \calR, g, g'\right)$ is still a {\pbcod} of $G' - \left(H_1 \cup C_1\right)$, where $C_1 = V\left(\calQ_1\right)$.
	Furthermore, each component of $G'-\left(H_1 \cup C_1\right)$ weighs at least $\lambda$ assuming that $G'$ does not contain any connected components of weight smaller than $\lambda$.
\end{lemma}
\begin{proof}
	Since $N_{G'}(C_1) \subseteq H_1$,
	we have that the $g$ and $g'$-assignments of $h\in H\setminus H_1$ remains the same,
	after the deletion.
	Also $\calR$ remains the same and $\calC$ only looses some whole connected components.
	From this, using that
	($\chrgg$) is a {\pbcod} of $G'$,
	it is easy to see that 
	$\left(\calC \setminus \calC\left(\calQ_1\right), H \setminus H_1, \calR, g, g'\right)$ is a {\pbcod} of $G'-H_1-C_1$.

	Now, we prove the second part of the lemma.	
	Suppose there was a connected component of weight less than $\lambda$ in $G'-H_1-C_1$.
	Let this component be $S$.
	Since $G'$ did not have any connected component of weight less than $\lambda$, we have
	that $S$ has edge to either $H_1$ or $C_1$.
	It cannot have edge to $C_1$, as then it would be part of $C_1$.
	It cannot have edge to $H_1$, because then also it will be part of $C_1$.
\end{proof}

Next, we analyze the running time of step~\ref{step:bip}.

\begin{lemma}
	In the algorithm {\algCrown}, step \ref{step:bip} runs in time $\bigO(|V|\,|E|)$.
\end{lemma}

\begin{proof}
	\label{lemma::runTimeStep4}
	Finding a {\expanse} in $G'$ costs $\bigO(|V|\,|E|)$ by Lemma~\ref{lemma::separatorPacking} and \cref{lemma::cycleCanceling}.
	All remaining steps within step \ref{step:bip} take no more time than $\bigO(|V|\,|E|)$.
\end{proof}

Finally, we show that the outer and inner index stays unchanged after this step.

\begin{lemma}
	\label{lemma:step4OuterInner}
	In the algorithm {\algCrown}, after step \ref{step:bip} the outer and inner index stays unchanged.
\end{lemma}
\begin{proof}
	Recall that the outer index corresponds to $|H^*| + |H| + |\calR|$ and the inner index to $|H^*| + |H|$. 
	Since we remove $H_1$ from $H$ and add to $H^*$, and $\calR$ is unmodified, the lemma holds.
\end{proof}

\subsection{Step \uproman{5}: Private Assignment}
\label{sec::step5}

\def\algPriv{\texttt{AssignPriv}}
\def\algDeficit{\texttt{FillDeficit}}
\def\algMergeSC{\texttt{MergeSubComp}}
\def\hatG{\widehat{G}}
\def\hatV{\widehat{V}}

Let $(\oldchrgg)$ be the {\pbcod} $L$ at the start of the step.
In this step, we give an algorithm that given $(\oldchrgg)$ and a function $f:\qpriv_0 \to H$ such that the $f$-weight of each $h \in H$ is at most $3\lambda- 3$ (such an $f$ is given by Step~\ref{step:bip}) as input, outputs a {\pbcod} $(\chrgg)$ such that:
\begin{enumerate}
	\item \label{prop::1} $(\chrgg)$ is a {\pbcod}, and
	\item \label{prop::2} all sub-components in $\cCpriv_0$ are $g$-assigned \label{Zcond:cpriv}, and
\end{enumerate}

We give this step in the form of 3 subroutines {\algPriv}, {\algMergeSC}, and {\algDeficit} that will be executed in that order.
In {\algPriv} we $g$-assign all the private components, but may violate the property that a sub-component is $g$-assigned to its neighbor. In {\algMergeSC} we rectify this by merging some sub-components.
In {\algDeficit}, we fill up the $g$-assignments with the non-private sub-components so that the weight conditions required for {\pbcod} are satisfied.

Let $\qpriv_0$ and $\cCpriv_0$ be the set of private components and private sub-components respectively in $\calC_0$.
The algorithm {\algPriv} will create an assignment $g$ from $\calC_0$ to $H$ such that 
all the private sub-components in $\cCpriv_0$ are $g$-assigned, and 
the following conditions are satisfied for each $h$ (we call them the $h$-conditions):
\begin{enumerate}
		\item
		\label{hcond::1}
		$N^g[h]$ is connected,
		
		\item 
		\label{hcond::2}
		$w^g[h]$ is at most $3\lambda - 3$, and
		
		\item 
		\label{hcond::3}
		either $w^g[h] \geq 2 \lambda -1$, or $\goldinv(h) \cap \cCpriv_0 \subseteq \ginv(h)$.
\end{enumerate}
First of all, we will create for each $h \in H$ an initial assignment of $g$ according to $f$.
That is, for every $h \in H$ we $g$-assign to $h$ all the sub-components of each $Q \in \finv(h)$.
At this point, all sub-components in $\cCpriv_0$ are $g$-assigned, and we will maintain this invariant during the remaining algorithm.
We also satisfy $h$-conditions 1 and 2 throughout.
The algorithm terminates when $h$-condition 3 is satisfied.
The correctness and run-time analysis are performed below each step.

\paragraph{Algorithm \algPriv:}
\begin{enumerate}
		\item \label{Priv:1}
			For each  $Q\in \qpriv_0$, and for each sub-component $C$ of $Q$, we set $g(C)=f(Q)$. 

			Note that $N^g[h]$ is connected for every $h \in H$ by definition of $f$. 
		Also, $w^g[h] = w^f[h] \leq  3\lambda - 3$ for every $h \in H$.
		Thus $h$-conditions 1 and 2 are satisfied for all $h\in H$.
		Also, all the private sub-components are now $g$-assigned.
			Note that it is possible now that $g(C) \notin N(C)$ for some $C \in \ginv(h)$.
			This will only be rectified in {\algMergeSC}.
		\newline
		\emph{Running time:} $\bigO(|V|)$.
		
		\item
			Construct an auxiliary graph $\widehat{G} = (\widehat{V}= H \cup \cCpriv_0, \widehat{E})$ where $\hatE$ contains the pair $C_1,C_2 \in \calC_0$ as an edge if there is an edge between $C_1$ and $C_2$ in $G$,
			and the pair $C \in \calC_0,h\in H$ as an edge if $C$ has an edge to $h$ in $G$. 

		Consequently, if $\widehat{G}[H' \cup \calC_0']$ is connected for some $\calC_0' \subseteq \calC_0$ and $H' \subseteq H$, then $G[H' \cup V(\calC_0')]$ is also connected.
		\newline
		\emph{Running time:} $\bigO(|E|)$.
		
		\item
		\label{pre5::3}
		For each $h\in H$, construct a spanning tree $T_h$ rooted at $h$ of the graph
		$\hatG[\left\{ h \right\}\cup \ginv(h)]$.

		Note that $\cCpriv_0 = \bigcup_{h \in H} \widehat{V}(T_h-h)$ and that $T_{h}$ and $T_{h'}$ are disjoint for distinct $h$ and $h'$.
		Also observe that the subtree of $T_h$ rooted at $x\in \hatV(T_h-h)$ has weight at most $\lambda$ because it is completely contained in a connected component of $G[\calC_0]$ due to its connectedness. 
		\newline
		\emph{Running time:} $\mathcal{O}(|E|)$.

	\item
		As long as there is an $h\in H$ whose $h$-condition 3 is violated:
		since $h$-condition 3 is violated,
		we know that 
		$\goldinv(h) \cap \cCpriv_0 \nsubseteq \ginv(h)$ and hence
		there exist some $\tC\in \left(\goldinv(h)\cap \cCpriv_0\right)\setminus \ginv(h)$.
		We know $\tC\in \ginv(h')$ for some $h'\neq h$ because all sub-components in $\cCpriv_0$ are $g$-assigned by Step~\ref{Priv:1} of {\algPriv}. 
		This implies that $\tC$ is a vertex of $T_{h'}$.
		Consider the subtree rooted at $\tC$, i.e. $T_{h'}(\tC)$. 
		We update $g(C)=h$ for all $C\in \hatV(T_{h'}(\tC))$.
		We then remove $T_{h'}(\tC)$ from $T_{h'}$ and attach it to $T_h$ with $\tC$ as a direct child of $h$.
		Note that $\tC$ indeed has an edge to $h$ for such an attaching, because $g_0(\tC)=h$.

		To see that $h$-condition 1 is not violated, observe that we maintain for all $h\in H$ that $T_h$ contains exactly those vertices of $G$ in $N^g[h]$,
		and that $T_h$ still remains a tree.
		To see that $h$-condition 2 is not violated, observe that we only $g$-assign something to $h$ only if $w^g[h]<2\lambda-1$. 
		Also, the weight of the new additional assignment is at most $\lambda$,
		because we only move a proper subtree of $T_h'$ from $h'$ to $h$, and any proper subtree of $T_{h'}$ has weight at most $\lambda$.
		We do not $g$-unassign something that was $g$-assigned, so it is clear that all private sub-components remain $g$-assigned.
		Since step 4 is repeated as long as $h$-condition 3 is violated by some $h$, when the step terminates we have that $h$-condition 3 is satisfied.
		We show below that this step indeed terminates, and in fact is only repeated at most $\bigO(|V|)$ times.

		\emph{Running time:} $\mathcal{O}(|V||E|)$.
		 It is easy to see that each iteration of the step runs in $\bigO(|E|)$ time.
		 We will bound the number of iterations by $\bigO(|V|)$.
		 For this, observe that each $\tC$ will be picked only once as the root of the tree to be moved. 
		 This is because once it is picked then $g(\tC)=g_0(\tC)$, and 
		 it is afterwards not moved when some other vertex is picked, because $\tC$ now is a direct child of $h$ and does not occur in the subtree of any other vertex in $T_h-{h}$. 

\end{enumerate}

From the discussions at the end of each step of {\algPriv}, we have the following lemma.
\begin{lemma}
	{\algPriv} runs in $\bigO(|V||E|)$ time and when it terminates, the $h$-conditions are satisfied	
	and all the sets in $\cCpriv_0$ are $g$-assigned.
\end{lemma}

\paragraph{Algorithm \algMergeSC:} For each $h\in H$, and for each child $C$ of $H$: 
merge all the sub-components that are in the sub-tree of $T_h$ rooted at $C$ and add it as a single sub-component to $\calC$; then $g$-assign this component to $h$. 
Finally add all the sub-components in $\calC_0\setminus \cCpriv_0$ to $\calC_1$.

The merged sub-components have weight less than $\lambda$ as each proper sub-tree of $T_h$ has weight less than $\lambda$.
Also every vertex in $V(\calC_0)$ now appear in a sub-component in $\calC$ as each sub-component in $\cCpriv_0$ appears in some $T_h$.
Crucially, now for each $g$-assigned $C\in \calC$, we have $g(C)\in N(C)$.
Also, notice that each private sub-component of $\calC$ is formed by merging of private sub-components in $\calC_0$,
and the non-private sub-components of $\calC$ are the non-private sub-components of $\calC_0$.
\newline
\emph{Running time:} $\mathcal{O}(|E|)$.

\paragraph{Algorithm \algDeficit:}
For each $h\in H$ , if $w^g[h]\ge 2\lambda-1$ then assign $g'(h)=\varnothing$, otherwise do the following:
\begin{enumerate}[label=(\alph*)]
	\item	
		\label{item:Cp}
let $C_1,C_2,\dots,C_p$ be the sub-components in $\calC_0\setminus \cCpriv_0$ that were $g_0$-assigned to $h$, in decreasing order of weight.
\item
Starting with $i=1$, as long as $w^g[h]+w(C_i)\le 3\lambda-3$ and $i\le p$, we $g$-assign $C_i$ to $h$, and increase $i$ by $1$. 
\item
If $i\le p$ then assign $g'(h)=C_i$, otherwise assign $g'(h)$ the sub-component in $\calC$ that contains $g_0'(h)$ as a subset (such a sub-component should exist because we only merged the components in $\calC_0$ to get the components in $\calC$).
\end{enumerate}
\emph{Running time:} $\mathcal{O}(|E|+|V|\log |V|)$.
Here $|V|\log|V|$ is for the sorting in decreasing weight.

\begin{lemma}
	\label{lem:allCp}
	For some $h\in H$, if Step~\ref{item:Cp} of {\algDeficit} $g$-assigns all of $C_1,C_2,\dots,C_p$ to $h$, then each $C\in\ginv_0(h)$ is a subset of some $C'\in \ginv(h)$.			
\end{lemma}
\begin{proof}
	Since {\algDeficit} enters Step~\ref{item:Cp} for $h$, it should be the case that
	$w^g[h]$ was less than $2\lambda-1$ before {\algDeficit}.
	Then, since the $h$-condition 3 was satisfied after {\algPriv}, we have that all subcomponents in $C\in \ginv_0(h)\cap \left(\calC_0\cap \cCpriv_0\right)$ were $g$-assigned to $h$ after {\algPriv}.
	Also, we $g$-assign to $h$ all $C\in \ginv_0(h)\cap \left(\calC_0\setminus \cCpriv_0\right)$ in {\algDeficit}, by 
	assumption of the lemma.
	Hence, the lemma follows.
\end{proof}
\begin{lemma}
	\label{lem:filldefUB}
	After {\algDeficit}, we have $w^g[h]\le 3\lambda-3$ for all $h\in H$.
\end{lemma}
\begin{proof}
				We maintained the upper bound $3\lambda-3$ during {\algPriv} by $h$-condition 2.
				The $g$-weight does not change in {\algMergeSC}.
				In {\algDeficit}, we $g$-assign the $C_i$ one by one, only as long as the $g$-weight remains at most $3\lambda-3$.
				Thus we have the required upper bound.
\end{proof}
\begin{lemma}
	\label{lem:defpbcod}
	After {\algDeficit}, $(\chrgg)$	is a {\pbcod} and all private sub-components in $\calC$ are $g$-assigned.
\end{lemma}
\begin{proof}
	We first prove the 3 conditions for $(\chrg)$ to be a {\cod}.	
	\begin{enumerate}
		\item component condition:
			this is satisfied as $V\left( \calC \right)=V(\calC_0)$.
		\item $g$-neighbor condition: due to {\algMergeSC}, we have that $g(C)\in N(C)$ for all $g-$assigned $C$ in $\calC$. 
		\item $\calR$-condition: follows because we have not modified $\calR$ in this step of {\algCrown}.
			\end{enumerate}
			Next, we prove the 4 conditions required by {\pbcod}.  \begin{enumerate} \item	
				$g'$-neighbor condition:
				By {\algDeficit}, we have that if $g'(h)\neq \varnothing$ then $g'(h)$ is either a superset of $g_0'(h)$ or in $\ginv_0(h)$, and hence is in $\calN_\calC(h)$.
				So it only remains to prove that $g'(h)$ is not in $\ginv(h)$.
				If {\algDeficit} allocates one of the $C_i: i\in [p]$ as $g'(h)$ then this is indeed the case clearly as {\algDeficit} $g$-assigns only $C_1,\dots,C_{i-1}$.
				So assume the other case, i.e., when $g'(h)$ is the superset of $g_0'(h)$.
				Suppose
			$g'(h)\in \ginv(h)$ for the sake of contradiction.
			This means that all of $C_1,C_2,\dots,C_p$ was assigned to $h$ by {\algDeficit}.
				Then we claim that $w^g[h]\ge w^{g_0+g_0'}[h]$.
				By Lemma~\ref{lem:allCp},
				we have that $V(\ginv_0(h))$ is contained in $V(\ginv(h))$.
				Since $V(g_0'(h))\subseteq V(g'(h))$ is also contained in $\ginv(h)$ by assumption,
				we have that indeed $w^g[h]\ge w^{g_0+g_0'}[h]\ge 3\lambda-2$, a contradiction to \cref{lem:filldefUB}.
			\item
				$g$-weight condition: 
				the upper bound $3\lambda-3$ on the $g$-weight of any $h\in H$ is given by \cref{lem:filldefUB}.
				Now, we prove the lower bound $2\lambda-1$ on the $g$-weight of any $h\in H$.
				Suppose there is an $h$ violating this lower bound i.e., $w^g[h]<2\lambda-1$ after {\algDeficit}.
				This means that $w^g[h]<2\lambda-1$ before {\algDeficit} and hence Step~\ref{item:Cp} of {\algDeficit} was entered for $h$.
				If {\algDeficit} $g$-assigned all of $C_1,C_2,\dots,C_p$ to $h$ in Step~\ref{item:Cp}, 
				then by \cref{lem:allCp}, we have that $w^g[h]\ge w^{g_0}[h]\ge 2\lambda-1$, a contradiction.
				Thus there exists some $i\in [p]$ such that only $C_1,C_2,\dots,C_{i-1}$ was $g$-assigned to $h$ in Step~\ref{item:Cp}.
				But then $w^g[h]+w(C_i)\ge 3\lambda-2$, implying $w(C_i)\ge \lambda$, a contradiction.
			\item
				$(g+g')$-weight condition: 
				If $h$ has already a $g$-weight of at least $2\lambda-1$ after {\algPriv}, then observe that we don't $g$-assign anything to $h$ afterwards.
				So in that case all the $g$-assigned sub-components to $h$ are private sub-components.
				Then the {\ggwcond} is trivially satisfied.
				So assume that $h$ has $g$-weight less than $2\lambda-1$ after {\algPriv}.
				This means that in {\algDeficit}, Step~\ref{item:Cp} is entered for $h$.
				If {\algDeficit} allocates one of the $C_i: i\in [p]$ as $g'(h)$ then it is clear by construction that $w^{g+g'}[h]\ge 3\lambda-2$. 
				So assume that it $g$-assigns all of $C_1,C_2,\dots C_p$ to $h$ and $g'(h)\supseteq g'_0(h)$. 
				Then by using \cref{lem:allCp} and the properties of {\pbcod} $(\calC_0,H,\calR,g_0,h_0)$, we have
			that $w^{g+g'}[h]\ge w^{g_0+g'_0}[h]\ge 3\lambda-2$.
		\item
			{\splitcond}: 
				If $h$ has already a $g$-weight of at least $2\lambda-1$ after {\algPriv}, then observe that we don't $g$-assign anything to $h$ afterwards.
				So in that case all the $g$-assigned sub-components to $h$ are private sub-components.
				Then the {\splitcond} is trivially satisfied.
				So assume that $h$ has $g$-weight less than $2\lambda-1$ after {\algPriv}.
				Also, assume that the {\splitcond} condition is violated for the sake of contradiction.
				Then we have that $w^g[h]<2.5(\lambda-1)$ and there is a non-private sub-component in $\calC_1$ that has weight less than $0.5(\lambda-1)$.
				Let this sub-component be $C$.
				Now, note that we $g$-assign non-private sub-components only during {\algDeficit}.
				Thus, $C$ was $g$-assigned in {\algDeficit} and hence $C=C_i$ for some $i\in [p]$.
				But this means that $C$ was also $g_0$-assigned to $h$.
				Then, by the {\splitcond} of $(\calC_0,H,\calR,g_0,h_0)$ we have that the $g_0$-weight of $h$ is at least $2.5(\lambda-1)$.
				This means that we did not $g$-assign all of $C_1,C_2,\dots C_p$ to $h$ because otherwise by \cref{lem:allCp}, we have that $w^g[h]\ge w^{g_0}[h]\ge 2.5(\lambda-1)$, a contradiction.
				This means that $g'(h)=C_j$ for some $j>i\in [p]$.
				Also $w(C_j)\ge 0.5(\lambda-1)$ because otherwise we would have $g$-assigned $C_j$ also.
				This means that $C_1,C_2,\dots, C_{j-1}$ have weight at least $0.5(\lambda-1)$ as we arranged them in decreasing order of weight.
				In particular, $w(C_i)\ge 0.5(\lambda-1)$, a contradiction. 
	\end{enumerate}
	Now, it only remains to prove that the private sub-components of $\calC$ are $g$-assigned.
	This holds because all private sub-components of $\calC_0$ were $g$-assigned after {\algPriv} 
	and the private sub-components of $\calC$ are formed by the merging of private sub-components of $\calC$ in {\algMergeSC}, where each such merged sub-components was $g$-assigned.
\end{proof}
	Lastly, we point out that the outer-index or inner-index 
	does not change during this step of {\algCrown},
	since $H^*,H$ and $\calR$ are the same as before.

\subsection{Step \uproman{6}: Merge unassigned sub-components}
\label{sec::step6}

\def\algMergUa{\texttt{MergUnassgndSComp}}

In this step, we do the following algorithm:
\paragraph{Algorithm \algMergUa:}
As long as there is a sub-component $C\in \calC$ that is unassigned and not a whole connected component then:
\begin{enumerate}
	\item
		\label{it:merge}
Clearly there exists a $C'\in \calC$ such that there is an edge between $C$ and $C'$.
Merge $C$ and $C'$ into a single sub-component of $\calC$ (and remove $C$ and $C'$ from $\calC$).
Let $\tC$ be the new merged component.
If $C'$ was $g$-assigned to some $h$, then $g$-assign $\tC$ to $h$.
If now $g$-weight of $h$ is at least $3\lambda-2$, then:
\begin{enumerate}[label=(\alph*)]
	\item
		\label{it:remove}
		Let $C_1,C_2,\dots, C_p$ be the \emph{non-private} sub-components $g$-assigned to $h$ in \emph{increasing} order of weight.	
		(Note that $p\ge 1$ as $\tC$ is one such sub-component).
		Remove them one by one from $\ginv(h)$ in order, until $w^g[h]\le 3\lambda-3$.
		Assign $g'(h)$ as the last removed $C_i$. 
\end{enumerate}
\end{enumerate}

To see that in Step~\ref{it:merge}~\ref{it:remove}, we do get back to a $g$-weight of at most $3\lambda-3$ by the deletions, note that $\tC$ is a non-private sub-component and removing $\tC$ would make the $g$-weight at most $3\lambda-3$, as the $g$-weight went above $3\lambda-3$ only after the addition of $C'\subseteq \tC$ to the $g$-neighborhood of $h$.

\begin{lemma}
	\label{lem:mergua}
	After {\algMergUa}, the following holds:
	\begin{enumerate}[label=(\alph*)]
		\item
			$(\chrgg)$ is a {\pbcod},
		\item
	  all sub-components in $\cCpriv$ are $g$-assigned, and
  \item
	 for each $Q\in \comp(V(\calC))$, either all sub-components of $Q$ are $g$-assigned, or $Q$ itself is a single sub-component in $\calC$. 
	\end{enumerate}
	
\end{lemma}
\begin{proof}
	\begin{enumerate}[label=(\alph*)]
	\item	
		Since $V(\calC)$ and $\calR$ have not been modified, the {\compcond} and the {\Rcond} of {\cod} are satisfied.
		The {\gnbrcond} also remains satisfied because ${\tC}$ is the only newly assigned sub-component and it is $g$-assigned to a neighbor of it by construction. The $g$-weight of $h$ remains at most $3\lambda-3$ because of Step~\ref{it:merge}~\ref{it:remove}.
		It does not go below $2\lambda-1$ either, as we delete sub-components only as long as the $g$-weight is at least $3\lambda-2$ and the weight of the deleted sub-component is less than $\lambda$.
		Thus the $g$-weight condition is satisfied.

		Now let us prove the {\gdnbrcond}.
		Suppose this is violated.
		Therefore for some $h\in H$, we have $g'(h)\neq \varnothing$ and $g'(h)\notin \calN_{\calC}(h)\setminus \ginv(h)$.
		If $g'(h)$ was assigned during Step~\ref{it:merge}~\ref{it:remove}, then it is clear by construction that $g'(h)\in \calN_{\calC}(h)\setminus \ginv(h)$.
		Thus Step~\ref{it:merge}~\ref{it:remove} was not executed and
		$g'(h)$ was assigned before {\algMergUa}.
		Since {\gdnbrcond} was satisfied before, we know that $g'(h)\in \calN_{\calC}(h)$.
		Thus, we have $g'(h)\in \ginv(h)$.
		Since $g'(h)$ was not in $\ginv(h)$ before, this means that $g'(h)=C'$.
		But then by {\ggwcond} we know that adding $\tC$ to the $g$-neighborhood would have increased the $g$-weight to at least $3\lambda-2$.
		Thus Step~\ref{it:merge}~\ref{it:remove} was executed, a contradiction.

		Now, we prove {\ggwcond}.
		First note that we do not assign $g'(h)$ as $\varnothing$ in {\algMergUa}.
		The only point where we delete from $g$-neighborhood is in Step~\ref{it:merge}~\ref{it:remove}.
		Thus, $\wgg[h]$ can decrease only if Step~\ref{it:merge}~\ref{it:remove} is executed.
		But in this case by construction we have that $\wgg[h]\ge 3\lambda-2$.

		Now, we prove {\splitcond}.
		The only possibility where it can be violated is in Step~\ref{it:merge}~\ref{it:remove}.
		Suppose we removed $C_1,C_2,\dots,C_i$ in this step.
		If $w^g[h]$ is still at least $2.5(\lambda-1)$, then we are fine.
		So suppose it is lesser.
		Then $w(C_i)\ge 0.5(\lambda-1)$ as before the deletion of $C_i$, the $g$-weight of $h$ was at least $3\lambda-2$.
		Then we have that $w(C_j)\ge 0.5(\lambda-1)$ for all $j>i$ as $C_1, C_2,\dots, C_p$ are arranged in increasing order of weight.
		Thus all the non-private components that remain in $\ginv(h)$ have weight at least $0.5(\lambda-1)$, thus
		satisfying the {\splitcond}.

	\item We know this condition was satisfied before {\algMergUa} from the output of Step~\ref{step:cggmod} of {\algCrown}.
		Here in Step~\ref{step:merge-unassigned}, we do not $g$-unassign any private components.
		Note that $C,C',\tC$ and the sub-components we delete in 
		Step~\ref{it:merge}~\ref{it:remove} are all private sub-components.
		So this condition remains true.
	\item By construction, the algorithm {\algMergUa} terminates only when this condition is satisfied.
		It indeed terminates in at most $|V|$ iterations because at each iteration we merge two sub-components and therefore there can be only at most $|V|$ merges. \qedhere
	\end{enumerate}
\end{proof}
\begin{lemma}
	{\algMergUa} (and hence Step~\ref{step:merge-unassigned} of {\algCrown}) can be implemented in $\bigO(|V||E|)$ time.
\end{lemma}
\begin{proof}
	The number of iterations of the algorithm is at most $|V|$ because each iteration merges $2$ sub-components.	
	We can sort the sub-components by weight at the start of the algorithm.
	This only takes $\bigO(|V|\log |V|)$ per iteration.
	We only need to update the sorted list during each merge.
	This only takes $\bigO(|V|)$ per iteration.
	The remaining procedures in an iteration can be done in $\bigO(|E|)$ by standard methods.
\end{proof}

	Lastly, we point out that the outer-index or inner-index 
	does not change during this step of {\algCrown},
	since $H^*,H$ and $\calR$ are the same as before.

\subsection{Step \uproman{7}: Final Step}
\label{sec::step7}
The algorithm for Step~\ref{algFindDCr::6} is completely given in the description in Section~\ref{sub:algorithm_overview},
so we do not repeat it here.
We prove here the related lemmas that are required for correctness.
\begin{lemma}
	\label{lemma::Step6}
	In the algorithm {\algCrown} in Step~\ref{algFindDCr::6}, if $(\chrgg)$ is a {\fbcod} of $G'$, then the deduced $\left(\chrfstar\right)$ is a {\ccd} of $G$.
\end{lemma}
\begin{proof}
	Every time if we extend $C^*$, $H^*$ and $f^*$ in Step~\ref{algFindDCr::1} and Step~\ref{step:bip} we ensured that $C^*$, $H^*$ and $f^*$ are a WCD of $G$ by Lemma~\ref{lemma:weightedCrown} and Lemma~\ref{lemma:weightedCrown2}, respectively.
	That is, the first four conditions are satisfied in the definition of {\ccd}~\ref{def:lccd}.
	
	The remaining condition concerning $\calR^*$ is ensured by property of {\fbcod}, which we have already explained in Step~\ref{algFindDCr::6} in Section~\ref{sub:algorithm_overview}.
\end{proof}
Since Step~\ref{algFindDCr::6} is the only possibility for the algorithm to terminate, we have the following corollary. 
\begin{corollary}
	\label{cor:ifterm}
	If the algorithm terminates, then it terminates by giving a $\lambda$-BCD.
\end{corollary}

\subsection{Proof Conclusion of Balanced Crown Decomposition Theorem}

\def\calT{\mathcal{T}}

In this section, we conclude the proof of correctness and runtime of algorithm {\algCrown} as desired in Theorem~\ref{theorem:lccd}.
We show in the Sections~\ref{sub:steps_1_and_2}-\ref{sec::step7} that each step in {\algCrown} works correctly
and
that each step needs no more time than $\bigO(|V||E|)$.
The runtime is easy to see for Steps~\ref{algFindDCr::1},\ref{algFindDCr::2}~and~\ref{algFindDCr::6} and hence we did not state this explicitly. 
	In particular, we have shown that $(\chrgg)$ is a {\pbcod} at the end of any step throughout the algorithm, by Lemmata \ref{lem:initpbcod}, \ref{lem:divpbcod}, 
	\ref{lem:nocutcleanup},
	\ref{lem:CutPbcod},
	\ref{lemma::Step4Connected},
	\ref{lem:defpbcod}, and
	\ref{lem:mergua}.
For the correctness of the algorithm {\algCrown}, it only remains to show that it terminates,
because if it terminates then it gives a {\ccd} by \cref{cor:ifterm}.  
We will show now that in fact the algorithm can perform at most $k^2$ iterations.
Recall that $k=|H^*|+|\calR^*|$ according to Theorem~\ref{theorem:lccd}.
This will then also prove the required running time.

\begin{lemma}
	\label{lem:k2it}
	There are at most $k^2$ iterations of Algorithm {\algCrown}.
	Moreover, the outer index never decreases.
\end{lemma}
\begin{proof}
	We proved that the outer index never decreases during any step, and as long as outer index remains same, the inner index does not decrease, by Lemmata 
	\ref{lem:DivideIndex}, 
	\ref{lem:CutIndex},
	\ref{lemma:step4OuterInner}.
	Note that in Steps 
	\ref{step:cggmod}, and
	\ref{step:merge-unassigned}, and
	\ref{algFindDCr::6},
	we do not modify $H^*,\calR$ and $V(\calC)$ and hence it is clear that outer and inner indices remain same.

	Moreover, in each iteration, when Step~\ref{step:divorcut} is executed we make progress, i.e.~either the outer index $|H^*| + |H| + |\calR|$ increases (Lemma~\ref{lem:DivideIndex}) or the outer index remains the same and the inner index $|H^*| + |H|$ increases (Lemma \ref{lem:CutIndex}).

	Observe that $k=|H^*|+|\calR^*|=|H^*|+|H|+|R|$ at the time of termination in Step~\ref{algFindDCr::6} and hence outer-index and inner index are at most $k$ always.
	Thus there are at most $k^2$ iterations of Algorithm {\algCrown}.
\end{proof}
This concludes the proof of \cref{theorem:lccd}.
We conclude the section with three tool-lemmas that will be used in the applications in later sections to build connected-vertex partitions.

\begin{lemma}
	\label{lem::outerIndexPartition}
	In the algorithm {\algCrown}, there exists a $[\lambda,\infty)$-{\cvp} of $V$ of size equal to the outer index $|H^*| + |H| + |\calR|$ after any step in any iteration.
	Moreover, such a {\cvp} can be found in linear time given the {\pbcod} $(\chrgg)$ and the WCD $(C^*,H^*,f^*)$, maintained by the algorithm.
\end{lemma}
\begin{proof}
	First we show how to construct a $[\lambda,\infty)$-connected packing of $V$ of size $|H| + |\calR|$ from $(\chrgg)$.
	We have shown that $(\chrgg)$ is a {\pbcod} at the end of any step throughout the algorithm, by Lemmata \ref{lem:initpbcod}, \ref{lem:divpbcod}, 
	\ref{lem:nocutcleanup},
	\ref{lem:CutPbcod},
	\ref{lemma::Step4Connected},
	\ref{lem:defpbcod}, and
	\ref{lem:mergua}.
	Let $V_h := \{h\} \cup V\left(g^{-1}(h)\right)$ for every $h \in H$.
	Recall that by property of {\pbcod} we have $w(V_h) \geq \lambda$ for every $h \in H$ and $G[V_h]$ is connected, since $g(C') \in N(C')$ for every $C' \in \ginv(H)$.
	Moreover, $g$ is a function, i.e.~$\bigcap_{h \in H} V_h = \varnothing$.
	As a result, we obtain by $\calV_H:=\{V_h\mid h \in H\}$ a $[\lambda,\infty)$-connected packing of $H \cup V(\calC)$ of size $|H|$.
	By property of \pbcod, $\calR$ is a $[\lambda,\infty)$-connected packing of size $|\calR|$ of $V \setminus H \cup V(\calC)$.
	Thus, $\calV_H \cup \calR$ is a $[\lambda,\infty)$-connected packing of $V$ of size $|H| + |\mathfrak{R}|$. 
	
	Next, we show how to construct a $[\lambda,\infty)$-connected packing of $C^* \cup H^*$ of size $|H^*|$ using $f^*$.
	We will use that $(C^*,H^*,f^*)$ is always a weighted crown decomposition of $G$ (see Lemma~\ref{lemma:weightedCrown} and Lemma~\ref{lemma:weightedCrown2}).
	Let $V^*_h := \{h\} \cup V\left({f^*}^{-1}(h)\right)$ for every $h \in H^*$.
	Recall by the weighted crown decomposition we have $w(V^*_h) \geq \lambda$ for every $h \in H$ and $G[V^*_h]$ is connected, since $f(Q) \in N(Q)$ for every $Q \in \mathbb{CC}(C^*)$.
	Moreover, $f^*\colon \mathbb{CC}(C^*) \to H$ is a function, i.e.~$\bigcap_{h \in H} V^*_h = \varnothing$.
	As a result, $\calV^*_H:=\left\{V^*_h\mid h \in H^* \right\}$ is a $[\lambda,\infty)$-connected packing of $C^* \cup H^*$ of size $|H^*|$.
	
	Thus $\calT:=\calV_H \cup \calR\cup\calV^*_H$ is a 
	$[\lambda,\infty)$-connected packing of $V$ of size equal to the outer index $|H^*| + |H| + |\calR|$. 
	Observe that $V(\T)$ includes all vertices $V(\calR) \cup H \cup H^* \cup V\left(\finv(H)\right)$.
	That is, only vertices in $V(\calC)$ are not included in $V(\T)$.
	Just add those vertices to some vertex sets in $\T$ such that $\T$ becomes a $[\lambda,\infty)$-{\cvp} of $V$ of size $|H^*| + |H| + |\calR|$.
	Note that such a {\cvp} is guaranteed, since each component in the working graph $G'$ in the algorithm weighs at least $\lambda$ and each component $\comp(V(\calC))$ weighs less than $\lambda$.
\end{proof}

\begin{lemma}
	\label{lemma::Hpacking}
	Let $G = (V,E,w)$ be a vertex-weighted graph.
	Let $\lambda \in \mathbb{N}$ and let $(C, H, f, \mathfrak{R})$ be a {\ccd} of $G$.
	There exists a $[\lambda,\infty)$-{\cvp} of $C \cup H$ of size $|H|$.
\end{lemma}

\begin{proof}
	Let $V_h := \{h\} \cup V(f^{-1}(h))$ for every $h \in H$.
	Recall by property of {\ccd} we have $w(V_h) \geq \lambda$ for every $h \in H$ and $G[V_h]$ is connected, since $f(Q) \in N(Q)$ for every $Q \in \mathbb{CC}(C)$.
	Moreover, $f\colon \mathbb{CC}(C) \to H$ is a function, i.e. $\bigcap_{h \in H} V_h = \varnothing$.
	As a result, we obtain a $[\lambda,\infty)$-{\cvp} of $C \cup H$ by $\{V_h\mid h \in H\}$ of size $|H|$.
\end{proof}

\begin{lemma}
	\label{lemma::PackingInLambdaComponentDecomposition}
	Let $G = (V,E,w)$ be a vertex-weighted graph.
	Let $\lambda \in \mathbb{N}$ and let $(C, H, f, \mathfrak{R})$ be a {\ccd} in $G$.
	There exists a $[\lambda,\infty)$-{\cvp} of $V$ of size $|H| + |\calR|$.
\end{lemma}

\begin{proof}
	We make use of Lemma \ref{lemma::Hpacking} and obtain a $[\lambda,\infty)$-{\cvp} $\calV_H$ of $C \cup H$ of size $|H|$. 
	By property of {\ccd}, $\calR$ is a $[\lambda,3(\lambda-1)]$-{\cvp} of size $|\calR|$ of $V \setminus (C \cup H)$.
	Thus, $\calV_H \cup \calR$ is a $[\lambda,\infty)$-{\cvp} of $V$ of size $|H| + |\mathfrak{R}|$. 
\end{proof}

	\section{Applications of Balanced Crown Decomposition}
		\label{sec:appendix-applications}
		In the following subsections we give the detailed versions of the applications of the balanced crown decomposition.

	\subsection{Kernel for \texorpdfstring{$W$}{W}-weight Separator}

The prototype application for crown decompositions is kernelizations for vertex deletion problems. As our first application of the balanced crown decomposition we therefore consider the \textsc{$W$-weight Separator} problem. For  $W \in \mathbb{N}$, a set $S \subseteq V$ is called a \emph{$W$-weight vertex separator} for a vertex-weighted graph $G=\left(V,E,w\right)$ if each connected component of $G - S$ has weight less than $W$.
By {\wsep} we denote the parameterized problem to decide for instances $\left(G,W,k\right)$ if there exists a $W$-weight vertex separator for $G$ of size at most $k$. We say $S$ is a solution for $\left(G,W,k\right)$ if $|S| \leq k$ and $S$ is a $W$-weight vertex separator for $G$.

We investigate the parameterized complexity of {\wsep} and consider the parameter $k+W$. Recall that {\wsep} is not fixed parameter tractable by considering $k$ or $W$ individually as parameter, as it is para-\textbf{NP}-hard for parameter~$W$ and is  \textbf{W[1]}-hard for parameter~$k$ by~\cite{drange2014computational}.

Let $\left(G,W,k\right)$ be an instance of {\wsep}.
To simplify the algorithm, we assume that $W > w_{\max}$, since such vertices have to be included in any solution.
Moreover, we remove all connected components that have a weight less than $W$ in $G$.
Note that this is a valid reduction rule, since those components require no further separation.
Thus, we can assume that each component in $G$ has a weight of at least $W$.
That is, we satisfy the input of the algorithm {\algCrown} to obtain a {\wccd} {$\left(\chrf\right)$} of $G$ (see Definition \ref{def:lccd}) which yields the following formal kernelization algorithm.

\paragraph{Algorithm \algkernelsep:}
We first call the algorithm {\algCrown} with $W$ to obtain a {\wccd} {$\left(\chrf\right)$} of $G$.
If during the algorithm {\algCrown} the outer index becomes greater than $k$, then just cut-off {\algCrown} and output a trivial no-instance.
Otherwise, output $\left(G - \left(C \cup H\right),W,k - |H|\right)$.
\newline

We start by observing that $|H| + |\calR| > k$ implies that $(G,W,k)$ is a no-instance.
Since the outer index only increases (\cref{lem:k2it}) and finally becomes $|H| + |\calR|$ (Theorem \ref{theorem:lccd}), such a lemma ensures the correctness of this step. 

\begin{lemma}
	If $|H| + |\calR| > k$, then $\left(G,W,k\right)$ is a no-instance.
\end{lemma}

\begin{proof}
	We make use of Lemma \ref{lemma::PackingInLambdaComponentDecomposition} and obtain a $[W,\infty)$-{\cvp} of $V$ of size $|H| + |\mathfrak{R}|$.
	Since the vertex sets in this partition are disjoint, we need at least one vertex in each part for a feasible solution.
	This fact concludes the lemma.
\end{proof}

For the correctness of {\algkernelsep} it remains to show that $\left(G - \left(C \cup H\right),W,k - |H|\right)$ is a valid reduction rule for $(G,W,k)$, and that we obtain the desired kernel in case $|H| + |\calR| \leq k$.
We point out, that the following proof uses the standard techniques to derive a kernel from a crown decomposition (see e.g.~\cite{fomin2019kernelization, kumar20162, xiao2017linear}). Nevertheless, the full proof is provided, for the sake of completeness.

\begin{lemma}
	\label{theorem::reducible}
 	$\left(G,W,k\right)$ is a yes-instance if and only if $\left(G - \left\{H \cup C\right\},W,k - |H|\right)$ is a yes-instance.
 	Furthermore, $|H| + |\calR| \leq k$ implies $w\left(G - \left(H \cup C\right)\right) \leq 3k\left(W-1\right)$. 	
\end{lemma}

\begin{proof}
	Let $S$ be a solution of size $k$ for the instance $\left(G,W,k\right)$.
	Consider the partition of $S$ into $S_c = S \cap \left(C \cup H\right)$ and $S_r=S\setminus S_c$.
	By Lemma \ref{lemma::Hpacking} there is a $[W,\infty)$-{\cvp} $\calV_H$ of $C \cup H$ of size $|H|$.
	Since $S$ is a feasible solution, we obtain $S \cap V' \neq \varnothing$ for each $V'\in\calV_H$. 
	Moreover, the $|H|$ vertex sets in $\calV_H$ are pairwise disjoint, hence $|H| \leq |S_c|$.	Consequently,  we obtain $|S_r| \leq |S| -|H| = k-|H|$. 	
	Let $G' = G - \left(H \cup C\right)$.
	We claim that $S_r$ is a solution of the instance $\left(G',W,k - |H|\right)$, i.e.~$S_r$ is a $W$-weight vertex separator in $G'$.
	Assuming that $S_r$ is not a $W$-weight vertex separator in $G'$, there exists at least one connected vertex set $D$ in $G'$, with $w\left(D\right) \geq W$ and $D \cap S_r = \varnothing$.
	Since $S$ is a solution of $G$, we obtain that $S \cap D \ne \varnothing$ which contradicts $D \cap S_r = \varnothing$ as $S\setminus S_r = S_c \subseteq C \cup H$ and $D \subset V\left(G\right) \setminus \left(C \cup H\right)$.
	Thus, $S_r$ is a $W$-weight vertex separator in $G'$ of size at most $k-|H|$.
	
	Conversely, let $S'$ be a solution for the instance $\left(G - \left\{H \cup C\right\},W,k - |H|\right)$ of size $k - |H|$.
	Since the components in $\mathbb{CC}(C)$ have weight less than $W$ and $H$ separates $V \setminus \left(H \cup C\right)$ from $C$, it is easy to see that $S' \cup H$ has size $k$ and is a solution for $\left(G,W,k\right)$.
	
	The second part of the lemma can be seen as follows. By $|H| + |\mathfrak{R}| \leq k$ we obtain $|\mathfrak{R}| \leq k - |H|$.
	Since $\mathfrak{R}$ is a $\left[W,3(W-1)\right]$-{\cvp} of $V \setminus \left(C \cup H\right)$, the reduced instance $\left(G - \left(C \cup H\right),W, k - |H|\right)$ satisfies $w\left(G - \left(C \cup H\right)\right) \leq 3|\mathfrak{R}|\left(W-1\right) \leq 3\left(k-|H|\right)\left(W-1\right)$.
\end{proof}

\paragraph{Running time:}	
The algorithm {\algCrown} runs in $\bigO\left({\widetilde{k}}^2|V|\,|E|\right)$ (see Theorem \ref{theorem:lccd}).
Since we have at most $k+1$ iterations in {\algCrown}, (i.e.~the outer index is bounded by $k + 1$), we obtain $\bigO\left({\widetilde{k}}^2|V|\,|E|\right) \subseteq \bigO\left(\left(k+1\right)^2|V|\,|E|\right) = \bigO\left(k^2|V|\,|E|\right)$.
This completes the proof of the following theorem.

\WSepKernelThm*


	\subsection{Kernelization and Approximation for  \texorpdfstring{$W$}{W}-weight Packing}

In a similar fashion as in the previous section, a balanced crown decomposition can be used to derive a kernelization for connected packing problems -  in a linear programming sense, the dual of the  \textsc{$W$-weight separator} problem. For this type of problem, the packing provided by our balanced crown decomposition can also be used to derive an approximate solution. 
Formally, we consider the following notion of packing. 
For $W \in \mathbb{N}$, we call $\T = \left\{T_1, \dots, T_{\ell}\right\}$, a \emph{$W$-packing} of a vertex-weighted graph $G=\left(V,E,w\right)$ if $T_i \subseteq V$, $G\left[T_i\right]$ is connected and $w\left(T_i\right) \geq W$ for all $i \in [\ell]$, and $T_i\cap T_j=\varnothing$ for all $i,j\in[\ell]$ with $i\neq j$.

We denote the optimization problem to find a $W$-packing of maximum size by {\textsc{WPack}}. By {\wpack} we denote the parameterized problem to decide for instances  $\left(G,W,k\right)$ if there exists a $W$-packing of size at least $k$ in $G$. We say $\T$ is a solution for $\left(G,W,k\right)$, if $|\T| \geq k$ and $\T$ is a $W$-packing  for $G$. 
We show that a balanced crown decomposition directly gives a $3k\left(W-1\right)$-kernel for {\wpack}, as well as a 3-approximation for {\textsc{WPack}}. 

We start with the kernel result, so let $\left(G,W,k\right)$ be an instance of {\wpack}. As reduction rule, we first remove all connected components in $G$, of weight less than $W$. Note that this is valid, since such components cannot be part of any solution. Thus, we assume that the components in $G$ have weight at least $W$. With this, we satisfy the input of the algorithm {\algCrown} to obtain a $W$-BCD {($\chrf$)} of $G$ (see Definition \ref{def:lccd}), which formally yields the following kernelization algorithm.

\paragraph{Algorithm \algkernelpack:}
We first call the algorithm {\algCrown} with $W$ to obtain a $W$-BCD {$\left(\chrf\right)$} of $G$.
If during the algorithm {\algCrown} the outer-index becomes $k$, then just cut-off {\algCrown} and output a trivial yes-instance.
Otherwise, output $\left(G - \left(C \cup H\right),W,k - |H|\right)$.
\newline

Since the outer-index only increases (\cref{lem:k2it}) and finally becomes $|H| + |\calR|$ (Theorem \ref{theorem:lccd}), the following lemma shows the correctness in case $|H| + |\calR| \geq k$.

\begin{lemma}
	\label{lemma::yesInstance}
	If $|H| + |\calR| \geq k$, then $\left(G,W,k\right)$ is a yes-instance.
\end{lemma}

\begin{proof}
	We make use of Lemma \ref{lemma::PackingInLambdaComponentDecomposition} and obtain a $[W,\infty)$-{\cvp} of $V$ of size $|H| + |\mathfrak{R}|$ which is a solution for  the {\wpack} instance $\left(G,W,k\right)$.
\end{proof}

Next, we show a more general relation between a $W$-BCD $\left(\chrf\right)$ and a $W$-packing $\T$ for $G$, that will be useful for both the kernelization and the later approximation result. Denote for a $W$-packing $\T$ of a vertex-weighted graph $G=\left(V,E,w\right)$ and a subset $H\subseteq V$ by $\T^H$ the subset of $\T$ of sets intersecting $H$, i.e. $\T^{H} := \left\{T \in \T| T \cap \left(C \cup H\right) \ne \varnothing \right\}$. Then the structure of the  $W$-BCD yields the following connection.

\begin{lemma}
	\label{lemma::HTreePack}
	Let $G$ be a vertex-weighted graph and $W\in\mathbb N$. For any $W$-BCD $\left(\chrf\right)$ and $W$-packing $\T$ for $G$, it follows that  $\left|\T^{H}\right| \leq |H|$.
\end{lemma}

\begin{proof}
Let $\left(\chrf\right)$ be a $W$-BCD of $G$ and  $\T$ be a $W$-packing for $G$. By the separator property of a $W$-BCD, the set $H$ separates $C$ from $V \setminus \left(C \cup H\right)$, and the connected components of $G[C]$ have weight less than $W$.
	Thus, every $T \in \T^{H}$ contains at least one vertex of $H$; recall that each set $T$ in $\T$ has to be connected and weigh at least $W$.
	Since the sets in $\T$ are disjoint and $\T^{H}\subseteq \T$, every $h \in H$ is included in at most one set in $\T^{H}$. As, by definition, at least one vertex from $H$ is in each set in $\T^H$, it follows that $\left|\T^{H}\right| \leq |H|$.
\end{proof}

For the correctness of the algorithm {\algkernelpack} it remains to show that the reduction in case $|H| + |\calR| < k$ is valid, and gives a kernel of the desired weight.

\begin{lemma}
	\label{lemma::reduciblePacking}
	$\left(G,W,k\right)$ is a yes-instance if and only if $\left(G - \left(C \cup H\right),W,k - |H|\right)$ is a yes-instance.
	Furthermore, if $|H| + |\calR| < k$, then $w\left(G - \left(C \cup H\right)\right) \leq 3k\left(W-1\right)$.
\end{lemma}

\begin{proof}
	Assume $(G,W,k)$ is a yes-instance and let $\T = \left\{T_1, \dots, T_k\right\}$ be a corresponding solution for $(G,W,k)$.
	Denote by $\T^{\overline{H}} := \left\{T \in \T| T \cap \left(C \cup H\right) = \varnothing \right\}$, i.e $\T^{H}, \T^{\overline{H}}$ is a partition of $\T$. Observe that by definition $\T^{\overline{H}}$ is a $W$-packing for $G - \left(C \cup H\right)$. By  Lemma \ref{lemma::HTreePack}, it follows that $\left|\T^{H}\right| \leq |H|$ which yields $\left|\T^{\overline{H}}\right|=\left|\T\right|-\left|\T^H\right|\geq k-|H|$.
	As a result, $\T^{\overline{H}}$ is a feasible solution of size at least $k - |H|$ in $G - \left(C \cup H\right)$, hence $\left(G - \left(C \cup H\right),W,k - |H|\right)$ is a yes-instance.
	
	Conversely, assume  $\left(G - \left(C \cup H\right),W,k - |H|\right)$ is a yes-instance and let $\T' = \left\{T'_1, \dots, \allowbreak T'_{k - |H|}\right\}$ be a solution for  $\left(G - \left(C \cup H\right),W,k - |H|\right)$.
	By Lemma \ref{lemma::Hpacking} we obtain a $[W,\infty)$-{\cvp} $\calV_H$ of $C \cup H$ of size $|H|$.
	Thus, $\T' \cup \calV_H$ is $W$-packing of $G$ of size $k$, hence $\left(G,W,k\right)$ is a yes-instance.
	
		The bound on the kernel size can be seen as follows. By $|H| + |\mathfrak{R}| < k$ we obtain $|\mathfrak{R}| \leq k - |H|$.
Since $\mathfrak{R}$ is a $\left[W,3\left(W-1\right)\right]$-{\cvp} of $V \setminus \left(C \cup H\right)$, the reduced instance $\left(G - \left(C \cup H\right),W, k - |H|\right)$ satisfies $w\left(G - \left(C \cup H\right)\right) \leq 3\left|\mathfrak{R}\right|\left(W-1\right) \leq 3\left(k-|H|\right)\left(W-1\right)$.
\end{proof}

\paragraph{Running time:}	
The algorithm {\algCrown} runs in $\bigO\left(\widetilde{k}^2|V|\,|E|\right)$ (see Theorem \ref{theorem:lccd}).
Since we have at most $k$ iterations in {\algCrown}, (i.e.~the outer-index is bounded by $k$), we obtain $\bigO\left(\widetilde{k}^2|V|\,|E|\right) \subseteq \bigO\left(k^2|V|\,|E|\right)$.
This completes the proof of the following theorem.

\WPackKernelThm*


Lastly, we show that a $W$-BCD can also be used to derive a 3-approximation for \textsc{WPack}.

\begin{lemma}\label{lem::packing_apx}
	Let $G=\left(V,E,w\right)$ be a vertex-weighted graph and $W\in \mathbb N$. If there exists a  $W$-packing of size $k$ for $G$, then a $W$-packing of size at least $\left\lceil k/3 \right\rceil$ can be derived from any  $W$-BCD for $G$ in linear time. 
\end{lemma}

\begin{proof}
	Let $\T = \left\{T_1, \dots, T_{k} \right\}$ be a  $W$-packing for $G$ and $\left(\chrf\right)$ be a  $W$-BCD for $G$.
	By Lemma \ref{lemma::HTreePack} we have that  $\left|{T}^H\right| \leq |H|$. Denote again  ${\T}^{\overline{H}}=\T\setminus \T^H$. Since  ${\T}^{\overline{H}}\cup \T^H$ is a partition of $\T$, we obtain $\left| {\T}^{\overline{H}} \right| = k - \left| {\T}^H \right| \leq k - |H|$ and therefore, $k \leq \left| {\T}^{\overline{H}} \right| + |H|$.
	
Since $\calR$ is a $[W,3(W-1)]$-{\cvp} of $V \setminus \left(C \cup H\right)$ it follows that $w\left(V \setminus \left(C \cup H\right)\right) = w\left(V\left(\calR\right)\right) < 3\left|\calR\right|W$. Any packing that includes only vertices of $V \setminus \left(C \cup H\right)$ consequently has size at most $3\left|\calR\right| - 1$.
	Since we have $T \subseteq V \setminus \left(C \cup H\right)$ for every $T \in {\T}^{\overline{H}}$, we obtain $\left| {\T}^{\overline{H}} \right| < 3 \left|\calR\right|$.
	
	To construct a $W$-packing of cardinality at least $\left\lceil k/3 \right\rceil$ from the $W$-BCD, note that $\calR$ is a $W$-packing of $V \setminus \left(C \cup H\right)$.	We make use of Lemma \ref{lemma::Hpacking} and obtain a $[W,\infty)$-{\cvp} $\calV_H$ of $C \cup H$ of size $|H|$.
	As a result, $\calR \cup \calV_H$ is a $W$-packing of $V$ of size $|\calR| + |\calV_H|$ of $G$ as $V\left(\calR\right) \cap V\left(\calV_H\right) = \varnothing$. 
	Finally, $k \leq \left| {\T}^{\overline{H}} \right| + |H| < 3|\calR| + |\calV_H|$ yields the claimed size of at least $\left\lceil k/3 \right\rceil$  for the  derived $W$-packing $\calR \cup \calV_H$.
\end{proof}

Computing a $W$-BCD $(\chrf)$ for a vertex-weighted graph $G$ can be done in $\bigO\left(\widetilde{k}^2|V|\,|E|\right)$ by Theorem~\ref{theorem:lccd} where $k=|H|+|\calR|$. Since any $W$-BCD $\left(\chrf\right)$ yields a $W$-packing of size $|H|+|\calR|$, $k$ cannot be larger than the size of a maximum $W$-packing for $G$. Hence, Lemma~\ref{lem::packing_apx} in particular shows the following.

\WPackApxThm*


	\subsection{Approximations for BCP}
		Our third example of problems that benefit from the structure of a balanced crown decomposition, are connected partition problems.
		We consider the problems  {\maxminbcp} and {\minmaxbcp}, and first recall their definition.
		
		For a vertex-weighted graph $G=(V,E,w)$ and $V' \subseteq V$, we call a partition $\T$ of $V'$ a \emph{$\cvp_k$ of $V'$} if $\T$ is a {\cvp} of $V'$ with $|\T| = k$.
		In {\maxminbcp} we search, on input $(G,k)$ with $k\in \mathbb N$, for  a $\cvp_k$  $\T = \{T_1, \dots, T_k\}$ of $V$ such that	$\min_{i \in [k]} w(T_i)$ is maximized.	Conversely, as the name might suggest, in {\minmaxbcp} the value $\max_{i \in [k]} w(T_i)$ should be minimized.
	
	\subsubsection{Approximation for Max-Min BCP} 
	

Let $(G,k)$ be an instance of {\maxminbcp}.
We can assume that $G$ has at most $k$ connected components and at least $k$ vertices.
Otherwise, the instance does not have a feasible solution.
Furthermore, we denote the smallest weight of the connected components in $G$ by~$\wmin$.
Let $X^*$ be the optimal value for the instance $(G,k)$.
Note that $X^*\le \min\left(w(G)/k, \wmin\right)$.
For any given  $X \leq \min\left(w(G)/k, \wmin\right)$, we now use our balanced partition (as formally stated in the algorithm {\algmaxmin} below) to either give a $[X/3, \infty)$-{$\cvp_{k}$}, or report that $X>X^*$.
Once we have this procedure in hand, a binary search for the largest $X$ in the interval $\left(0, \min\left( \left\lceil w(G)/k \right\rceil , \wmin\right) \right]$ for which we find a $[X/3, \infty)$-{$\cvp_{k}$} can be used to obtain an $[X^*/3, \infty)$-{$\cvp_{k}$}, costing only an additional running time of factor $\bigO\left(\log \frac{w(G)}{k}\right)$ as $X^*\le \min\left(w(G)/k, \wmin\right)$.
Note that using algorithm {\algCrown} in the following procedure is valid, since all connected components of the input graph have a weight at least $W_{\min}$ and we only consider values $X$, with $X \leq \min\left(w(G)/k, \wmin\right)$.  

\paragraph{Algorithm \algmaxmin:}
We first call {\algCrown} with $\lambda=\lceil X/3\rceil$ to obtain a {\ccd} {($\chrf$)} of $G$.
If during the algorithm {\algCrown} the outer index becomes $k$, then just cut-off {\algCrown} and output an $[X/3, \infty)$-{$\cvp_{k}$} of $V(G)$ by using ~\cref{lem::outerIndexPartition}.
Otherwise, report that $X>X^*$.

Note that the outer-index in {\algCrown} only increases (\cref{lem:k2it}). 
Moreover, we point out that in case the outer-index is $k'<k$ after the termination of {\algCrown}, the algorithm has computed a {\ccd} {($\chrf$)} for $\lambda=X$ with $k' = |H| + |\calR| < k$.
The correctness of {\algmaxmin} is proven through the following lemma.

\begin{lemma}
	\label{lemma::Max-MinKeyLemma}
	If Algorithm {\algmaxmin} terminates with outer index $k'<k$, then $X>X^*$.
\end{lemma}
\begin{proof}
Let {($\chrf$)} be the computed  {\ccd} for $\lambda=\lceil X/3\rceil$, and note that $k' = |H| + |\calR| < k$. Suppose $k'<k$ and $X\le X^*$. 	
	Let $\T^* = \left\{T^*_1, \dots, T^*_k\right\}$ be an optimal solution of $(G,k)$, i.e., $\T^*$ is a connected partition of $V(G)$ with $w(T^*_i) \geq X^*$ for all $i \in [k]$, and $\min_{i \in [k]} w(T_i)=X^*$.
	
Consider the partition of $\T^*$ into $\T^{C \cup H} := \left\{T^*_i \in \T^*|  T^*_i \cap (C \cup H) \ne \varnothing \right\}$ and $\T^{\overline{C \cup H}} := \left\{T^*_i \in \T^*|  T^*_i \cap \left( C \cup H \right) = \varnothing \right\}$.
	We show that $\left|\T^{C \cup H}\right| \leq |H|$ and $\left|\T^{\overline{C \cup H}}\right| \leq |\calR|$, implying that $\left|\T^*\right|\leq |H|+|\calR|=k'<k$, contradicting $\left|\T^*\right|=k$.

	First we show $\left|\T^{C \cup H}\right| \leq |H|$.
	For this, it is sufficient to prove that $T^*_i\cap H\neq \varnothing$ for each $T^*_i \in \T^{C \cup H}$, since $\T^*$ is a partition of $V$.
	We prove this by contradiction to the properties of the  {\ccd} {($\chrf$)}.
	Suppose there is a $T^*_i\in \T^{C \cup H}$, such that $T^*_i\cap H=\varnothing$.
	This implies that $T^*_i\subseteq C$, as
	$H$ separates $C$ from $V \setminus (C \cup H)$. Moreover, $T_i^*$ is completely contained in some connected component of $G[C]$ as $T_i^*$ is connected. Hence, $w(T_i^*)<X$ because the connected components of $G[C]$ have weight less than $\lambda=X$.
	Thus, $w(T_i^*)<\lambda=\lceil X/3\rceil\le \lceil X^*/3\rceil$, which is a contradiction to  $\min_{i \in [k]} w(T_i)=X^*$.
	
	Next, we show $\left|\T^{\overline{C \cup H}}\right| \leq |\calR|$.
	From the properties of the {\ccd} {($\chrf$)}, we know that $w(R)\le 3\lambda-3 $ for each $R \in \calR$.
	Thus, $w(R)\le 3\lambda-3=3\lceil X/3\rceil-3\le 3\lceil X^*/3\rceil-3\le X^*$.
	Hence, $w(V \setminus (C \cup H)) \leq |\calR| X^*$.
	Since each $T^*_i \in \T^*$ weighs at least $X^*$, we obtain $\left|\T^{\overline{C \cup H}}\right| \leq |\calR|$.
	\end{proof}
	
\paragraph{Running time:}	
	The algorithm {\algCrown} runs in $\bigO\left(\widetilde{k}^2|V|\,|E|\right)$ (see Theorem \ref{theorem:lccd}).
	Since we have at most $k$ iterations in {\algCrown}, (i.e.~the outer-index is bounded by $k$), we obtain $\bigO\left(\widetilde{k}^2|V|\,|E|\right) \subseteq \bigO\left(k^2|V|\,|E|\right)$.
	We modify slightly the binary search to optimize the running time. 
	Let $g(\ell) := 2^{\ell}$, $\mathbb{N} \ni \ell \geq 1$.
	We increase stepwise $\ell$ in $g(\ell)$ until we find an $\ell^*$ with $g(\ell^*) < X^* \leq \min\left(g\left(\ell^* + 1\right), w(G)/k \right) =: \widehat{X} \leq 2X^*$.
	Afterwards, we perform a binary search  in the interval $X \in \left[g(\ell^*),\widehat{X}\right]$.
	As a result, the algorithm runs in $\bigO\left(\left(\log X^* + log X^* \right) (k^2 |V|\,|E|)\right)$ and thus, we obtain a running time in $\bigO\left(log\left(X^*\right)(k^2 |V|\,|E|)\right)$.
	This completes the proof of the following theorem.

\BCPMaxMinApxThm*

Finally, we point out that Theorem \ref{theorem::MaxMinAPX} provides also an approximation algorithm for the problem of finding a connected $k$-subgraph Max-Min edge partition in an edge-weighted graph $G$. Formally the problem \textsc{MaxMin Balanced Connected Edge Partition} (\textsc{MaxMinBCEP}) searches on an edge-weighted graph $G=(V,E,w_E)$ with $w\colon E \to \mathbb{N}$ and $k\in \mathbb N$ as input, for a partition  $\{E_1, \dots, E_k\}$  of $E$ such that $G[E_i]$ is connected for each $i\in [k]$ and that $\min_{i\in [k]} w_E(E_i)$ is minimal.
The best known approximation for this problem is~2 by~\cite{borndorfer2019approximatin}, however this result only holds for instances $G=(V,E,w_E)$ where $w_E(e)\leq  w_E(E)/2k$ for each $e\in E$.
Our following result only gives an approximation ratio of~3, but it holds without restrictions on the weights.  

We can simply use our approach above, since we can consider the line graph as vertex-weighted graph; note that connected vertices in the line graph can be represented as connected edges in the original graph. This directly yields.

\begin{theorem}
	A 3-approximation for \textsc{MaxMin BCEP} problem can be computed in  $\mathcal{O}\left(\log\left(X^*\right) k^2|V|\,|E|\right)$, where $X^*$ denotes the optimum value. 
\end{theorem}

	\subsubsection{Approximation for Min-Max BCP}\label{sec::minmaxapx}
	

For the Min-Max objective it is not so trivial to derive an approximation from a {\ccd}. The main problem is that, in contrast to the Max-Min case, an optimal solution can (and sometimes has to) build more than $|H|$ sets from the vertices in $H\cup C$. With the connectivity constraints, this means that some components in $G[C]$ are in fact a set in the optimal partition. Hence, when computing an approximate solution from a balanced crown decomposition, we have to also choose some components from $G[C]$ to be sets, while others are combined with some vertex in $H$. In order to make the decision of where to place the components in $G[C]$, we use a min-cost flow on a network that models the options for components in  $G[C]$ to either be sets or be combined with some vertex in $H$. With this additional use of a cost-flow network, our balanced crown structure can be used to derive a 3-approximation for {\minmaxbcp}.

We use our crown decomposition by fixing a target value~$X$ and deriving from a {\ccd} with $\lambda=X$ either a $(0,3X)$-{$\cvp_{k}$} or deducing that $X < X^*$.  As already indicated, we need more sophisticated further computations with the cost-flow network to find an approximate solution or decide $X < X^*$ for the Min-Max objective. The overall structure however remains as for the Max-Min case, that with such a procedure (as formally stated in {\algminmax} below), a binary search for $X$ in the interval $\left[\max\left(\left\lceil w(G)/k \right\rceil,w_{\max}\right) , w(G)\right]$ can be used to obtain a $(0,3X^*)$-{$\cvp_{k}$}, and hence a 3-approximation, costing only an additional running time factor in $\bigO\left(\log \left( w(G) - \frac{w(G)}{k} \right) \right)$ as $X^* \geq \max\left(w(G)/k,w_{\max}\right)$.

To apply our balanced crown decomposition, first observe the following characteristics for any instance $(G,k)$ of {\minmaxbcp}. We can assume that $G$ has at most $k$ connected components and has at least $k$ vertices, as otherwise the instance does not have a feasible solution. Also, the optimum value  $X^*$ for instance  $(G,k)$ satisfies $X^*\ge \max\left(\left\lceil w(G)/k \right\rceil, w_{\max}\right)$, where $w_{\max}=\max_{v\in G(V)} w(v)$. When searching for a solution of a fixed value $X$, we can also assume that $G$ has only connected components of weight at least $X$, since we may put connected components of $G$ with lower weight than $X$ to the desired {\cvp} and accordingly update $k$ by subtracting the number of those components. This yields a valid input for our algorithm {\algCrown} to compute a {\ccd} of $G$ with $\lambda=X$. Furthermore, if we have a $(0,3\lambda)$-{$\cvp_{k'}$} with $k' < k$, one can easily construct an $(0,3\lambda)$-{$\cvp_{k}$} by splitting some partitions in $(0,3\lambda)$-{$\cvp_{k'}$}.
(This is possible because we can also take a single vertex as one vertex set of the partition).

One step of the binary search for $X$ looks as follows on a high level. (The precise construction of the network $\Nw_{\lambda}$ and how it is used to build a connected partition follows in Definition~\ref{def::network} and Lemma~\ref{lemma::fact3}.

\paragraph{Algorithm \algminmax:}
Call {\algCrown} with $\lambda=X$ to obtain a {\ccd} {($\chrf$)} for $G$.
If during the algorithm {\algCrown} the outer index becomes greater than $k$, cut-off {\algCrown} and report $X < X^*$. Otherwise use the constructed {\ccd} {($\chrf$)} to build the  cost-flow
 network $\Nw_{\lambda}$ and compute a min-cost flow $Y^*$ for it.
If $p(Y^*) + |\calR| > k$, report $X < X^*$. Otherwise, compute from $Y^*$ a $(0,3X)$-{$\cvp_{k'}$} of $V(G)$ where  $k' \leq k$.
\newline

We point out that we use the first report of $X < X^*$ in case the outer index becomes greater than $k$ to optimize the running time. The correctness of this decision can be easily seen as follows.

\begin{lemma} 
If there is a {\ccd} {($\chrf$)} for $G$ with $\lambda=X$ and $k' = |H| + |\calR| > k$, then $X < X^*$ for the optimum value $X^*$ of $(G,k)$.
\end{lemma}

\begin{proof}
Note that $X^*$ being the optimum value of the {\minmaxbcp} instance $(G,k)$ in particular, implies that $w(G)\leq kX^*$. Conversely, any 
 {\ccd}  {($\chrf$)} for  $G$ with $\lambda=X$ and $k' = |H| + |\calR|$ yields a $[X, \infty)$-{$\cvp_{k'}$} of $V(G)$ by Lemma~\ref{lemma::PackingInLambdaComponentDecomposition}. Any $[X, \infty)$-{$\cvp_{k'}$} of $V(G)$ in turn  implies $w(G)\geq k' X$, which yields $k'X\leq kX^*$, hence $k'> k$ implies $X<X^*$.
 \end{proof}

To create the network $N_\lambda$ recall that we denote a flow in a network with edges $\overrightarrow{E}$ by $Y$, where $Y = \{y_{\overrightarrow{e}} \in \mathbb{N}_0 \mid\overrightarrow{e} \in \overrightarrow{E}(\Nw_{\lambda})\}$.
We will use capacitated min-cost flow, which means that the input network has both a capacity and a cost function on its edges. For the capacity, we use a function $c\colon \overrightarrow{E} \to \mathbb{N}$. For the cost, we use a linear function $p_{\overrightarrow{e}}\colon  \mathbb{Q}\to  \mathbb{Q}$ for each $\overrightarrow{e}\in \overrightarrow{E}$. Corresponding networks are hence given by tuples of the form $\Nw=(V,\overrightarrow{E},c,p)$. For a flow $Y$ on such a network,  we denote by $p(Y):=\sum_{\dir{e}\in \overrightarrow{E}} p_{\overrightarrow{e}}(y_{\overrightarrow{e}})$ the \emph{cost of $Y$} and by $v(Y) := \sum_{\dir{e} \in \delta^-(t)} y_{\overrightarrow{e}}$  the \emph{value of $Y$}.
For a fixed \emph{required flow value} $F\in \mathbb N$, we call $Y$ a \emph{feasible flow} if $Y$ satisfies the capacity constraints, the flow conservation for every vertex in $V(N) \setminus \{s,t\}$, and if $v(Y) = F$.

We now define the cost-flow network $\Nw_\lambda$ that we use to decide an assignment for components in $G[C]$ either to a vertex in $H$ or to themselves (indicating that they are sets in the balanced partition). In the network, each component  $Q \in\mathbb{CC}(C)$ corresponds to a node~$q$ that can route a flow of capacity equal to the weight of its component $w(Q)$ from the source. From~$q$, the flow can either be routed through a copy of~$q$ to the target (indicating $Q$ is a set in the solution) or through a node representing a vertex in $H$ that is connected to $Q$ (indicating that $Q$ lands in a set connected via a vertex in $H$). Capacities for the arc leading from a node $h\in H$ to $t$ model assignments from  $Q \in\mathbb{CC}(C)$ to build around $h$ a component of weight $\lambda$. We choose the costs on the arcs leading to $t$ such that an assignment of the whole weight $w(Q)$ going to $q'$, and the whole weight $\lambda$  yields~1 (modeling that we pay one unit from~$k$ by making $Q$ a set in the solution). Similarly, we choose the cost for the vertices in~$H$. The network could be defined without the copies $q'$ and also with easier weights (especially without the added constants), however this more complicated formulation allows an easier comparison with the optimum value for {\minmaxbcp} on $(G,k)$ in the proofs that follow. The formal definition of the network is given below and we recommend comparing it with Figure~\ref{figure::networkLambda}.

\begin{definition}[$\Nw_{\lambda}$]\label{def::network}
For a  {\ccd} ($\chrf$)  of $G$, the network $\Nw_{\lambda}$ is given by the tuple $\left(H \cup \calQ \cup \calQ' \cup \{s, t\}, \overrightarrow{E}, c, p\right)$ defined as follows.\begin{itemize}
\item $\calQ$ and $\calQ'$ both represent $\mathbb{CC}(C)$, i.e.~for each component  $Q$ in $G[C]$ there are corresponding vertices $q\in\calQ $ and  $q'\in \calQ'$. 
\item  $s$ and $t$ are sink, and source vertex, and $H$ is a copy of the set $H$ in the {\ccd}
\item For each $q \in \calQ$ there is an arc $\overrightarrow{s q}$ with capacity $c(\overrightarrow{sq})=w(Q)$ (where we always use $Q$ to denote the component in $\mathbb{CC}(C)$ corresponding to $q$) and cost function $p_{\dir{sq}}\equiv 0$.
\item For each $q \in \calQ$ there is an arc $\dir{qh}$ for each $h \in N(Q)$ with capacity set to $c(\dir{qh})=w(Q)$ and  cost function $p_{\dir{qh}}\equiv 0$.
\item Every $q \in \calQ$ is connected to its copy $q' \in \calQ'$ through an arc $\overrightarrow{qq'}$ with capacity $c(\overrightarrow{qq'})=w(Q)$ and cost function $p_{\overrightarrow{qq'}}\equiv 0$.

\item For each $h \in H$ there is an arc $\overrightarrow{ht}$ with capacity $c(\overrightarrow{ht})=\lambda - w(h)$ and $p_{\dir{ht}}(y)= \frac 1\lambda (w(h) + y)$ as cost function.

\item For each $q' \in \calQ'$ there is an arc $\overrightarrow{q't}$ with capacity $c(\overrightarrow{q't}) = w(Q)$ and  $p_{\dir{q't}}(y)=y/w(Q)= y/w(Q')$ as cost function.
\end{itemize}
Further, we set the required flow value for $\Nw_{\lambda}$ to $F := \sum_{Q \in \mathbb{CC}(C)} w(Q)$.
\end{definition}

Observe the following properties that $\Nw_\lambda$ directly inherits from the  {\ccd} ($\chrf$). By the bound on the weights of the components in $G[C]$, it follows that $w(Q) < \lambda$ for every $Q \in \mathbb{CC}(C)$.
Since every $Q \in \calQ$ is connected to at least one vertex in $H$, each $q\in\calQ$ is connected to at least one $h\in H$. Note also that we can assume that $\lambda - w(h) \geq 0$ by $\lambda \geq \max\left(w(G)/k,w_{\max}\right)$, which gives a valid capacity to these arcs in the above definition. 
Moreover, we point out that the additional costs $w(h)/\lambda$ for arcs in $\left\{\dir{ht} \in \dir{E}\mid h \in H \right\}$ are constants in $\Nw_{\lambda}$, and we introduce them only to simplify some later proofs.

Note that all arcs except $\overrightarrow{e} \in \delta^-(t)$ have  cost zero.
Moreover, note that once we reach vertices from $H$ or $\calQ'$ by a flow, then the flow conservation implies that this flow is passed on to $t$, which in turn means that we have to pay for it.
Lastly, we point out that for every $q \in \calQ$ there is only one arc that leads from $q$ to a vertex in $\calQ'$ which is the arc that leads to the vertex $q'$ that corresponds to the same component in $\mathbb{CC}(C)$.
Hence if we refer for a $q \in \calQ$ about the arc $\dir{qq'}$ or its copy $q'$, then this is unique and needs no further explanation. 
 \begin{figure}[h]
\begin{center}
	\begin{tikzpicture}[y=0.80pt, x=0.80pt]
	\node at (-115,30) {$s$};
	\node at (235,30) {$t$};
	\node at (0,-20) {$\mathcal Q$};
	\node at (120,-20) {$\mathcal Q'$};
	\node at (60, 140){$H$};
	\node at (-10,10) {$q_1$};
	\node at (-10,40) {$q_2$};
	\node at (-10,70) {$q_3$};
	\node at (110,10) {$q'_1$};
	\node at (110,40) {$q'_2$};
	\node at (110,70) {$q'_3$};
	\node at (65,110) {$h_1$};
	\node at (65,90) {$h_2$};
	
	\node at (-50,60) {\textcolor{blue}{$w(Q_i)$}, \textcolor{red}{0}};
	\node at (60,-20){\textcolor{blue}{$w(Q_i)$}, \textcolor{red}{0}};
	\node at (160,100) {\textcolor{blue}{$\lambda-w(h_j)$},};
	\node at (190,82) {\textcolor{red}{$w(h_j)/\lambda+y_{\overrightarrow{h_jt}}/\lambda$}};
	\node at (190,0) {\textcolor{blue}{$w(Q_i)$},};
	\node at (190,-15) {\textcolor{red}{$y_{\overrightarrow{q'_it}}/w(Q_i)$}};
	\node at (10,100){\textcolor{blue}{$w(Q_i)$},};
	\node at (10,85){\textcolor{red}{0}};
	
	\path[fill=black, line width=0.25mm] (-100,30) circle (0.09cm);
	\path[fill=black, line width=0.25mm] (220,31) circle (0.09cm);
	\path[fill=black, line width=0.25mm] (0,30) circle (0.09cm);
	\path[fill=black, line width=0.25mm] (0,60) circle (0.09cm);
	\path[fill=black, line width=0.25mm] (0,0) circle (0.09cm);
	
	\path[fill=black, line width=0.25mm] (120,60) circle (0.09cm);
	\path[fill=black, line width=0.25mm] (120,30) circle (0.09cm);
	\path[fill=black, line width=0.25mm] (120,0) circle (0.09cm);
	\path[fill=black, line width=0.25mm] (60,100) circle (0.09cm);
	\path[fill=black, line width=0.25mm] (60,120) circle (0.09cm);
	
	\path[draw=black, line width=0.25mm, -latex] (-100,30)--(-2.5,29);
	\path[draw=black, line width=0.25mm, -latex] (-100,30)--(-2.5,59);
	\path[draw=black, line width=0.25mm, -latex] (-100,30)--(-2.5,1);
	
	\path[draw=black, line width=0.25mm, -latex] (0,60)--(117.5,60);
	\path[draw=black, line width=0.25mm, -latex] (0,30)--(117.5,30);
	\path[draw=black, line width=0.25mm, -latex] (0,0)--(117.5,0);
	
	\path[draw=black, line width=0.25mm, -latex] (120,60)--(218,33);
	\path[draw=black, line width=0.25mm, -latex] (120,30)--(217,30);
	\path[draw=black, line width=0.25mm, -latex] (120,0)--(217,27);
	\path[draw=black, line width=0.25mm, -latex] (60,100)--(218,33);
	\path[draw=black, line width=0.25mm, -latex] (60,120)--(220,36);
	
	\path[draw=black, line width=0.25mm, -latex] (0,60)--(58,117);
	\path[draw=black, line width=0.25mm, -latex] (0,30)--(56.5,98);
	\path[draw=black, line width=0.25mm, -latex] (0,0)--(59,96.5);
		\end{tikzpicture}
\end{center}
\caption{Min-cost flow network $N_{\lambda}$ resulting from a {\ccd} ($\chrf$) of $G$ through $H$ and $\mathbb{CC}(C) = \calQ$ with corresponding \textcolor{blue}{capacities} and \textcolor{red}{costs}.  Note that $Q_i$ is the component in $G[C]$ corresponding to vertex $q_i$ and that $y_{e}$ denotes the flow through edge $e$}
\label{figure::networkLambda}
\end{figure}
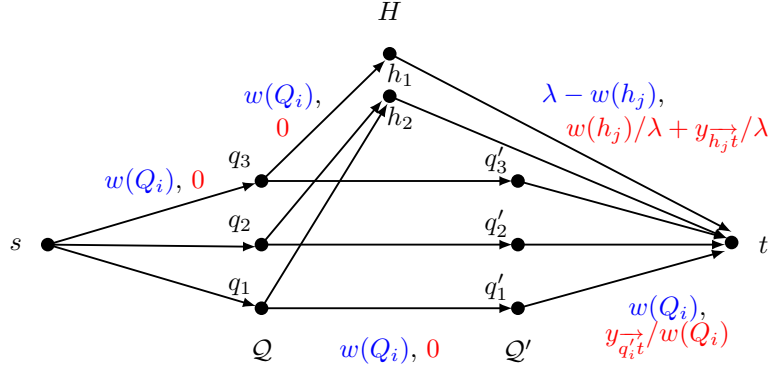
First, observe that there exists a feasible flow in~$Y$, and that it gives us a full assignment of the weight of each component $Q$.

\begin{lemma}\label{lemma::feasibleFlow}
	There exists a feasible flow $Y$ for $\Nw_{\lambda}$ w.r.t.~the required flow $F$. 
	Furthermore, a feasible flow $Y$ satisfies $y_{\overrightarrow{s q}} = c\left(\overrightarrow{s q}\right) = w(Q)$ for every $q \in \calQ$.
\end{lemma}
\begin{proof}
	Obviously, we obtain $y_{\overrightarrow{s q}} = c\left(\overrightarrow{s q}\right) = w(Q)$ for every $q \in \calQ$, since $\sum_{q \in \calQ} w(Q) = F$ is the required flow.
	Finally, we may easily build a feasible solution with the required flow by sending $w(Q)$ flow via paths $\overrightarrow{sq}, \overrightarrow{q'}, \overrightarrow{q't}$ for all $q \in \calQ$.
\end{proof}

In order to build a best possible assignment, we want the components to favor vertices in $H$ which translates to a min-cost flow saturating the capacities from each $h\in H$ to $t$. The properties of the balanced crown decomposition used to build $\Nw_\lambda$ can be used to show this.

\begin{lemma}
	\label{lemma::SaturatedCapacities}
Any min-cost flow $Y^*$ of $\Nw_{\lambda}$ satisfies $y^*_{\overrightarrow{ht}} = c\left(\overrightarrow{ht}\right) = \lambda - w(h)$ for every $h \in H$.
\end{lemma}
\begin{proof}
Let $Y^*$ be a min-cost flow  of $\Nw_{\lambda}$, and suppose there exists an $h_0 \in H$ with $y^*_{\overrightarrow{h_0t}} < \lambda - w(h_0)$.
	Let $R$ be the residual digraph resulting from $Y^*$, where we denote the arc weights by $r\colon \dir{E}(R) \to \mathbb{N}_0$.
	That is, for every $\dir{e} \in \dir{E}$ we have $r\left(\dir{e}\right) = c\left(\dir{e}\right) - y^*_{\dir{e}}$ and $r\left(\overleftarrow{e}\right) = y^*_{\dir{e}}$ for the reverse direction.	We remove all arcs $\dir{e} \in \dir{E}(R)$ with $r\left(\dir{e}\right) = 0$ from $R$.	Note that $\dir{sq} \notin \dir{E}(R)$ for every $q \in \calQ$ by Lemma~\ref{lemma::feasibleFlow}.
	
	Consider a  $q'_\ell$-$t$-path  in the residual network $R$ of the form $P = q'_\ell,q_\ell,h_{\ell-1},q_{\ell-1},\dots, \allowbreak h_1,q_1,h_0,t$ where $\ell\in\mathbb N$ and  $q_i$ denotes some vertex in $\calQ$ and $q'_i$  its copy in $\calQ'$, for each $i\in[\ell]$. Observe that if such a path exists in $R$, then we can reduce the flow costs of $Y^*$ adjusting with respect to the residual graph, since $1/\lambda < 1/\max_{q \in \calQ} w(Q)$.	(Note that $\lambda > w(Q)$ for every $q \in \calQ$ by the definition of a {\ccd}.)
	Thus, such a path would contradict the optimality $Y^*$.
	
We show that $y^*_{\overrightarrow{h_0t}} < \lambda - w(h_0)$ yields the existence of such a path $P$. To this end, consider the graph $R'$ obtained by deleting $t$ and $s$ from $R$. If we now find a path of the form $P' = q'_\ell,q_\ell,h_{\ell-1},q_{\ell-1},\dots,h_1,q_1,h_0$ in this graph, then the assumption $y^*_{\overrightarrow{h_0t}} < \lambda - w(h_0)$ yields the existence of $\dir{h_0 t}$ in $R$ to build the path $P$. In order to find such a path $P'$ pointing to $h_0$, we reverse the direction of all edges in $R'$ and start a breadth first search from $h_0$. The structure of the balanced crown decomposition yields a contradiction to this search not reaching any $q'\in\calQ'$.

Denote by $\hr$ the graph obtained from $R'$ by reversing all edges. Let $T_{h_0}$ be the arborescence $T_{h_0}$ rooted at $h_0$ that we find by a breadth first search from  $h_0$. Observe that a path in $T_{h_0}$ alternates between vertices $h \in H$ and $q \in \calQ$, so once we reach some vertex $q' \in \calQ'$ in $T_{h_0}$ we have shown the existence of a path $P'$. We claim that $\calQ' \cap T_{h_0} = \varnothing$ leads to a contradiction to the properties of the {\ccd} ($\chrf$) underlying $\Nw_{\lambda}$ .
	
Notice that, for a component $Q \in \calQ = \mathbb{CC}(C)$ and an $h \in H$, the assignment  $f(Q) = h$ implies $h \in \delta_{\Nw_{\lambda}}^+(Q)$ by the construction of  $\Nw_{\lambda}$.  
	We define for every $h \in H$ the mismatches comparing $Y^*$ with $f$,
	i.e.~$\calQ^h_{Y^*}:= f^{-1}(h) \cap \left\{q \in \calQ\mid \overrightarrow{q h} \in \overrightarrow{E}(\Nw_{\lambda})\mid y_{\overrightarrow{q h}} < c\left(\overrightarrow{q h}\right) \right\} \subseteq \delta_{R}^-(h) = \delta_{\hr}^+(h)$, 	
	and for every $q \in \calQ$ the distributed outgoing flow from $Y^*$ through $q$ to vertices in $H \cup \calQ'$,
	i.e.~$V_{Y^*}^{q}:= \left\{v \in H \cup \calQ'\mid \overrightarrow{q v} \in \overrightarrow{E}(\Nw_{\lambda})\mid y_{\overrightarrow{q v}} > 0 \right\} \subseteq \delta_{R}^-(q) = \delta_{\hr}^+(q)$.
	Notice also that, for $q\in \calQ$ and $v \in H \cup \calQ'$ we consider the same arc $\dir{qv} \in \dir{E}(\Nw_{\lambda})$ from different directions, which leads to this relation in the residual graph $R$ and its reverse version $\hr$. Further, note that if  $q\in\calQ^h_{Y^*}$ for some  $h\in H$ and $q\in f^{-1}(h)$, then the arc $y_{\overrightarrow{ hq}}$ is  in $\hr$, and similarly, $y_{\overrightarrow{hq}}$ in $\hr$ implies that $h\in V_{Y^*}^{q}$.
	
	Let $\calQ_{T_{h_0}} := \calQ \cap V(T_{h_0})$ and  $H_{T_{h_0}} := T_{h_0}\setminus \calQ_{T_{h_0}}$.
	Since $T_{h_0}$ is exhausted and $\calQ' \cap V(T_{h_0}) = \varnothing$,
	we obtain $H_{T_{h_0}}\subseteq H$. Furthermore, the definition of $ V_{Y^*}^{q}$ implies $\bigcup_{q \in \calQ_{T_{h_0}}} V_{Y^*}^{q}= H_{T_{h_0}}$. Vice versa, the exhaustive search and the definition of $\calQ^h_{Y^*}$ implies $\bigcup_{h \in H_{T_{h_0}}} \calQ^h_{Y^*}\subseteq \calQ_{T_{h_0}}$. 
	
	By Lemma \ref{lemma::feasibleFlow} we have $y_{\overrightarrow{sq}} = c\left(\overrightarrow{sq}\right) = w(Q)$ for all $q \in \calQ$ which implies $\sum_{v \in \delta_{\Nw_{\lambda}}^+(q)} y^*_{\overrightarrow{q v}} = w(Q)$ by the flow conservation.	Moreover, we have for every $q \in \bigcup_{h \in H_{T_{h_0}}} f^{-1}(h)$ which is not in $\calQ_{T_{h_0}}$ that  $y^*_{\overrightarrow{q h}} = c\left(\dir{qh}\right)$ for some $h \in H_{T_{h_0}}$ by definition of $\calQ^h_{Y^*}$.	
	Thus, we obtain for every $q \in \bigcup_{h \in H_{T_{h_0}}} f^{-1}(h)$ that $\sum_{h \in \delta_{\Nw_{\lambda}}^+(q) \cap H_{T_{h_0}}} y^*_{\overrightarrow{q h}} = w(Q)$.
	
	From the definition of a {\ccd} we have $w(h) + w\left(f^{-1}(h)\right) \geq \lambda$ and $\bigcap_{h \in H} \{h\} \cup f^{-1}(h) = \varnothing$.	
	Thus, we obtain $\sum_{h \in H_{T_{h_0}}} \left(w(h) + w\left(f^{-1}(h)\right)\right) \geq \left|H_{T_{h_0}} \right| \lambda$.
	Since we assumed $y^*_{\overrightarrow{h_0 t}} + w(h_0)
	< c(\overrightarrow{h_0 t}) = \lambda$ and 
	$y^*_{\overrightarrow{h t}} + w(h) \leq w(h) + c\left(\dir{ht}\right)
	= \lambda$ for all $h \in H_{T_{h_0}} \setminus \{h_0\}$, we obtain $\sum_{h \in H_{T_{h_0}}} \left(y^*_{\overrightarrow{h t}} + w(h)\right) < \sum_{h \in H_{T_{h_0}}} \left(c(\overrightarrow{h t}) + w(h)\right) = |H_{T_{h_0}}| \lambda$.
	
For every $q \in \bigcup_{h \in H_{T_{h_0}}} f^{-1}(h)$  we have shown that $\sum_{h \in \delta_{\Nw_{\lambda}}^+(q) \cap H_{T_{h_0}}} y^*_{\overrightarrow{q h}} = w(Q)$ which leads by the flow conservation to $\sum_{h \in H_{T_{h_0}}} \left(y^*_{\overrightarrow{h t}} + w(h)\right) = \sum_{h \in H_{T_{h_0}}} \left(w(h) + w\left(f^{-1}(h)\right)\right) \geq |H_{T_{h_0}}| \lambda$.
	
	This concludes the contradiction to   $\calQ' \cap T_{h_0} = \varnothing$ and therefore shows the existence of a$h_0$-$q'_\ell$-path in $\hr$. Recall that such a path allows to  reduce the cost of $Y^*$, which is a contradiction to $Y^*$ being a min-cost flow.	As a result, we obtain $y^*_{\overrightarrow{ht}} = c(\overrightarrow{ht}) = \lambda - w(h)$ for every $h \in H$.		
\end{proof}

Now we show how the cost of the min-cost flow $Y^*$ is linked to the optimum value $X*$ of $\minmaxbcp(G,k)$, which justifies our second decision to report $X<X^*$ in the algorithm {\algminmax}.

\begin{lemma}
	\label{lemma::forBinarySearchMinMax}	
	If $p(Y^*) + |\calR| > k$, then $X < X^*$.
\end{lemma}
\begin{proof}
	We show this by contraposition, i.e.~if $X \geq X^*$, then $p(Y^*) + |\calR| \leq k$.
	Let $\T^* = \{T^*_1, \dots, T^*_k\}$ be an optimal solution of an instance $\minmaxbcp(G,k)$.
	Let $\T_C := \{T \in \T^*\mid  T \cap (V \setminus C) = \varnothing\}$ and let $\T'_H := \{T \in \T^*\mid T \cap H \ne \varnothing\}$. Denote by $\T_H$ the set of components obtained from $\T'_H $ by deleting  $V \setminus (C \cup H)$ from each $T \in \T'_H$.
	
	 Since $H$ separates $C$ from the rest of the graph, it immediately follows that $\T_C$ contains exactly the $T\in \T^*$ with $T \subseteq C$. Further, removing $V \setminus (C \cup H)$ does not separate sets in $\T'_H $ which shows that  $\T_C \cup \T_H$ is a connected partition of $C \cup H$.

	Our intermediate goal is to prove $|\T_C| + \frac{w\left(\bigcup_{T \in \T_H} T\right)}{X} \geq p(Y^*)$.
	To show this, we construct a feasible flow $Y$ in $\Nw_{\lambda}$ through $\T_C$ and $\T_H$ with cost at most $|\T_C| + \frac{w\left(\bigcup_{T \in \T_H} T\right)}{X} $ as follows.
	Recall that the required flow is $F = \sum_{Q \in \calQ} w(Q)$,	
	and that the vertices from $\calQ$ in $\Nw_{\lambda}$ correspond to the connected components of $G[C]$.
	
We represent $\T_C \cup \T_H$ as flow in $\Nw_{\lambda}$
by routing for every $q \in \calQ$ a flow of $w(Q)$ according to the partition $\T^{C \cup H} := \T_C \cup \T_H$ of $C \cup H$ in the following way:
	At first, set $y_{\dir{sq}} = w(Q) = c\left(\dir{sq}\right)$ for every $q \in \calQ$.
For each $T \in \T_C$, we route a flow of $w(T)$ trough $y_{\dir{qq'}}$ for the component $Q$ that contains $T$; recall that the sets in $T$ are subsets of some component in $G[C]$ which means that for each $T \in \T_C$ there is a unique component $Q(T)$ containing $T$. We also extend the routing towards $t$ and hence set for every $q \in \calQ$ that contains some $T\in \T_C$  the values  $y_{\dir{qq'}} = y_{\dir{q't}} = \sum_{T\in \T_C} w(T \cap Q)$. Observe that this is valid since  $\sum_{T\in \T_C} w(T \cap Q)\leq w(Q)=c\left(\dir{qq'}\right) = c\left(\dir{q't}\right)$. The total cost of this assignment is  $\sum_{T \in \T_C} \frac{w\left(T\right)}{w\left(Q(T)\right)}$.
		
For each $T\in \T_H$, we assign flow as follows. Denote by $H_T := H \cap T$ the nodes from $H$ in a $T\in \T_H$, and by $\calC_T := \{T \cap Q \mid Q\in \mathbb{CC}(C)\}$ the partition of the vertices from $C$ in $T$ with respect to the components in $\mathbb{CC}(C)$.	We denote by $Q(C')$ the unique $Q \in \calQ$ of which $C' \in \calC_T$ is a part,
	and by $\calQ(T)$ the set $\{Q \in \calQ\mid V(\calC_T) \cap Q \ne \varnothing\}$.	Note that for every $C' \in \calC_T$ we have $Q(C') \in \calQ(T)$.
	
Since $H$ separates $C$ from $V \setminus (C \cup H)$, $H$ separates $T^* \cap C$ from $V \setminus (C \cup H)$ for every $T^* \in \T'_H$. For the restriction $T\in  \T_H$ of $T^*$ to $C \cup H$, we hence find that for every $Q \in \calQ(T)$ there exists at least one arc $\overrightarrow{q h}$ in $\Nw_\lambda$ to some $h \in H_T$.
Hence we can, for every $T\in \T_H$,  send for every $C' \in \calC_T$ a flow of $w(C')$ through $\overrightarrow{q_{c'} h}, \overrightarrow{h t}$,
where $q_{c'}$ is the vertex in $\calQ$ corresponding to the component $Q(C')$ and $h \in \delta^+(Q(C')) \cap H_T$ chosen arbitrarily.
	Since the sets in $\T_H$ are disjoint, it follows  for each $Q \in \mathbb{CC}(C)$ that $\{T \cap Q\mid T\in \T_H\}$ is a packing of $Q$, hence $\sum_{T \in \T_H} w(Q \cap T) \leq w(Q)$ and we consequently do not violate any capacity constraint in $\dir{qh} \in \dir{E}(\Nw_{\lambda})$ with $h \in H$ and $q \in \calQ$ by our settings.
	
	Moreover, each vertex in $h$ occurs at most in one $T \in \T_H$, and further the sets $T$ are subsets of sets in the optimum solution $\T^*$, which implies that $w(T)\leq X^*$. We only route flow for sets $C'=Q \cap T$ to some $h\in T\cap H$, hence we route at most a total of $\sum_{Q\in \mathbb{CC}(C)} w(T\cap Q)$ to a vertex $h\in  T\cap H$.	With the assumption  $X \geq X^*$, we see that $c\left(\overrightarrow{h t}\right) = X- w(h) \geq X^* - w(h) \geq \sum_{Q\in \mathbb{CC}(C)} w(T\cap Q)$ for every $T \in \T_H$ and every $h \in H$, which means that we do not violate any capacity constraints on any edge $\dir{ht}$.
	Finally, the cost for this flow is 
	$\sum_{T \in \T_H} \left( \sum_{C' \in \calC_T} \frac{w(C')}{X} + \sum_{h \in H_T} \frac{w(h)}{X} \right)
	= \sum_{T \in \T_H} \frac{w(T)}{X}
	= \frac{w\left(\bigcup_{T \in \T_H} T \right)}{X}$, since the flow cost for every $h \in H_T$ is defined as $p_{\dir{ht}}(y_{\dir{ht}}) = \frac{w(h)}{X} + \frac{y_{\dir{ht}}}{X}$.
	
	Finally, $\T_C \cup T_H$ is a partition of $C\cup H$, which ensures that  $Q = \bigcup_{T\in T_H\cup T_C} T\cap Q$ for each $Q\in \mathbb{CC}(C)$. Since we route a flow of $w(T\cap Q)$ for each  $T\in\T_C \cup T_H$ and  $Q\in \mathbb{CC}(C)$,  we  send a flow of $w(Q)$ through $y_{\dir{sq}}$ for every $q \in \calQ$, satisfying the flow conservation.	As a result, we obtain a feasible flow $Y$ with 
	costs $p(Y) = \sum_{\overrightarrow{e} \in \delta^-(t)} p(y_{\overrightarrow{e}})
	= \frac{w\left(\bigcup_{T \in \T_H} T\right)}{X} + \sum_{T \in \T_C} \frac{w\left(T\right)}{w\left(Q(T)\right)}$.
	
With this, we show our intermediate goal $|\T_C| + \frac{w\left(\bigcup_{T \in \T_H} T\right)}{X} \geq p(Y^*)$ as follows.
	Since $Y$ is feasible, we obtain $p(Y) \geq p(Y^*)$.
	Therefore, it is enough to show that $|\T_C| \geq  \sum_{T \in \T_C} \frac{w\left(T\right)}{w\left(Q(T)\right)}$, which immediately follows from $\frac{w\left(T\right)}{w\left(Q(T)\right)} \leq 1$ since by definition $T\subseteq Q(T)$, for each $T\in \T_C$.
	
The intermediate goal gives us a relationship of the flow to the sets in the optimum solution involving vertices from $H$ and $C$. To combine this into a result on the whole graph, recall that $\calR$ is a $[X, 3X - 3]$-{\cvp} of $V \setminus (C \cup H)$.
	Hence, $X|\calR| \leq w(V - (C \cup H))$.
	Finally, we may show the desired contraposition of the lemma, i.e.~$X \geq X^*$ leads to $p(Y^*) + |\calR| \leq k$.
	We only need to combine all inequalities as follows.
	Let $W_H = w\left(V(\T_H)\right)$ and $W_C = w\left(V\left(\T_C\right)\right)$.
	We obtain:
	$Y^* + |\calR| 
	\leq |\T_C| + \frac{W_H}{X} + |\calR|
	= |\T_C| + \frac{W_H + X|\calR|}{X}
	\leq |\T_C| + \frac{W_H + w(V \setminus (C \cup H))}{X}
	= |\T_C| + \frac{w(V) - W_C }{X}
	\leq |\T_C| + (k - |\T_C|) = k$.
	Note that for the last inequality we use that $\T_C \subset \T^* $, which in particular means that $\T^*\setminus T_C$ is a $[1, X^*]$-{\cvp} of a graph with total weight $w(G) - W_C$ and obviously containing $k-|T_C|$ sets. This gives the trivial bound of $w(G) - W_C\leq X^*(k-|T_C|)$ and the assumption $X\geq X^*$ hence gives $w(G) - W_C\leq X(k-|T_C|)$.	

\end{proof}

It remains to show how to construct the approximate solution from $\Nw_\lambda$ in the case that $p(Y^*) + |\calR| \leq k$.

\begin{lemma} \label{lemma::fact3}
	If $|\calR| +  w\left(Y^*\right) \leq k$, then  a $(0,3X)$-{$\cvp_{k'}$} of $V(G)$ with $k' \leq k$.
can be computed in $\mathcal{O}\left(|V|\,|E| + |H|\,|V|^2\right)$.
\end{lemma}

\begin{proof}
	We prove this lemma by showing that we find a $(0,3X)$-{$\cvp_{\hk}$} of $C \cup H$ with $\hk := k - |\calR| \leq w\left(Y^*\right)$, since  $\calR$ is already a $(0,3X)$-{$\cvp_{|\calR|}$} of $V \setminus (C \cup H)$.
	
	We find this $(0,3X)$-{$\cvp_{\hk}$} of $C \cup H$ with the flow $Y^*$.
	 By Lemma~\ref{lemma::feasibleFlow}, it follows that  $y^*_{\dir{sq}}=w(Q) $ for every $q\in \calQ$, which means that the full weight of each component is routed in flow $Y^*$, and by this assigned to the $h\in H$ or to itself (via $q'$).	Essentially, we will round the assignment of components in $G[C]$ to $h\in H$ given by the flow. For this, we first show that an optimum flow does not assume too much integrality in the sense of splitting assignment for some $q\in \calQ$ among some $h\in H$ and the copy $q'$. For this, we define for every $h \in H$ the set of such undecided components $\calQ_h = \left\{q \in \calQ\mid  y^*_{\dir{qh}} > 0 \wedge y^*_{\dir{qq'}} > 0 \right\}$.
	We claim that if $|\calQ_h| \geq 2$ for some $h\in H$, then the weights $w(Q)$ are equal for all $q \in \calQ_h$.

First recall that $X > \max_{q \in \calQ} w(Q)$ and note that this means that cost $p_{\dir{q' t}}(y)$ is $y/w(Q') = y/w(Q)>\frac yX$ for every $q' \in \calQ'$, $h\in H$ and $y>0$.
	
	Suppose there exist $q_{1}, q_{2} \in \calQ_h$ w.l.o.g.~with $w(Q_{1}) > w(Q_{2})$.
	Thus, by $1/w(Q_{2}) > 1/w(Q_{1})$ we would improve the flow costs $p(Y^*)$ by altering the edge-assignments
	$y^*_{\overrightarrow{q_{2} q_{2}'}} - 1$, $y^*_{\overrightarrow{q_{2} h}} + 1$, $y^*_{\overrightarrow{q_{1}} h} - 1$ and $y^*_{\overrightarrow{q_{1} q_{1}'}} + 1$. Note that
 this yields an improvement of $1/w(Q_{2})-1/w(Q_{1})>0$.	
	Clearly, we do not violate any capacity constraint by this change, as by the definition of $\calQ_h$ we know that  $y^*_{\dir{q_ih}} > 0$ and  $y^*_{\dir{q_iq_i'}}>0$ for $i\in\{1,2\}$ which also means that the capacity on both of these edges is not fully saturated by the flow.
	Thus, this contradicts that $Y^*$ is a min-cost flow and therefore, it follows that $|\calQ_h| \geq 2$ yields that all $q \in \calQ_h$ have the same weight.
	 
 	With this condition in hand we can modify the solution $Y^*$, such that we obtain $|\calQ_h| \leq 1$ for every $h \in H$ without changing the flow costs $p(Y^*)$ as follows.
 Consider any situation where $q_{1}, q_{2} \in \calQ_h$ with $|\calQ_h| \geq 2$ for an $h \in H$.
Assume  $y^*_{\overrightarrow{q_{1} q_{1}'}} = x$ and $y^*_{\overrightarrow{q_{2} q_{2}'}} \geq x$.
 	By  adjusting 	$y^*_{\overrightarrow{Q_{1} Q_{1}'}} - x$, $y^*_{\overrightarrow{Q_{1} h}} + x$, $y^*_{\overrightarrow{Q_{2}} h} - x$ and $y^*_{\overrightarrow{Q_{2} Q_{2}'}} + x$  the flow costs $p(Y^*)$ do not change.
 	Note that we obtain $y^*_{\overrightarrow{Q_{1} Q_{1}'}} - x = 0$ and for the same reason as above, we do not violate any capacity constraint with these adjustments.
 	Thus, we reduce $|\calQ_h|$ by one without changing the flow costs $p(Y^*)$, or violating any capacity constraint.
Repeating this process, we obtain a flow for which $|\calQ_h| \leq 1$ for every $h \in H$. We keep referring to this flow by $Y^*$. This procedure can be performed in time $\mathcal{O}(|H|\,|V|^2)$ by checking every pair of $\calQ$ w.r.t.~the vertices in $H$.
 	
 	Next, we  construct  a $(X,3X)$-{$\cvp_{k'}$} of $C \cup H$ with $Y^*$.
 	For this we define a function $g\colon H \times \calQ\to \mathbb N$, initialized with $g(h,q) = 0$ for all $(h,q)\in  H \times \calQ$.
 	Recall that $\calQ = \mathbb{CC}(C)$ corresponds to the connected components in $G[C]$. 
 	Except for the vertex $q \in \calQ_h$ (if it exists) we set for each $h \in H$,
 	$g(h,q) = y^*_{\overrightarrow{h q}}$ for every $\overrightarrow{q h} \in \overrightarrow{E}$ with $y^*_{\overrightarrow{h Q}} > 0$.
 	Afterwards, for each $q\in \calQ$ such that $q\in \calQ_h$ for some $h \in H$ we set arbitrarily $g(h,q) =  y^*_{\overrightarrow{h q}} + y^*_{\overrightarrow{qq'}}$ for one of the $h \in H$ with  $q\in \calQ_h$.  

 	As a result, for every $h \in H$ we obtain $w(h) + \sum_{q \in \calQ} g(h,q) \leq 2 X - 1$,
 	since $c\left(\dir{ht}\right) = X - w(h)$, $w(Q) < X$ and $|\calQ_h| = 1$.
 	Determining the weights for $g$ can be done in $\mathcal{O}(|E|)$ time.
 	Note that $g(h,q) > 0$ implies $\dir{qh} \in \dir{E}(\Nw_{\lambda}$) for $h \in H$, $q \in \calQ$. 	
 	Let $\widehat{\calQ} = \left\{q \in \calQ\mid \exists h \in H: g(q,h) > 0 \right\}$ and $\widehat{E} = \left\{(h,q) \in \calQ \times H\mid g(q,h) > 0 \right\}$.	

We construct a bipartite graph $G_B = (H \cup \widehat{\calQ},\widehat{E})$  and observe that $(\varnothing,H,g,2X -1)$ is a \fexpanse{ }for $G_B$. Hence we can use cycle canceling like in \cref{lemma::cycleCanceling} to derive a \expanse{ }$f$ for $G_B$.
 		With this, we obtain a function $f'\colon \widehat{\calQ} \to H$ with $f'(q) \in N_B(q)$ for every $q \in \widehat{\calQ}$.
Note that $f'(q) \in N_B(q)$ implies $\dir{f'(q) q} \in \dir{E}(\Nw_{\lambda})$ which in turn means that  the component corresponding to $q$ in $G[C]$ is connected to the vertex $f'(q)$ in $H$. With this, we obtain a connected vertex set $T_h := \{h\} \cup V\left(f'^{-1}(h)\right)$ in $G[C \cup H]$ for every $h \in H$.
 	Moreover, we obtain from \cref{lemma::cycleCanceling} for every $h \in H$ that 
 	$w(T_h) = w(h) + w\left(f'^{-1}(h)\right) \leq 2X - 1 + \left(\max_{q \in \widehat{\calQ}} w(Q) - 1\right) \leq 2X- 1 + X - 2 = 3X - 3$, 	as $\max_{q \in \widehat{\calQ}} w(Q) \leq X - 1$.
 	Finally, \cref{lemma::cycleCanceling} provides that we can compute $f$ in $\mathcal{O}(|V|\,|E|)$.
 	Note that vertices $q \in \calQ$ with $\sum_{h \in (\delta^+(q) \cap H)} y^*_{\overrightarrow{q h}} > 0$ are in $\widehat{\calQ}$.
 	Let $\hk = |H| + |\calQ \setminus \widehat{\calQ}|$.
 	
Note that for every $q \in \calQ \setminus \widehat{\calQ}$ we have 
    $\sum_{h \in (\delta^+(q) \cap H)} y^*_{\overrightarrow{q h}} = 0$, we obtain a $(0,3\lambda)$-{$\cvp_{\hk}$} by     
    $\T^{C \cup H} := \{Q\in\mathbb{CC}(C) \mid q\in \calQ \setminus \widehat{\calQ}\} \cup \{T_h \mid h\in H\}$ of $C \cup H$, since $\max_{Q \in \calQ} w(Q) < X$.  Note that $V(\calQ) \cup H = C \cup H$. 
    
It remains to show that $\hk \leq p(Y^*)$.
        Since for every $q \in \calQ \setminus \widehat{\calQ}$ we have $\sum_{h \in (\delta^+(q) \cap H)} y^*_{\overrightarrow{q h}} = 0$, the corresponding flow satisfies $y^*_{\overrightarrow{qq'}} = w(Q)$.      
    Thus, each $q \in \calQ \setminus \widehat{\calQ}$ yields a cost of one for $Y^*$, as $p_{\dir{q't}}(y^*_{\overrightarrow{qq'}}) = y^*_{\overrightarrow{Q Q'}}/w\left(Q\right) = 1$.
  	Hence, we obtain 
  	$\sum_{q \in \calQ \setminus \widehat{\calQ}} p_{\overrightarrow{q't}}(y^*_{\overrightarrow{q' t}})
  	=  |\calQ \setminus \widehat{\calQ}|$.
  	
  	By Lemma \ref{lemma::SaturatedCapacities} we obtain $y^*_{\overrightarrow{ht}} = c(\overrightarrow{ht}) = X - w(h)$ for all $h \in H$ and therefore, $\sum_{h \in H} p_{\overrightarrow{ht}} \left(y^*_{\overrightarrow{ht}}\right)
  	 = \sum_{h \in H} \frac{w(h) + y^*_{\overrightarrow{ht}}}{X}
  	 = \sum_{h \in H} \frac{X}{X}  = |H|$.	At last, we obtain 
  	$\hk = |H| + |\calQ \setminus \widehat{\calQ}| 
  	= \sum_{h \in H} p_{\overrightarrow{ht}} \left(y^*_{\overrightarrow{ht}}\right) + \sum_{q \in \calQ \setminus \widehat{\calQ}} p_{\overrightarrow{q't}}(y^*_{\overrightarrow{q' t}}) = \sum_{\overrightarrow{e} \in \delta^-(t)} p_{\overrightarrow{e}}\left(y^*_{\overrightarrow{e}}\right) = p\left(Y^*\right)$.
  	
  	Lastly, the construction runs in time $\mathcal{O}\left(|H|\,|V|^2 + |E| +|V|\,|E| + |V|\right)$ and thus, in time $\mathcal{O}\left(|V|\,|E| + |H|\,|V|^2\right)$.
\end{proof}

\paragraph{Running time:}
Similar to the Max-Min case we modify slightly the binary search to optimize the running time. 
Let $g(\ell) := 2^{\ell}$, $\mathbb{N} \ni \ell \geq 1$.
We increase stepwise $\ell$ in $g(\ell)$ until we find an $\ell^*$ with $g(\ell^*) \leq X^* \leq g(\ell^* + 1) \leq 2X^*$.
Thus, we obtain a running time of $\bigO(\log X^*)$ for the binary search and may bound $X$ by $2X^*$.

For computing a min-cost flow in $\Nw_{\lambda}$ we have already indicated that the network was  designed with the rational cost functions only for the sake of easier correctness proofs (from the 1-1-correspondence of flow cost and number of sets in the partition). In order to find a min-cost flow $Y^*$ in $\Nw_{\lambda}$, we can equivalently search for a min-cost flow in the network $\Nw_{\lambda}'$ derived from $\Nw_{\lambda}$ by choosing a different cost function $p'$ defined by:
\begin{itemize}
\item $p'_{\dir{q't}}(y)=(\lambda+1-w(Q))y$ for all $q\in \calQ'$,
\item $p'_{\dir{ht}}(y)=y$ for all $h\in H$, and
\item all-zero assignment on the remaining arcs.
\end{itemize}
 Note that this alteration of the cost function yields that a flow $Y$ is optimal for $\Nw_{\lambda}'$ if and only if it is optimal for $\Nw_{\lambda}$. First note that the capacity function of the two networks is the same, so especially any flow $Y$ is valid in $\Nw_{\lambda}'$  if and only if it is valid in $\Nw_{\lambda}$. Further, for any such flow $Y$, the residual network of   $\Nw_{\lambda}'$ and $\Nw_{\lambda}$ contains the same arcs. A valid flow (of the desired value $F$) in $Y$ is a min-cost flow, if there are no cost-reducing cycles in the residual network. Since, by our construction,  $\Nw_{\lambda}$ is bipartite, the existence of any cost-reducing cycle yields the existence of a cost-reducing simple cycle (i.e.~ a cycle without repetition of vertices). Since the only arcs in  $\Nw_{\lambda}$, and also in $\Nw_{\lambda}'$, with non-zero cost are adjacent to $t$, any cost-reducing cycle contains at most two arcs with non-zero cost. By our definition of $p'$ we see that for any pair of arcs $\dir{e_1},\dir{e_2}\in \dir{E}$ we have $p_{\dir(e_1)}<p_{\dir(e_2)}$ if and only if $p'_{\dir(e_1)}<p'_{\dir(e_2)}$, which shows that any cycle in the residual network is cost-reducing in  $\Nw_{\lambda}'$  if and only if it is cost-reducing in $\Nw_{\lambda}$.

By~\cite{ahuja1992finding}, a min-cost flow can be computed in $\bigO\left(|V|\,|E|\log \log U \log\left(|V| P \right)\right)$, where $U \in \mathbb{N}$ is the maximum capacity and $P \in \mathbb{N}$ is the maximum flow cost. For this altered network  $\Nw_{\lambda}'$, we can thus compute $Y^*$ in $\bigO\left(|V|\,|E|\log \log X \log\left(|V| \wmax \right)\right)$; note that in  $\Nw_{\lambda}'$ the cost of  $Y^*$  is, very roughly estimated, at most $|V|^3\wmax$, where $\wmax$ denotes the maximum weight $\wmax=\max_{v\in V} w(v)$.

The algorithm {\algCrown} that provides a {\ccd} ($\chrf$) of $G$  and runs in $\bigO(\widetilde{k'}^2|V|\,|E|)$ (see Theorem \ref{theorem:lccd}).
Since we have at most $k$ iterations in {\algCrown}, (i.e.~the outer-index is bounded by $k + 1$), we obtain $\bigO(\widetilde{k}^2|V|\,|E|) \subseteq \bigO(k^2|V|\,|E|)$. Moreover, for the step where we actually construct a solution, this construction requires time $\bigO\left(|V|\,|E| + |H|\,|V|^2 \right)$, where we know that $|H| \leq k$. With the adjusted binary search, that never checks a value larger than $2X^*$, this yields an overall running time in $\mathcal{O}\left(\log\left(X^* \right) |V|\,|E| \left(\log \log X^* \log\left(|V| \wmax \right) + k^2 \right) \right)$ for the algorithm {\algminmax}.

This completes the proof of the following theorem.

\BCPMinMaxApxThm*

	

	\appendix

\end{document}